\pdfoutput=1
\documentclass[10pt,aps,pra,twocolumn,nofootinbib,superscriptaddress,floatfix]{revtex4-2}

\usepackage{amsmath, amssymb}
\usepackage{xcolor}
\usepackage{url}
\usepackage{dirtytalk}
\usepackage{graphicx} 
\usepackage{subcaption}
\usepackage{wrapfig} 
\usepackage[T1]{fontenc}
\usepackage{tcolorbox}
\usepackage{mathtools}
\usepackage{svg}
\usepackage{bm}
\usepackage{multirow}
\usepackage[ruled,lined]{algorithm2e}
\DontPrintSemicolon                       
\usepackage{xfrac}
\usepackage{qcircuit}
\usepackage{braket}
\usepackage{amsthm}
\usepackage{hyperref}
\usepackage{cleveref}
\allowdisplaybreaks

\hypersetup{
    colorlinks=true,
    linkcolor=blue,
    citecolor=magenta,
    urlcolor=magenta
}
\DeclareMathOperator{\Tr}{Tr}

\newtheorem{theorem}{Theorem}

\newtheorem{definition}{Definition}
\newtheorem{assumption}{Assumption}

\newtheorem{lemma}{Lemma}

\newtheorem{remark}{Remark}

\newtheorem{proposition}{Proposition}

\begin{document}

\title{Bose--Einstein thermal operators for
semidefinite optimization}

\author{Michele Minervini}
\affiliation{School of Electrical and Computer Engineering, Cornell University, Ithaca, New York 14850, USA}

\author{Nana Liu}
\affiliation{Institute of Natural Sciences, Shanghai Jiao Tong University, Shanghai 200240, China}
\affiliation{School of Mathematical Sciences, Shanghai Jiao Tong University, Shanghai 200240, China}
\affiliation{Ministry of Education Key Laboratory in Scientific and Engineering Computing,
Shanghai Jiao Tong University, Shanghai 200240, China}
\affiliation{Global College, Shanghai Jiao Tong University, Shanghai 200240, China}

\author{Mark M. Wilde}
\affiliation{School of Electrical and Computer Engineering, Cornell University, Ithaca, New York 14850, USA}

\date{\today}

\begin{abstract}
We establish that semidefinite programs (SDPs) over the unbounded positive semidefinite cone are mathematically equivalent to thermodynamic systems of independent bosonic modes: the eigenvalues of the optimization variable play the role of expected occupation numbers, the linear objective plays the role of total expected energy, and the linear equality constraints play the role of conserved non-commuting charges. Building on this perspective, we recast general SDPs as bosonic free-energy minimization problems at strictly positive temperature, regularized by the Bose--Einstein entropy; the original SDP is recovered in the zero-temperature limit. The optimal primal variable takes the form of a Bose--Einstein thermal operator parametrized by the dual variables. We prove an approximation-error bound that depends on the ground-space degeneracy and the spectral gap of the dual slack operator, improving on the linear-in-dimension worst-case duality gap of interior-point methods. We also introduce the Bose--Einstein quantum relative entropy as a Bregman divergence on the unbounded positive semidefinite cone, generated by the negative Bose--Einstein entropy. We propose it as a natural divergence for unnormalized positive operators, for which the standard Umegaki relative entropy can become negative, and we show that it satisfies a restricted monotonicity property under affine maps modeling bosonic Gaussian channels. Finally, we develop hybrid quantum--classical algorithms for the regularized SDP using only Hamiltonian simulation, Hadamard tests, and classical sampling, and bound their runtime in closed form. Unlike existing quantum SDP solvers, whose runtimes scale polynomially with an a priori upper bound on the primal trace, our framework operates directly on the unbounded cone, replacing this bound with a dependence on the spectral structure of the dual slack operator.
\end{abstract}

\maketitle

\tableofcontents

\section{Introduction}

\label{sec:introduction}

\subsection{Background and motivation}
\label{subsec:background-and-motivation}

A semidefinite program (SDP) is a linear optimization problem over the cone of positive semidefinite operators. In its standard form, one minimizes a linear objective $\Tr[HX]$ over a $d \times d$ positive semidefinite operator $X \geq 0$ subject to linear equality constraints $\Tr[Q_i X] = q_i$~\cite{Boyd2004convex,vandenberghe1996semidefinite}. Solving SDPs efficiently is a cornerstone of modern convex optimization. Interestingly, this framework also lies at the foundation of quantum information theory~\cite{Mironowicz2024SDP}. Indeed, the core objects of quantum mechanics -- quantum states, measurements, and channels -- are all defined by positive semidefinite constraints. Many fundamental tasks -- such as state discrimination~\cite{Helstrom1969,Holevo1972}, channel-capacity bounds~\cite{Wang2018semidefinite}, entanglement quantification~\cite{Rains2001SDP}, and quantum hypothesis testing~\cite{Wang2019Resource} -- can therefore be cast as SDPs. 

The classical literature on SDPs is rich: interior-point methods (IPMs)~\cite{Karmarkar1984,NesterovNemirovski1994,Alizadeh1995IPMSDP} solve general SDPs to high precision in time polynomial in the dimension $d$ of the optimization variable, and first-order methods such as the matrix multiplicative weights update method achieve faster scaling for large dimensions, though their runtime acquires polynomial dependence on an a priori upper bound on the trace of the primal variable. Mature open-source and commercial solvers -- such as SDPA~\cite{Yamashita2010SDPA}, MOSEK~\cite{MOSEK}, CVX~\cite{CVX2014}, and CVXPY~\cite{Diamond2016CVXPY} -- implement these methods and have made SDPs a standard tool in convex-optimization practice. This polynomial-in-$d$ scaling, however, becomes prohibitive for SDPs arising in quantum information, where the dimension typically grows exponentially with the underlying system size (e.g., $d = 2^n$ for an $n$-qubit system). This motivates the search for quantum SDP solvers: a substantial line of work has established polynomial quantum speedups~\cite{BrandaoSvore2017,vanApeldoorn2017quantumSDP,Brandao2019QuantumSDP,vanApeldoorn2019SDP,KerenidisPrakash2020QIPM,AugustinoNannicini2023QIPM}, ranging from quantum implementations of the matrix multiplicative weights method to quantum interior-point methods.

A complementary line of work, recently initiated in~\cite{liu2025sdp,liu2026sdp_fermi} and continued in the present paper, takes a thermodynamic perspective directly inspired by Jaynes' approach to statistical mechanics~\cite{Jaynes1957}. Rather than treating an SDP as an abstract convex program, one interprets the optimization variable as the state of a thermodynamic system, identifies a physically natural entropy associated with that system, and replaces the linear SDP objective with the corresponding free energy at a strictly positive temperature $T>0$. The resulting unconstrained dual is concave, has a closed-form gradient and Hessian expressed as thermal expectation values, and can be solved by gradient ascent or Newton's method using either classical or hybrid quantum--classical primitives.

This program has produced two prior frameworks. In~\cite{liu2025sdp}, the variable $X$ is constrained to the set of density operators ($X=\rho\geq 0$, $\Tr[\rho]=1$), the corresponding entropy is the von Neumann entropy, and the regularized optimizer is a quantum Boltzmann state. In~\cite{liu2026sdp_fermi}, the variable is constrained to the set of measurement operators ($0\leq  M\leq  I$), the corresponding entropy is the Fermi--Dirac entropy, and the regularized optimizer is a Fermi--Dirac thermal measurement $M_T = (e^{(H-\mu\cdot Q)/T}+I)^{-1}$. Both frameworks rest on the observation that the constraint on the eigenvalues of the optimization variable matches a physical statistical-mechanics constraint: trace-one normalization for Boltzmann statistics and the Pauli exclusion principle for Fermi--Dirac statistics.

The standard form of an SDP, however, sits in neither of these settings: the constraint $X\geq 0$ in~\eqref{eq:general_SDP} imposes no upper bound on the eigenvalues of $X$, and no normalization constraint. The natural physical analog is therefore a system of non-interacting bosonic modes, whose occupation numbers are non-negative real numbers without an upper bound. This is the setting that we develop in the present paper.

\subsection{Approach and technical overview}
\label{subsec:approach}

Similar to~\cite{liu2025sdp,liu2026sdp_fermi}, we interpret each eigenmode of the positive semidefinite operator $X$ as an independent bosonic mode, with eigenvalue $x_j\geq 0$ playing the role of the expected occupation number of bosons in the $j$th mode. Under this interpretation, the linear objective $\Tr[HX]=\sum_j x_j\langle\phi_j|H|\phi_j\rangle$ is the total expected energy of the bosonic system, and each linear constraint $\Tr[Q_iX]=q_i$ is the conservation law for the non-commuting charge $Q_i$. The freedom of the optimization is to adjust both the modes $|\phi_j\rangle$ and their occupation numbers $x_j$ to minimize the total expected energy subject to the conservation constraints.

Inspired by the fact that physical systems operate at a strictly positive temperature, we regularize the optimization by minimizing the following bosonic free energy:
\begin{equation}
    F[X] \coloneqq \Tr[HX] - T\cdot S_{\mathrm{BE}}(X)
\end{equation}
at temperature $T>0$, where the Bose--Einstein entropy
\begin{equation}
    S_{\mathrm{BE}}(X) \coloneqq \Tr[(X+I)\ln(X+I) - X\ln X]
\end{equation}
is the von Neumann entropy of a thermal ensemble of independent bosonic modes with the same expected occupation as $X$. By Lagrangian duality, this constrained free-energy minimization problem becomes an unconstrained concave maximization problem in the dual chemical-potential vector $\mu\in\mathbb{R}^c$, with the optimal primal variable taking the form of a Bose--Einstein thermal operator $X_{T}(\mu)=(e^{K_\mu/T}-I)^{-1}$, where $K_\mu = H-\mu\cdot Q$ plays the role of a grand canonical Hamiltonian. As $T\to 0^+$, the optimal value of the regularized problem converges to the optimal value of the original SDP.

A central new ingredient introduced in this paper is the \emph{Bose--Einstein quantum relative entropy} $D_{\mathrm{BE}}(X\|Y)$, defined as the matrix Bregman divergence generated by the negative Bose--Einstein entropy. This object plays a triple role: it provides a concise alternative proof of the dual representation, it gives an information-geometric interpretation of the regularized optimization as Quantum Mirror Descent on the unbounded positive semidefinite cone, and it is the natural divergence to use in quantum information processing tasks involving unnormalized positive operators (e.g., bosonic mode occupations). 

The dual objective and its analytical gradient and Hessian (both expressed as thermal expectation values) admit hybrid quantum--classical estimators, which are based solely on Hamiltonian simulation, Hadamard tests, and classical Cauchy sampling. These algorithms serve as the core quantum primitives driving hybrid quantum--classical first- and second-order optimization methods to solve the regularized dual problem. To ensure this regularized solution accurately recovers the optimal value of the original SDP, the temperature parameter $T$ is chosen sufficiently small in order to guarantee an overall target precision $\varepsilon$. Furthermore, we establish an approximation-error bound that, when the dual slack operator $K_\mu = H-\mu\cdot Q$ has a small ground-space degeneracy and a non-vanishing spectral gap, depends only logarithmically on the ambient Hilbert-space dimension $d$, thus refining the worst-case $\mathcal{O}(T\cdot d)$ duality-gap bound inherent to classical IPMs.

\subsection{Main contributions and paper organization}
\label{subsec:summary-of-contributions}

Following a brief review of notation in Section~\ref{sec:notation}, the remainder of our paper formalizes this thermodynamic framework for general SDPs. Our main contributions are organized as follows:

\begin{itemize}
    \item \textbf{Bosonic free-energy formulation (Section~\ref{sec:general_sdp}):} We observe that the constraint $X\geq 0$ in a standard SDP coincides with the physical occupation rules of independent bosonic modes (Section~\ref{subsec:Interpretation-bosonic}). We recast the standard SDP as a thermodynamic free-energy minimization problem, deriving an unconstrained concave dual problem in the chemical-potential vector $\mu\in\mathbb{R}^c$ and showing that the optimal primal variable is a Bose--Einstein thermal operator (Sections~\ref{sec:free_energy_min} and~\ref{subsec:Bose--Einstein-thermal-operators}).
    
    \item \textbf{Thermodynamic approximation bounds and analytical derivatives (Section~\ref{sec:general_sdp}):} We establish two complementary approximation-error bounds between the original SDP and its finite-temperature Bose--Einstein formulation, including a bound that is independent of the global Hilbert-space dimension $d$ in the gapped regime (Section~\ref{subsec:Spectral-gap-based-approximation}). We also derive analytic expressions for the gradient and Hessian of the dual objective function, proving that the Hessian is negative semidefinite with bounded spectral norm (Section~\ref{subsec:Gradient-and-Hessian}). 
    
    \item \textbf{Quantum algorithms for bosonic thermal expectation values (Section~\ref{sec:quantum_algorithms}):} We propose hybrid quantum--classical algorithms for estimating the gradient and the Hessian matrix elements of the dual objective. These algorithms are based solely on Hamiltonian simulation, the Hadamard test, and classical sampling, and we bound the per-call quantum runtime in closed form.
    
    \item \textbf{End-to-end complexity and comparisons (Sections~\ref{sec:overall-complexity} and~\ref{sec:comparison}):} We construct stochastic gradient ascent and stochastic Newton schemes for the dual problem, with an end-to-end runtime guarantee for the first-order method and a per-iteration cost analysis for the second-order method (Section~\ref{sec:overall-complexity}). We then position the framework against classical interior-point methods and classical and quantum matrix multiplicative weights (MMW) SDP solvers, highlighting both its qualitative advantages and quantitative limitations (Section~\ref{sec:comparison}).
    
    \item \textbf{Bose--Einstein relative entropy (Section~\ref{sec:be-relative-entropy}):} We introduce $D_{\mathrm{BE}}(X\|Y)$ as a stand-alone information-theoretic primitive. We establish its faithfulness, unitary invariance, spectral expansion, and additivity under direct sums. We show that, although it does not satisfy the generalized data-processing inequality, it does satisfy a restricted monotonicity under the class of affine maps $Z \mapsto aZ + bI$ such that $2b+1 \geq a$, related to the attenuator, amplifier, and additive-noise bosonic Gaussian channels (Section~\ref{subsec:BE-affine-monotonicity}). We also identify the associated Bose--Einstein Fisher information matrix, establishing an information-geometric link between the dual Hessian and the curvature of the divergence.
\end{itemize}

Finally, in Section~\ref{sec:conclusion} we conclude with a discussion of open problems and future directions. All proofs of the main results are deferred to the appendices.

\section{Notation and preliminaries}
\label{sec:notation}

Throughout this paper, we employ the following notation and conventions. For every positive integer $c \in \mathbb{N}$, we use the shorthand $[c] \coloneqq \{1, \dots, c\}$. 

Let $\mathrm{Herm}_d$ denote the space of $d \times d$ Hermitian matrices. For every operator $X \in \mathrm{Herm}_d$, we write $X \geq 0$ to indicate that $X$ is positive semidefinite (having all non-negative eigenvalues), and $X > 0$ to indicate that $X$ is strictly positive definite (having all strictly positive eigenvalues). The set of $d \times d$ quantum density operators is denoted by $\mathcal{D}_d \coloneqq \{ \rho \in \mathrm{Herm}_d : \rho \geq 0, \Tr[\rho] = 1 \}$.

For a vector $\mu = (\mu_1, \dots, \mu_c) \in \mathbb{R}^c$ and a vector of scalars, $q = (q_1, \dots, q_c) \in \mathbb{R}^c$, the standard inner product is denoted by $\mu \cdot q \equiv \sum_{i=1}^c \mu_i q_i$. We extend this notation to vectors of operators: if $Q = (Q_1, \dots, Q_c)$ is a tuple of Hermitian operators, we define the scalar-operator inner product as $\mu \cdot Q \equiv \sum_{i=1}^c \mu_i Q_i$.

When evaluating operator sizes for a matrix $A$, we denote its spectral norm (maximum singular value) by $\|A\|$, and its trace norm by $\|A\|_1 \coloneqq \Tr[\sqrt{A^\dagger A}]$. In the quantum state model used in Section~\ref{sec:quantum_algorithms}, where an operator is decomposed as a linear combination of states $Q_i = \sum_k \alpha_{i,k} \rho_k$, we denote the $\ell_1$-norm of its classical coefficient vector by $\left\|\alpha_i\right\|_1 \coloneqq \sum_k |\alpha_{i,k}|$.

Finally, to analyze the runtime of our quantum algorithms, we use the standard Big-$\mathcal{O}$ notation to denote asymptotic upper bounds. Furthermore, we employ the Soft-$\mathcal{O}$ notation, denoted $\widetilde{\mathcal{O}}(\cdot)$, to suppress polylogarithmic factors in the physical parameters; that is, $f(x) = \widetilde{\mathcal{O}}(g(x))$ implies $f(x) = \mathcal{O}(g(x) \mathrm{polylog}(g(x)))$.

\section{General semidefinite optimization with Bose--Einstein operators}
\label{sec:general_sdp}

Let us begin by considering a general semidefinite optimization problem, in which we optimize a linear objective function over the intersection of the cone of positive semidefinite operators with an affine space. Because standard semidefinite programs (SDPs) are defined by the constraint $X \geq 0$ without an upper spectral bound, we demonstrate in Section~\ref{sec:free_energy_min} how this problem is naturally regularized by the Bose--Einstein entropy. As discussed in Section~\ref{subsec:Interpretation-bosonic}, this framework allows us to interpret the optimization variable $X$ as the expected occupation numbers of independent bosonic modes. Consequently, optimization problems in quantum information and many-body physics that involve unbounded positive operators can be solved natively within this thermodynamic paradigm.

\subsection{General semidefinite optimization problem}
\label{subsec:General-SDP-optimization}

Let $c, d \in \mathbb{N}$, let
\begin{equation}\label{eq:Q}
\mathcal{Q} \equiv (H, Q_1, \dots, Q_c)
\end{equation}
be a tuple of $d \times d$ Hermitian matrices, and let
\begin{equation}\label{eq:q}
q \equiv (q_1, \dots, q_c) \in \mathbb{R}^c.
\end{equation}
We can think of $H$ as a Hamiltonian and each Hermitian operator $Q_i$ as a non-commuting charge, such as electric charge, particle number, angular momentum, etc.

\begin{definition}
A semidefinite program in standard form is defined as follows:
\begin{equation}\label{eq:general_SDP}
    E(\mathcal{Q},q) \coloneqq \min_{X \geq 0} \left\{\Tr[HX] : \Tr[Q_i X] = q_i \quad \forall i \in [c]\right\},
\end{equation}
where the optimization is over the cone of $d \times d$ positive semidefinite operators $X \geq 0$.
\end{definition}

\begin{assumption}\label{ass:Slater}
We assume that the problem is strictly feasible (Slater's condition), meaning there exists an $X > 0$ such that $\Tr[Q_i X] = q_i$ for all $i \in [c]$~\cite[Section 5.2.3]{Boyd2004convex}. This ensures that strong duality holds and the dual optimum is attained.
\end{assumption}

\begin{proposition}\label{prop:dual_general_sdp}
If Assumption~\ref{ass:Slater} holds, then the optimal
value $E(\mathcal{Q}, q)$ in~\eqref{eq:general_SDP} can be expressed in terms of the following dual optimization problem: 
\begin{equation}\label{eq:dual_general_sdp}
    E(\mathcal{Q}, q) = \sup_{\mu \in \mathbb{R}^c} \left\{ f(\mu) : K_\mu \geq 0\right\} , 
\end{equation}
where the dual objective function $f(\mu)$ is defined as
\begin{equation}\label{eq:general_dual_obj_function}
    f(\mu) \coloneqq \mu \cdot q,
\end{equation}
and we have introduced the dual slack operator $K_\mu$, which takes the mathematical form of a grand canonical Hamiltonian in statistical mechanics~\cite{Jaynes1957, Pathria2011Statistical}:
\begin{equation}\label{eq:Kmu}
    K_\mu \coloneqq H - \mu \cdot Q,
\end{equation}
for which the linear matrix inequality $ K_\mu \geq 0$ must be satisfied. Furthermore, because the objective function $f(\mu)$ in~\eqref{eq:general_dual_obj_function} is linear in $\mu$ and the feasible region defined by $K_\mu \geq 0$ forms a convex set, the dual optimization problem constitutes a concave maximization problem.
\end{proposition}

\begin{proof}
See Appendix~\ref{app:Proof_dual_gen_sdp} for a proof of the dual form in~\eqref{eq:dual_general_sdp} and the derivation of the dual feasibility constraints.
\end{proof}

Assuming strong duality holds (Assumption~\ref{ass:Slater}), the optimal primal solution $X^\star$ and the optimal dual slack operator $K_{\mu^\star} = H - \mu^\star \cdot Q \geq 0$ satisfy the complementary slackness condition~\cite[Section 5.5.2]{Boyd2004convex}, which states that the duality gap vanishes at optimality, yielding
\begin{equation}
    \Tr[K_{\mu^\star} X^\star] = 0.
\end{equation}
Because both $X^\star$ and $K_{\mu^\star}$ are positive semidefinite matrices, the trace of their product can be zero if and only if their standard matrix product is exactly the zero matrix:
\begin{equation}\label{eq:complementary_slackness_2}
    K_{\mu^\star} X^\star = 0.
\end{equation}
Since their product is zero, the two operators commute and can be simultaneously diagonalized. Let $\{|\phi_j\rangle\}_j$ be a common eigenbasis, where $\lambda_j \geq 0$ are the eigenvalues of the dual slack operator $K_{\mu^\star}$ (representing the energy cost of each mode), and $x_j \geq 0$ are the eigenvalues of $X^\star$ (representing the expected occupation numbers). The vanishing matrix product in~\eqref{eq:complementary_slackness_2} requires that
\begin{equation}\label{eq:orthogonality_slackness}
    \lambda_j \cdot x_j = 0
\end{equation}
for every mode $j$.
Physically, this imposes a strict selection rule: if a mode has a strictly positive energy cost ($\lambda_j > 0$), it must remain completely unoccupied ($x_j = 0$). An occupation number $x_j$ is permitted to be non-zero only if its corresponding energy cost is exactly zero ($\lambda_j = 0$). 

Consequently, the optimal operator $X^\star$ must be supported entirely within the ground-state subspace (the zero-eigenspace) of the Hamiltonian $K_{\mu^\star}$, and $X^\star$ takes the restricted diagonalized form:
\begin{equation}
    X^\star = \sum_{j: \lambda_j = 0} x_j |\phi_j\rangle\!\langle\phi_j|. \label{eq:optimal-X-structure}
\end{equation}

Similar to the Fermi--Dirac case~\cite{liu2026sdp_fermi}, this unregularized solution represents a sharp, non-smooth assignment: the occupation number $x_j$ can be non-zero only for modes with zero energy cost, while all modes with strictly positive energy must have zero occupation. In the following subsections, we introduce Bose--Einstein thermal operators, which smooth this sharp transition into a continuous distribution, analogous to how finite-temperature bosonic statistics smooth out ground-state condensation.

\subsection{Interpreting eigenmodes of positive semidefinite operators as independent bosons}
\label{subsec:Interpretation-bosonic}

A distinctive feature of the standard SDP in~\eqref{eq:general_SDP} is that the optimization variable $X\geq 0$ has eigenvalues that are unbounded from above and need not normalize to one. This is in contrast to the trace-normalized density operators handled by the Boltzmann framework~\cite{liu2025sdp} ($\Tr[\rho]=1$) and to the operator-bounded measurement operators of the Fermi--Dirac framework~\cite{liu2026sdp_fermi} ($0\leq M\leq I$). The unbounded nature of the spectrum of $X$ has a natural physical interpretation: it matches precisely the spectrum of expected occupation numbers of independent bosonic modes, for which no exclusion principle limits the number of particles that can share the same mode. This bosonic interpretation is what motivates the entire framework developed below. In particular, primal-feasible operators with arbitrarily large trace arise naturally in SDPs over the unbounded cone; physically, they correspond to the onset of Bose--Einstein condensation~\cite{Pathria2011Statistical,Anderson1995Observation,Davis1995BoseEinstein}. Our formalism accommodates such operators directly, without rescaling or imposing an a priori upper bound on $\Tr[X]$.

Consider a particular positive semidefinite operator $X$ with a spectral decomposition as follows:
\begin{equation}
X=\sum_{j=1}^{d}x_{j}|\phi_{j}\rangle\!\langle\phi_{j}|.
\end{equation}
Let us refer to the eigenvalue-eigenprojector pair $\left(x_{j},|\phi_{j}\rangle\!\langle\phi_{j}|\right)$ as an eigenmode. Due to the assumption that $X$ is a positive semidefinite operator, the constraint $x_{j}\geq0$ holds for all $j\in\left[d\right]$, with no upper bound on the value of $x_j$. As such, we can think of each of the $d$ eigenmodes as an independent bosonic mode, with $x_{j}\geq0$ corresponding to an expected occupation number of bosons occupying the $j$th eigenmode.

Consider furthermore that
\begin{align}
\Tr\!\left[HX\right] & =\sum_{j=1}^{d}x_{j}\langle\phi_{j}|H|\phi_{j}\rangle.\label{eq:total-expected-energy-bosons}
\end{align}
Here, we can think of $\langle\phi_{j}|H|\phi_{j}\rangle$ as the energy of a single boson when occupying the $j$th eigenmode, so that the expected energy of the $j$th bosonic mode is $x_{j}\langle\phi_{j}|H|\phi_{j}\rangle$. Thus, the total expected energy of all $d$ bosonic modes is $\Tr\!\left[HX\right]$. Furthermore, for all $i\in\left[c\right]$, we similarly have that
\begin{equation}
\Tr\!\left[Q_{i}X\right]=\sum_{j=1}^{d}x_{j}\langle\phi_{j}|Q_{i}|\phi_{j}\rangle.
\end{equation}
Here, each $\langle\phi_{j}|Q_{i}|\phi_{j}\rangle$ is the value of the $i$th charge for a single boson in the $j$th mode, so that the expected value of the $i$th charge for the $j$th mode is $x_{j}\langle\phi_{j}|Q_{i}|\phi_{j}\rangle$, and the total expected charge of the entire bosonic system is $\Tr\!\left[Q_{i}X\right]$. 

As such, with this perspective, the goal of the standard semidefinite optimization problem in~\eqref{eq:general_SDP} is to minimize the total expected energy of $d$ independent bosonic modes subject to a constraint on the total expected value of the $i$th charge, for each $i\in\left[c\right]$. The freedom in the optimization problem is to adjust the eigenmodes -- both the states $|\phi_j\rangle$ and their occupation numbers $x_j$ -- as desired to minimize the total expected energy, while satisfying the constraints dictated by the charges $Q_i$.

\subsection{Free energy minimization and Bose--Einstein operators}
\label{sec:free_energy_min}

Following the approach from~\cite{liu2025sdp,liu2026sdp_fermi}, we are motivated by the fact that real physical systems operate at a strictly positive temperature $T>0$, in which case the goal shifts to minimizing the constrained free energy rather than the constrained energy. Recall that the free energy is equal to the expected value of the energy minus the entropy scaled by the temperature. For the $j$th eigenmode, the free energy is thus
\begin{equation}
x_{j}\langle\phi_{j}|H|\phi_{j}\rangle-T\cdot g(x_{j}),
\end{equation}
where the scalar bosonic entropy $g(x)$ is defined for expected occupation numbers $x\geq0$ as
\begin{equation}
g(x)\coloneqq(x+1)\ln(x+1)-x\ln x.\label{eq:bosonic-entropy-def}
\end{equation}
This functional form of the bosonic entropy arises from the maximum entropy principle~\cite{Jaynes1957, Wehrl1978entropy}. Specifically, if one maximizes the Shannon entropy $-\sum_{n=0}^{\infty} p(n) \ln p(n)$ over all discrete probability distributions for bosonic occupation states ($n \in \{0, 1, 2, \dots\}$) subject to the constraint of a fixed mean occupation $\sum_{n=0}^{\infty} n p(n) = x$, the optimal distribution is geometric. Evaluating the Shannon entropy of this optimal geometric distribution yields the entropy $g(x)$ as defined in~\eqref{eq:bosonic-entropy-def}. See Appendix~\ref{app:be_entropy_thermal_state} for the full derivation.

The total free energy of the system is thus given by
\begin{equation}
\sum_{j=1}^{d}\left[x_{j}\langle\phi_{j}|H|\phi_{j}\rangle-Tg(x_{j})\right]=\Tr\!\left[HX\right]-T\cdot S_{\mathrm{BE}}(X),\label{eq:rewrite-free-energy-bosons}
\end{equation}
where the Bose--Einstein entropy of the positive semidefinite operator $X$ is defined as
\begin{equation}
S_{\mathrm{BE}}(X)\coloneqq\Tr\!\left[(X+I)\ln(X+I)-X\ln X\right].\label{eq:BE-entropy-def}
\end{equation}
This extends the scalar bosonic entropy function $g(x)$ defined in~\eqref{eq:bosonic-entropy-def} to positive semidefinite operators, and coincides with the von Neumann entropy of a non-interacting bosonic thermal state whose mean occupation numbers, in the eigenmodes of $X$, are the eigenvalues of $X$ (see Appendix~\ref{app:be_entropy_thermal_state} for an explicit derivation).
The equality in~\eqref{eq:rewrite-free-energy-bosons} follows from \eqref{eq:total-expected-energy-bosons} and the fact that $S_{\mathrm{BE}}(X)$ is unitarily invariant, depending only on the eigenvalues of $X$, so that
\begin{equation}
S_{\mathrm{BE}}(X)=\sum_{j=1}^{d}g(x_{j}).\label{eq:BE-entropy-sum-bosonic-entropies}
\end{equation}
Observe that the function $S_{\mathrm{BE}}(X)$ is strictly concave in $X$, stemming from the strict concavity of $g(x)$ for $x>0$. A detailed proof of the concavity can be found in Appendix~\ref{app:concavity_BEentropy}.

We can then modify the optimization
problem in~\eqref{eq:general_SDP} to become the following
free-energy optimization problem:
\begin{definition}
For $T > 0$, and $\mathcal{Q}$ and $q$ as given in~\eqref{eq:Q} and~\eqref{eq:q}, respectively, the operator free-energy optimization problem is as follows:
\begin{equation}\label{eq:sdp_primal_regularized}
F_T(\mathcal{Q}, q) \coloneqq  \min_{X \geq 0} \left\{
\begin{array}{c}
 \Tr[HX] - T \cdot S_{\mathrm{BE}}(X) : \\
 \Tr[Q_i X] = q_i \ \forall i \in [c]
\end{array}
\right\}.
\end{equation}
\end{definition}

By solving for the dual in this case, we arrive at the following:
\begin{theorem}\label{prop:dual_sdp_regularized}
Fix $T > 0$. Given $\mathcal{Q}$ and $q$ as given in~\eqref{eq:Q} and~\eqref{eq:q}, respectively, and under Assumption~\ref{ass:Slater}, the following equality holds:
\begin{equation}\label{eq:dual_def}
F_T(\mathcal{Q}, q) = \sup_{\mu \in \mathbb{R}^c} f_T(\mu),
\end{equation}
where the temperature-$T$ dual objective function $f_T(\mu)$ is defined as
\begin{equation}\label{eq:dual_obj_function}
f_T(\mu) \coloneqq \mu \cdot q + T \Tr\!\left[ \ln \!\left( I - e^{-\frac{1}{T}(H - \mu \cdot Q)} \right) \right].
\end{equation}
Furthermore, an optimal operator $X_T(\mu)$ for~\eqref{eq:sdp_primal_regularized} is a Bose--Einstein thermal operator of the following form:
\begin{equation}\label{eq:BE_operator}
X_T(\mu) \coloneqq \left( e^{\frac{1}{T}(H - \mu \cdot Q)} - I \right)^{-1}.
\end{equation}
The unique optimal operator for the primal regularized problem in~\eqref{eq:sdp_primal_regularized} is exactly $X_T(\mu_T^\star)$, where $\mu_T^\star$ is the optimal dual vector that achieves the supremum in~\eqref{eq:dual_def}.
\end{theorem}

\begin{proof}
The full derivation is provided in Appendix~\ref{sec:Proof-of-dual-be-reg}, where we present two independent proofs. The first (Appendix~\ref{app:proof_lemma_first_order}) relies on the first-order optimality conditions of a strictly concave objective function. The second (Appendix~\ref{app:proof_lemma_rel_entropy}) leverages the introduction of the Bose--Einstein relative entropy, whose definition and properties are discussed in Section~\ref{sec:be-relative-entropy}: by rewriting the objective function in~\eqref{eq:sdp_primal_regularized} in terms of this relative entropy, the optimal primal variable $X_T(\mu)$ emerges as the unique operator that drives the relative entropy to zero.
\end{proof}

\begin{remark}[Zero-temperature limit of the dual objective $f_T(\mu)$]
    \label{rem:zero_T_limit}
As a consequence of the scalar limit $\lim_{T\to0^{+}}T\ln(1-e^{-x/T})=0$ holding for all $x > 0$, we conclude that for every strictly feasible dual point (i.e., $\mu\in\mathbb{R}^{c}$ such that $H-\mu\cdot Q>0$), the temperature-$T$ dual objective function $f_{T}(\mu)$ in~\eqref{eq:dual_obj_function} converges exactly to the unregularized dual objective function $f(\mu)$ in~\eqref{eq:general_dual_obj_function} in the zero-temperature limit:
\begin{equation}
\lim_{T\to0^{+}}f_{T}(\mu)=f(\mu).\label{eq:temp-T-obj-to-zero-temp-obj-BE}
\end{equation}
Furthermore, as an eigenvalue $x$ of the operator $H-\mu\cdot Q$ approaches $0^+$, the corresponding scalar term $T\ln(1-e^{-x/T})$ diverges to $-\infty$. This asymptotic behavior recovers the hard feasibility barrier of the standard, unregularized SDP dual. In the exact SDP, the condition $H-\mu\cdot Q \geq 0$ is a strict constraint; in the Bose--Einstein regularized framework, this constraint is naturally enforced by the domain of the matrix logarithm, which requires $I - e^{-\frac{1}{T}(H-\mu\cdot Q)} > 0$ (and thus $H-\mu\cdot Q > 0$) to remain finite.
\end{remark}

In the standard unregularized semidefinite problem presented in Section~\ref{subsec:General-SDP-optimization}, we have discussed that the optimality conditions enforce strict complementary slackness, $K_{\mu^\star} X^\star = 0$, creating a sharp boundary where modes with positive energy must have exactly zero occupation. When the Bose--Einstein entropic regularizer is introduced as done in~\eqref{eq:sdp_primal_regularized}, the stationarity condition -- obtained by first-order optimality -- provides a different relationship. Setting the gradient of the regularized objective function in~\eqref{eq:sdp_primal_regularized} with respect to $X$ to zero (as done in~\eqref{step:der_reg_obj_fun} in the steps of the proof of Theorem~\ref{prop:dual_sdp_regularized}) yields the optimal primal operator $X_T(\mu_T^\star)$ defined in~\eqref{eq:BE_operator}, from which we can express the optimal dual slack operator $K_{\mu_T^\star}$ directly in terms of the primal operator:
\begin{equation}
    K_{\mu_T^\star} = T \ln\!\left(I + (X_T(\mu_T^\star))^{-1}\right).
\end{equation}
Thus, the complementary slackness condition for the regularized case is as follows:
\begin{equation} \label{eq:reg_comp_slackness_matrix}
    K_{\mu_T^\star} X_T(\mu_T^\star) = T X_T(\mu_T^\star) \ln\!\left(I + (X_T(\mu_T^\star))^{-1}\right).
\end{equation}
Unlike the unregularized case, the trace of this product -- representing the total expected energy penalty or duality gap -- is strictly positive for all $T>0$. Because $K_{\mu_T^\star}$ and $X_T(\mu_T^\star)$ commute, they share a common eigenbasis. For every mode $j$ with energy cost $\lambda_j$ and expected occupation $x_j$, the orthogonality $\lambda_j \cdot x_j = 0$ defined in~\eqref{eq:orthogonality_slackness} is replaced by a smooth scalar relation:
\begin{equation} \label{eq:reg_comp_slackness_eigen}
    \lambda_j \cdot x_j = T x_j \ln\!\left(1 + \frac{1}{x_j}\right).
\end{equation}

By forcing $\Tr[K_{\mu_T^\star} X_T(\mu_T^\star)] > 0$, the entropic regularizer acts as a barrier that prevents the solution from touching the sharp boundary of the positive semidefinite cone, keeping it strictly within the smooth interior. Physically, this relation clarifies the role of the temperature parameter $T$. For every strictly positive energy cost ($\lambda_j > 0$), the condition in~\eqref{eq:reg_comp_slackness_eigen} now permits a non-zero expected occupation $x_j > 0$. Consequently, the sharp, discontinuous constraint of the unregularized problem in~\eqref{eq:orthogonality_slackness} is smoothed into a continuous, differentiable distribution. This captures the exact mathematical mechanism by which finite-temperature bosonic statistics smooth out the macroscopic ground-state condensation of a Bose gas: thermal energy ($T>0$) allows particles to leak into higher-energy excited modes, incurring an energy penalty ($\lambda_j x_j > 0$) in exchange for maximizing the entropy of the distribution. As $T \to 0^+$, the right-hand side in~\eqref{eq:reg_comp_slackness_matrix} vanishes, recovering the strict unregularized condition $K_{\mu^\star} X^\star = 0$. 

Finally, by applying the scalar bound $0 \leq x \ln\!\left(1+1/x\right) \leq 1$ to each of the $d$ modes, we find that the regularized energy penalty for any single mode is strictly bounded by $T$, regardless of how large its expected occupation $x$ becomes. This guarantees the following upper bound on the duality gap: 
\begin{equation}
\Tr[K_{\mu_T^\star} X_T(\mu_T^\star)] \leq T \cdot d.    
\end{equation}

\subsection{Bose--Einstein thermal operators}
\label{subsec:Bose--Einstein-thermal-operators}

Let us highlight that the positive semidefinite operator $X_{T}(\mu)$ defined in~\eqref{eq:BE_operator} has the form of a Bose--Einstein distribution, and the corresponding quantity in the dual objective function in~\eqref{eq:dual_obj_function},
\begin{equation}
T\Tr\!\left[\ln\!\left(I-e^{-\frac{1}{T}\left(H-\mu\cdot Q\right)}\right)\right],
\end{equation}
is equal to the bosonic free energy. Additionally, because standard semidefinite programs do not impose an upper spectral bound on $X$, we consider also the operator $X_{T}(\mu)+I$ which appears naturally in the context of bosonic statistics: while $X_T(\mu)$ corresponds to the expected number of bosons in a given mode, the term $X_T(\mu)+I$ is proportional to the total rate of emission into that mode~\cite{Scully1997QuantumOptics}. Specifically, it is given by:
\begin{equation}
X_{T}(\mu)+I=\left(I-e^{-\frac{1}{T}\left(H-\mu\cdot Q\right)}\right)^{-1},\label{eq:other-outcome-BE}
\end{equation}
as a consequence of the following scalar algebraic identity
\begin{equation}
\frac{1}{e^{x}-1}+1=\frac{e^{x}}{e^{x}-1}=\frac{1}{1-e^{-x}}.\label{eq:scalar-ident-BE-compl}
\end{equation}
Thus, we refer to the operator $X_{T}(\mu)$ as a Bose--Einstein thermal operator and define it formally as follows:

\begin{definition}[Bose--Einstein thermal operator]
Let $d\in\mathbb{N}$, let $A>0$ be a $d\times d$ strictly positive definite Hermitian matrix, and let $T>0$ be a temperature. A Bose--Einstein thermal operator is a positive semidefinite operator $X_{T}(A) \geq 0$ defined as:
\begin{equation}
X_{T}(A)  \coloneqq\left(e^{\frac{A}{T}}-I\right)^{-1},\label{eq:gen-BE-op}
\end{equation}
satisfying
\begin{equation}
X_{T}(A)+I  =\left(I-e^{-\frac{A}{T}}\right)^{-1}.\label{eq:gen-BE-op-2}    
\end{equation}
\end{definition}

As discussed after Remark~\ref{rem:zero_T_limit}, for every finite temperature $T>0$, the thermal operator defined in~\eqref{eq:gen-BE-op} smooths out the sharp, discontinuous occupation rules of the unregularized problem as discussed around~\eqref{eq:optimal-X-structure}. Instead of a strict requirement where only zero-energy modes can be occupied, a mode with a strictly positive energy cost $\lambda > 0$ now receives a finite, non-zero expected occupation number given by the continuous function $(e^{\lambda/T}-1)^{-1}$. This continuous distribution assigns decaying weights to higher-energy modes, capturing the physics of thermal excitations in a finite-temperature Bose gas~\cite{Pathria2011Statistical}.

From this continuous distribution, we can recover the sharp unregularized behavior by taking the zero-temperature limit.
For a strictly feasible dual vector $\mu\in\mathbb{R}^{c}$, which means the grand canonical Hamiltonian in~\eqref{eq:Kmu} is strictly positive ($K_\mu \coloneqq H-\mu\cdot Q > 0$), the corresponding thermal operator vanishes in the zero-temperature limit:
\begin{equation}
\lim_{T\to0^{+}}X_{T}(\mu)=0, \label{eq:temp-T-op-to-zero-temp-op-BE}
\end{equation}
which stems from the scalar limit for $x>0$
\begin{equation}
\lim_{T\to0^{+}}\frac{1}{e^{\frac{x}{T}}-1}=0.
\end{equation}

Conversely, if the dual matrix operator $K_\mu$ becomes singular (meaning an eigenvalue of $K_\mu$ approaches $0$ from above), the corresponding expected occupation number diverges to infinity, since $\lim_{x\to0^{+}} (e^{x/T}-1)^{-1} = \infty$. Interestingly, this divergence corresponds physically to the onset of Bose--Einstein condensation into the ground state \cite{Pathria2011Statistical, Anderson1995Observation, Davis1995BoseEinstein}. Mathematically, it captures the feature of standard semidefinite programs for which the primal variable $X$ is bounded only from below, with no a priori upper bound on its trace.

Finally, for a fixed dual vector $\mu\in\mathbb{R}^{c}$, the operator grows linearly with $T$ in the infinite-temperature limit:
\begin{equation}
\lim_{T\to\infty}X_{T}(\mu)= T(H-\mu\cdot Q)^{-1}.
\end{equation}
This asymptotic scaling follows from the scalar Taylor expansion $e^{x/T} \approx 1 + x/T$ for large $T$, yielding:
\begin{equation}
\lim_{T\to\infty}\frac{1}{e^{\frac{x}{T}}-1}=\lim_{T\to\infty}\frac{T}{x}=\infty.
\end{equation}
Thus, the infinite-temperature limit results in an unbounded expected particle number, consequently overpowering the constraints imposed by the semidefinite program.

\subsection{Approximation error}
\label{subsec:Approximation-error}

In this section we bound the error introduced by replacing the original SDP in~\eqref{eq:general_SDP} with its Bose--Einstein regularization in~\eqref{eq:sdp_primal_regularized}, i.e., the gap between the optimal value $E(\mathcal{Q},q)$ of the original SDP and either the optimal regularized free energy $F_{T}(\mathcal{Q},q)$ or the unregularized energy of the optimal thermal operator, denoted $\tilde{f}_{T}(\mu_{T}^{\star})$ and defined precisely below. These bounds dictate how small the temperature $T$ must be chosen to reach a target precision $\varepsilon$, and hence they directly drive the runtime of the quantum algorithms in Section~\ref{sec:quantum_algorithms}.

We derive two complementary families of bounds:
\begin{enumerate}
    \item Entropy-based bounds (Section~\ref{subsec:Simple-uniform-approximation}), which apply uniformly over the feasible set and depend on the maximum Bose--Einstein entropy $S_{\max}$. They yield the simplest temperature schedule but, in the worst case, treat all $d$ modes on the same footing, reproducing the linear-in-$d$ duality gap of classical interior-point methods.
    \item Spectral-gap-based bounds (Section~\ref{subsec:Spectral-gap-based-approximation}), which exploit the spectrum of the dual slack operator $K_{\mu_{T}^{\star}}$ at optimality. They replace the dimension-dependent penalty by an exponentially suppressed contribution from excited modes, leaving only a linear term in the ground-space degeneracy~$d_0$.
\end{enumerate}

Before stating the bounds, we introduce a useful auxiliary quantity. The dual objective $f_T(\mu)$ in~\eqref{eq:dual_obj_function} can be rewritten as
\begin{equation}
f_{T}(\mu)=\mu\cdot q+\Tr\!\left[\left(H-\mu\cdot Q\right)X_{T}(\mu)\right]-T\,S_{\mathrm{BE}}(X_{T}(\mu)),
\end{equation}
using Lemma~\ref{lem:opt-BE-free-energy} in Appendix~\ref{sec:Proof-of-dual-be-reg} (see~\eqref{eq:obj_fun_RelEntropy} in particular). Removing the entropy term yields the \emph{unregularized energy}
\begin{equation}\label{eq:f_tilde}
\tilde{f}_{T}(\mu) \coloneqq \mu \cdot q + \Tr\!\left[\left(H-\mu \cdot Q\right)X_{T}(\mu)\right],
\end{equation}
so that $f_T(\mu)=\tilde f_T(\mu)-T\,S_{\mathrm{BE}}(X_T(\mu))$. When $\mu=\mu_T^\star$ is a dual optimizer, the first-order condition on $f_T$ (see~\eqref{eq:gradient-BE-gen-obj-func}) forces each equality constraint $\Tr[Q_i X_T(\mu_T^\star)] = q_i$ to be satisfied, so that the term linear in $\mu_T^\star$ cancels and
\begin{equation}\label{eq:approx_at_optimality}
\tilde{f}_{T}(\mu_T^\star) = \Tr\!\left[ H\,X_{T}(\mu_T^\star)\right].
\end{equation}
That is, $\tilde{f}_T(\mu_T^\star)$ is precisely the physical energy of the thermal operator produced by the algorithm. In contrast, the regularized free energy at optimality is
\begin{equation}\label{eq:F_T=f}
F_T(\mathcal{Q},q) = f_{T}(\mu_T^\star) = \Tr\!\left[ HX_{T}(\mu_T^\star)\right] - T\,S_{\mathrm{BE}}(X_{T}(\mu_T^\star)).
\end{equation}
Bounds on $\tilde{f}_T(\mu_T^\star)$ are therefore the relevant ones for hybrid quantum--classical algorithms, where the entropy correction is not directly observable but the energy is.

\subsubsection{Entropy-based approximation bounds}
\label{subsec:Simple-uniform-approximation}

Unlike the Fermi--Dirac framework, where the Pauli exclusion principle ($0\leq M\leq I$) implies a universal entropy bound $S_{\max}\leq d\ln 2$~\cite{liu2026sdp_fermi}, the Bose--Einstein entropy $S_{\mathrm{BE}}(X)$ is unbounded on the cone $X\geq 0$ because bosonic occupation numbers admit no a priori upper bound. To obtain uniform approximation bounds we therefore assume that the feasible set of the primal SDP is compact, and define
\begin{equation} \label{eq:S_max_def}
    S_{\max} \coloneqq \sup_{X \geq 0} \left\{ S_{\mathrm{BE}}(X) :   \Tr[Q_i X] = q_i \, \forall i \in [c] \right\}.
\end{equation}
For example, when the SDP fixes the total expected particle number $\Tr[X]=N$, Schur-concavity of the Bose--Einstein entropy gives $S_{\max}\leq d\,g(N/d)$, where $g(x)$ is the scalar bosonic entropy defined in~\eqref{eq:bosonic-entropy-def}. In this trace-fixed regime, the Bose--Einstein framework coincides (up to rescaling) with the Boltzmann (trace-one) formulation of~\cite{liu2025sdp} at $R=N$; the distinction is that here $N$ appears only as a problem-defined linear constraint $\Tr[X]=N$, rather than as an a priori upper bound to be guessed by the algorithm.

With this convention in place, the same temperature threshold $T\leq\varepsilon/S_{\max}$ yields a clean bound on both $F_T(\mathcal{Q},q)$ and $\tilde{f}_T(\mu_T^\star)$.

\begin{proposition}
\label{prop:simple-approx-bnd}
Fix $\varepsilon > 0$ and $T>0$. If $T \leq \frac{\varepsilon}{S_{\max}}$, then
\begin{equation}
E(\mathcal{Q},q) \geq F_{T}(\mathcal{Q},q) \geq E(\mathcal{Q},q) - \varepsilon. \label{eq:simple-approx-bnd}
\end{equation}
\end{proposition}

\begin{proof}
See Appendix~\ref{sec:proof-simple-approx-bnd}. The proof leverages the non-negativity of the Bose--Einstein entropy and its supremum $S_{\max}$ over the feasible set.
\end{proof}

\begin{proposition}
\label{prop:no-BE-ent-simple-approx-err-bnd}
Fix $\varepsilon > 0$ and $T>0$. Let $\mu_{T}^{\star}\in\mathbb{R}^{c}$ be an optimal dual vector for the regularized problem $F_{T}(\mathcal{Q},q)$. If $T \leq \frac{\varepsilon}{S_{\max}}$, then
\begin{equation}
E(\mathcal{Q},q) \leq \tilde{f}_{T}(\mu_T^\star) \leq E(\mathcal{Q},q) + \varepsilon. \label{eq:simple-approx-bnd-no-BE-ent}
\end{equation}
\end{proposition}

\begin{proof}
See Appendix~\ref{sec:no-BE-ent-simple-approx-err-bnd}. The proof relies on the fact that the dual penalty vanishes at optimality, yielding $\tilde{f}_T(\mu_T^\star) = \Tr[H X_T(\mu_T^\star)]$.
\end{proof}

Both bounds become uninformative when the feasible set is not compact (i.e., when $S_{\max}=+\infty$). For the unregularized energy $\tilde{f}_T(\mu_T^\star)$, however, a strictly weaker but assumption-free dimension-only bound is available, derived directly from regularized complementary slackness.

\begin{remark}[Dimension-only bound for the unregularized energy]
\label{rem:dimension-only-bound}
By the regularized complementary slackness identity in~\eqref{eq:reg_comp_slackness_eigen}, each eigenmode $j$ satisfies $\lambda_j x_j \leq T$. Summing over all $d$ modes yields $\Tr[K_{\mu_T^\star} X_T(\mu_T^\star)] \leq T \cdot d$. Combining this bound with weak duality, which establishes that $E(\mathcal{Q},q) \leq \tilde{f}_T(\mu_T^\star) \leq E(\mathcal{Q},q) + \Tr[K_{\mu_T^\star} X_T(\mu_T^\star)]$ (see Appendix~\ref{sec:proof-approx-error-spectral-gap} for the formal derivation), leads to the dimension-only bound:
\begin{equation}
E(\mathcal{Q},q) \leq \tilde{f}_{T}(\mu_T^\star) \leq E(\mathcal{Q},q) + T\cdot d. \label{eq:dim-only-bound}
\end{equation}
Hence, the temperature constraint $T\leq\varepsilon/d$ suffices to guarantee $|\tilde{f}_T(\mu_T^\star)-E(\mathcal{Q},q)|\leq\varepsilon$ without requiring compactness of the feasible set.
\end{remark}

The dimension-only bound in~\eqref{eq:dim-only-bound} recovers the duality gap of classical interior-point methods on the central path~\cite{Boyd2004convex,vandenberghe1996semidefinite}, where the barrier penalty (the analog $T$ here) must scale as $\varepsilon/d$ to reach precision $\varepsilon$. While consistent with classical theory, such a bound represents the worst-case regime: all $d$ modes are treated as if they occupy the bottom of the spectrum of $K_{\mu_T^\star}$, with each contributing a thermal fluctuation of order $T$. In the next subsection, we show that whenever $K_{\mu_T^\star}$ possesses a spectral gap, this bound can be sharpened.

\subsubsection{Spectral-gap-based approximation bound}
\label{subsec:Spectral-gap-based-approximation}

We now exploit the spectrum of $K_{\mu_T^\star}$ at optimality to derive a bound on $\tilde{f}_T(\mu_T^\star)$ that scales only with the spectral properties of the slack operator $K_{\mu_T^\star}$ rather than with the dimension $d$. The relevant spectral parameters are as follows:
\begin{enumerate}
    \item the minimum eigenvalue $\lambda_{\min}$ of $K_{\mu_T^\star}$,
    \item its degeneracy $d_0$ (the dimension of the ground space), and
    \item the spectral gap $\Delta$ separating the ground space from the excited spectrum. 
\end{enumerate}
As discussed below, the two parameters $\lambda_{\min}$ and $\Delta$ play conceptually different roles: $\Delta$ is a structural property of $K_{\mu_T^\star}$ that remains bounded away from zero in the gapped regime, while $\lambda_{\min}$ is forced to vanish along the trajectory $\mu_T^\star \to \mu^\star$ by complementary slackness.

\begin{proposition}
\label{prop:approx-error-spectral-gap}
Let $T>0$. Let $\mu_{T}^{\star}\in\mathbb{R}^{c}$ be an optimal dual vector for the regularized free energy $F_{T}(\mathcal{Q},q)$, and suppose that the optimal dual slack operator $K_{\mu_T^\star} \coloneqq H - \mu_T^\star \cdot Q > 0$ has a minimum eigenvalue $\lambda_{\min} > 0$ with degeneracy $d_0$ and a spectral gap $\Delta > 0$ separating the ground space from the remaining $d-d_0$ excited eigenvalues. Then
\begin{multline}
E(\mathcal{Q},q) \leq \tilde{f}_{T}(\mu_{T}^{\star}) \leq  E(\mathcal{Q},q)  \\
 + T\,d_0 + \left(d-d_0\right) \frac{\lambda_{\min} + \Delta}{e^{(\lambda_{\min} + \Delta)/T}-1}. 
\label{eq:duality_gap_spectral_bound}
\end{multline}
Moreover, for every precision $\varepsilon > 0$, the temperature constraint
\begin{equation}
T \leq \min \left\{ \frac{\varepsilon}{2 d_0}, \, \frac{\lambda_{\min} + \Delta}{\ln\!\left(1 + \frac{2(d-d_0)(\lambda_{\min} + \Delta)}{\varepsilon}\right)} \right\} \label{eq:temp_threshold_spectral}
\end{equation}
guarantees that $\left|\tilde{f}_{T}(\mu_T^\star) - E(\mathcal{Q},q)\right| \leq \varepsilon$.
\end{proposition}

\begin{proof}
See Appendix~\ref{sec:proof-approx-error-spectral-gap}. The proof exploits the fact that the duality gap at optimality simplifies to $\Tr[K_{\mu_T^\star} X_T(\mu_T^\star)]$, which we then bound term by term using the spectral decomposition of $K_{\mu_T^\star}$.
\end{proof}

The bound in~\eqref{eq:duality_gap_spectral_bound} separates the contributions of the ground space and the excited space of $K_{\mu_T^\star}$:
\begin{itemize}
\item \emph{Ground-space contribution} $T\,d_0$: a linear-in-$T$ penalty whose prefactor is the ground-space degeneracy $d_0$, not the ambient dimension $d$.
\item \emph{Excited-space contribution} $(d-d_0)\,\frac{\lambda_{\min}+\Delta}{e^{(\lambda_{\min}+\Delta)/T}-1}$: a sum over the $d-d_0$ excited modes, each thermally suppressed by a Boltzmann factor at energy $\lambda_{\min}+\Delta$.
\end{itemize}
In the low-temperature regime $T\ll\lambda_{\min}+\Delta$, the denominator of the excited-space term is dominated by the exponential, $e^{(\lambda_{\min}+\Delta)/T}-1 \approx e^{(\lambda_{\min}+\Delta)/T}$, so that
\begin{equation}
\begin{split}
    \tilde{f}_{T}(\mu_{T}^{\star}) &- E(\mathcal{Q},q) \leq T\,d_0 \\
    &+ \mathcal{O}\!\left((d-d_0)(\lambda_{\min}+\Delta)\,e^{-(\lambda_{\min}+\Delta)/T}\right).
\end{split}
\label{eq:duality_gap_spectral_bound_O}
\end{equation}
This low-temperature regime remains well-defined in the limit $T \to 0$, even though complementary slackness will force $\lambda_{\min}\to 0$ along the way. The reason is that the relevant scale in the bound is $\lambda_{\min}+\Delta$, and the gap $\Delta$ is a structural property of the dual slack operator at the unregularized optimum: by continuity of the spectrum, $\Delta(K_{\mu_T^\star}) \to \Delta(K_{\mu^\star}) > 0$ as $T \to 0$ in the gapped regime, so that $\lambda_{\min}+\Delta \to \Delta(K_{\mu^\star}) > 0$ and $T/(\lambda_{\min}+\Delta) \to 0$. The exponential suppression of the excited modes is therefore driven by the structural gap.

The excited-space contribution is therefore exponentially small in $(\lambda_{\min}+\Delta)/T$ and quickly becomes negligible compared to the ground-space term $T\,d_0$. The dominant cost of reaching precision $\varepsilon$ is then dictated by $d_0$ rather than by $d$:
\begin{itemize}
\item \emph{Non-degenerate ground state} ($d_0=1$): the schedule $T\sim\mathcal{O}(\varepsilon)$ suffices, independently of $d$.
\item \emph{Degenerate ground space} ($d_0>1$): the schedule $T\sim\mathcal{O}(\varepsilon/d_0)$ suffices. The penalty grows with the macroscopic degeneracy $d_0$, but still not with the ambient dimension $d$.
\end{itemize}
Compared to the dimension-only bound in~\eqref{eq:dim-only-bound}, the spectral-gap bound is tighter whenever $d_0\ll d$ and $\Delta\gtrsim T$. In particular, if $\Delta^{-1}$ and $d_0$ scale at most polynomially in $\ln d$, then $T$ scales polynomially in $\ln d$ rather than in $d$, an exponential improvement that classical IPMs cannot leverage because the logarithmic barrier treats every dimension equally.

This sharper guarantee comes with a structural trade-off between the approximation bound and the algorithmic cost in the low-temperature regime. To state it precisely, it is useful to distinguish two separate roles played by the spectral parameters of $K_{\mu_T^\star}$ along the optimization trajectory:
\begin{itemize}
    \item The \emph{approximation bound} in~\eqref{eq:duality_gap_spectral_bound} is governed by the spectral gap $\Delta$, which in the gapped regime is a fixed property of the unregularized optimum and remains bounded away from zero as $T\to 0$. The approximation analysis therefore remains valid in the low-temperature limit.
    \item The \emph{per-call quantum runtime} of the algorithms presented in Section~\ref{sec:quantum_algorithms}, in contrast, scales inverse-polynomially in the minimum eigenvalue $\lambda_{\min}$, which is forced to vanish by exact complementary slackness at $\mu^\star$: as $\varepsilon\to 0$ and hence $T\to 0$, we have $\mu_T^\star\to\mu^\star$ and $\lambda_{\min}(K_{\mu_T^\star})\to 0$.
\end{itemize}
A precise quantitative description of this $\lambda_{\min}$-vs-$T$ coupling, and its effect on the end-to-end $\varepsilon$-complexity, is the open problem we return to in Section~\ref{sec:conclusion}.

\subsection{Gradient and Hessian of dual objective function}

\label{subsec:Gradient-and-Hessian}

In this section, we derive analytical expressions for the gradient
and the Hessian of the objective function $f_{T}(\mu)$ defined in~\eqref{eq:dual_obj_function}.
These expressions are essential in developing gradient-ascent algorithms
for optimizing $f_{T}(\mu)$ and analyzing the performance of such
algorithms. We also prove that the Hessian of $f_{T}(\mu)$ is negative semidefinite and has bounded spectral norm, thus guaranteeing convergence of
gradient ascent. The expressions for the
Hessian matrix elements can be viewed as being analogous to those
reported in~\cite{Patel2025a,liu2025sdp,liu2026sdp_fermi}.

\begin{proposition}
\label{prop:gradient-BE-gen-obj-func}For $T>0$, the $i$th partial
derivative of the objective function $f_{T}(\mu)$ in~\eqref{eq:dual_obj_function}
is given by
\begin{equation}\label{eq:gradient-BE-gen-obj-func}
\frac{\partial}{\partial\mu_{i}}f_{T}(\mu)=q_{i}-\Tr\!\left[X_{T}(\mu)Q_{i}\right],
\end{equation}
where the operator $X_{T}(\mu)$ is defined in~\eqref{eq:BE_operator}.
\end{proposition}

\begin{proof}
See Appendix~\ref{sec:Proof-of-gradient-BE_dual}.
\end{proof}

\begin{proposition}
\label{prop:hessian-BE-gen-obj-func}For all $c\in\mathbb{N}$, $T>0$,
and $i,j\in\left[c\right]$, the elements of the Hessian matrix $\nabla^{2}f_{T}(\mu)$
for the objective function $f_{T}(\mu)$ in~\eqref{eq:dual_obj_function}
are given by
\begin{align}
 & \frac{\partial^{2}}{\partial\mu_{i}\partial\mu_{j}}f_{T}(\mu)\nonumber \\
 & =-\frac{1}{T}\int_{0}^{1}ds\,\Tr\!\left[X_{T}(\mu,s)Q_{i}X_{T}(\mu,1-s)Q_{j}\right],\label{eq:1st-hessian-exp}\\
 & =-\frac{1}{T}\operatorname{Re}\!\left\{\Tr\!\left[X_{T}(\mu)\Phi_{\mu}(Q_{i})\left(X_{T}(\mu)+I\right)Q_{j}\right]\right\},\label{eq:2nd-hessian-exp}
\end{align}
where the operator $X_{T}(\mu,s)$ is defined for all
$s\in\left[0,1\right]$ as
\begin{equation}
X_{T}(\mu,s)\coloneqq\frac{e^{\frac{s}{T}\left(H-\mu\cdot Q\right)}}{e^{\frac{1}{T}\left(H-\mu\cdot Q\right)}-I},\label{eq:BEop-t-depend}
\end{equation}
and the quantum channel $\Phi_{\mu}$ and the high-peak tent probability
density $\gamma(t)$ are respectively defined as
\begin{align}
\Phi_{\mu}(A) & \coloneqq\int_{-\infty}^{\infty}dt\,\gamma(t)\,e^{-i\left(H-\mu\cdot Q\right)t/T}Ae^{i\left(H-\mu\cdot Q\right)t/T},\label{eq:Phi-channel}\\
\gamma(t) & \coloneqq\frac{2}{\pi}\ln\left|\coth\!\left(\frac{\pi t}{2}\right)\right|.\label{eq:high-peak-tent-prob-dens}
\end{align}
\end{proposition}

\begin{proof}
See Appendix~\ref{sec:Proof-of-Hessian-BE-gen-obj-func}.
\end{proof}

\begin{proposition}
\label{prop:hessian-NSD-spec-up-bnd}
For all $T>0$, the Hessian matrix
$\nabla^{2}f_{T}(\mu)$ is negative semidefinite; i.e., for all $v\in\mathbb{R}^{c}$,
\begin{equation}
v^{T}\nabla^{2}f_{T}(\mu)v\leq0,
\end{equation}
and its spectral norm is bounded as follows. Let $\mathcal{M}\subseteq\mathbb{R}^{c}$ be a parameter set representing the trajectory of the optimization algorithm, such that the dual slack operator is strictly feasible with a minimum eigenvalue bounded from below: 
\begin{equation}
\inf_{\mu\in\mathcal{M}}\lambda_{\min}(H-\mu\cdot Q)\geq\lambda_{\min}>0.    
\end{equation}
Then for all $\mu\in\mathcal{M}$,
\begin{equation}
    \left\Vert \nabla^{2}f_{T}(\mu)\right\Vert \leq L_{T}\coloneqq\frac{\bar{n}_{\lambda_{\min}}(\bar{n}_{\lambda_{\min}}+1)}{T}\sum_{i\in\left[c\right]}\left\Vert Q_{i}\right\Vert _{1}\left\Vert Q_{i}\right\Vert ,\label{eq:smoothness-parameter}
\end{equation}
where $\bar{n}_{\lambda_{\min}}$ denotes the worst-case Bose--Einstein occupation number:
\begin{equation}
\bar{n}_{\lambda_{\min}}\coloneqq (e^{\lambda_{\min}/T}-1)^{-1}.     
\end{equation}
\end{proposition}

\begin{proof}
See Appendix~\ref{sec:hessian-NSD-spec-up-bnd}.
\end{proof}

\begin{remark}
\label{rem:dual-concave-smoothness-param}
As a consequence of Proposition~\ref{prop:hessian-NSD-spec-up-bnd}, the function $f_{T}(\mu)$ is concave in $\mu$ with smoothness parameter $L_{T}$, as defined in~\eqref{eq:smoothness-parameter}. This guarantees that gradient ascent with step size $\eta\in(0,1/L_{T}]$ converges to a point $\delta$-close to the global maximum of $f_{T}(\mu)$ in $\mathcal{O}(L_{T}\left\Vert \mu_{T}^{\star}\right\Vert^{2}/\delta)$ steps.
\end{remark}

Let us note that the smoothness parameter $L_{T}$ differs from its Fermi--Dirac counterpart~\cite[Eq.~(81)]{liu2026sdp_fermi} by the multiplicative prefactor $\bar{n}_{\lambda_{\min}}(\bar{n}_{\lambda_{\min}}+1)$. This term captures the unbounded nature of bosonic occupation numbers, depending on $\lambda_{\min}$, the margin of strict dual feasibility.
This reveals how the behavior of the algorithm changes depending on the distance to the positive semidefinite cone boundary:
\begin{itemize}
    \item At low temperatures ($T\ll\lambda_{\min}$), we have $\bar{n}_{\lambda_{\min}}\approx e^{-\lambda_{\min}/T}\to 0$, resulting in $L_{T}\approx e^{-\lambda_{\min}/T}/T$. This exponentially small smoothness parameter means that the gradient of $f_T(\mu)$ varies slowly with $\mu$: the dual objective is well-conditioned, the admissible step size $\eta\leq 1/L_T$ is correspondingly large, and the iteration count $\mathcal{O}(L_T\left\|\mu_T^\star\right\|^2/\delta)$ of Remark~\ref{rem:dual-concave-smoothness-param} is small. The favorable regime is therefore one of \emph{good conditioning}, provided the current dual iterate remains far from exact ground-state condensation.
    \item When $T \gg \lambda_{\min}$, a regime reached in our setting by driving $\lambda_{\min}\to 0$ along the trajectory to the unregularized optimum, we enter the heavily populated regime where $\bar{n}_{\lambda_{\min}}\approx T/\lambda_{\min}$, and consequently $L_{T}\sim \frac{T}{\lambda_{\min}^{2}}\sum_{i}\left\Vert Q_{i}\right\Vert _{1}\left\Vert Q_{i}\right\Vert$. As the optimization path approaches $\mu^\star$, the margin of dual feasibility vanishes ($\lambda_{\min} \to 0$) and the smoothness parameter diverges, forcing the algorithm to take proportionally smaller step sizes and increasing the iteration count.
\end{itemize}

The two regimes above are the algorithmic counterpart of the trade-off discussed at the end of Section~\ref{subsec:Spectral-gap-based-approximation}. There, the approximation bound stayed well-controlled as $T\to 0$ thanks to the structural gap $\Delta$, which remains bounded away from zero. The smoothness parameter $L_T$ has no analogous counterpart: it depends inverse-polynomially on $\lambda_{\min}$ alone, and so diverges as $\lambda_{\min}\to 0$. This divergence is the main obstacle to a clean end-to-end $\varepsilon$-complexity scaling polynomially in $\ln d$, and motivates the open problem of characterizing how $\lambda_{\min}$ scales with $T$ along the trajectory $\mu_T^\star\to\mu^\star$.

\subsection{Gradient ascent for optimizing Bose--Einstein thermal operators}
\label{subsec:classical-gradient-ascent}

The smoothness and concavity properties established in Section~\ref{subsec:Gradient-and-Hessian} guarantee that the dual problem~\eqref{eq:dual_def} can be solved by classical first-order optimization. We collect here the resulting deterministic gradient-ascent template, which serves both as a stand-alone classical algorithm (when the gradient $\nabla f_T(\mu)$ is computed exactly) and as the conceptual backbone for the hybrid quantum--classical algorithms developed in Section~\ref{sec:overall-complexity}.

Recall from Proposition~\ref{prop:gradient-BE-gen-obj-func} that the gradient of the dual objective takes the form
\begin{equation}
\frac{\partial}{\partial\mu_i}f_T(\mu) = q_i - \Tr[X_T(\mu)\,Q_i] \qquad \forall i\in[c],
\end{equation}
where $X_T(\mu)=(e^{K_\mu/T}-I)^{-1}$. The deterministic gradient-ascent update is then $\mu^{j}\leftarrow\mu^{j-1}+\eta\,\nabla f_T(\mu^{j-1})$, with the step size $\eta$ controlled by the smoothness constant $L_T$ of~\eqref{eq:smoothness-parameter}. The complete procedure is summarized in Algorithm~\ref{alg:classical-gradient-ascent}.

\begin{algorithm}
\caption{Gradient ascent for the Bose--Einstein dual}
\label{alg:classical-gradient-ascent}
\KwData{Target accuracy $\delta>0$; learning rate $\eta\in(0,1/L_T]$ with $L_T$ as in~\eqref{eq:smoothness-parameter}; total number of steps $J\geq L_T\left\|\mu_T^\star\right\|^2/\delta$.}
\KwResult{Estimate $\tilde f_T(\mu^J)$ of the optimal SDP value $E(\mathcal{Q},q)$.}
Choose a temperature $T>0$ small enough to guarantee an approximation error of at most $\delta/2$ via Proposition~\ref{prop:no-BE-ent-simple-approx-err-bnd}, Remark~\ref{rem:dimension-only-bound}, or Proposition~\ref{prop:approx-error-spectral-gap}.\;
Initialize $\mu^{0}\in\mathbb{R}^c$ (e.g., $\mu^{0}\leftarrow 0$).\;
\For{$j=1,\dots,J$}{
    For each $i\in[c]$, compute the gradient component $g_i^{(j-1)} \leftarrow q_i - \Tr[X_T(\mu^{j-1})\,Q_i]$.\;
    Update $\mu^{j}\leftarrow\mu^{j-1}+\eta\,g^{(j-1)}$.\;
}
\Return $\tilde f_T(\mu^J) = \mu^J\cdot q + \Tr[(H-\mu^J\cdot Q)X_T(\mu^J)] = \Tr[H\,X_T(\mu^J)] + \mu^J\cdot(q - \Tr[X_T(\mu^J)Q])$.\;
\end{algorithm}

By Remark~\ref{rem:dual-concave-smoothness-param}, the iteration count $J=\Theta(L_T\left\|\mu_T^\star\right\|^2/\delta)$ is sufficient to drive the suboptimality below $\delta/2$ in the dual; combined with the approximation-error budget of $\delta/2$ from the temperature schedule, the total error on the original SDP is at most $\delta$. We note that Algorithm~\ref{alg:classical-gradient-ascent} is a deterministic, classical procedure: it assumes that the thermal traces $\Tr[X_T(\mu)Q_i]$ are evaluated exactly, which is computationally feasible only when $d$ is small enough to compute the matrix exponential of $K_\mu$ classically. The hybrid quantum--classical version in Section~\ref{subsec:stochastic-GA} replaces these exact evaluations with the unbiased stochastic estimator of Algorithm~\ref{alg:BE-thermal-alg}.

The same dual-objective machinery also admits a second-order Newton scheme. Recalling the expression of the Hessian elements from Proposition~\ref{prop:hessian-BE-gen-obj-func}, the Newton update is therefore
\begin{equation}
\mu^j\leftarrow\mu^{j-1}-\eta\,[\nabla^2 f_T(\mu^{j-1})]^{-1}\nabla f_T(\mu^{j-1}),     
\end{equation}
and inherits the standard convergence guarantees of Newton's method on smooth concave objectives~\cite{Nesterov2006Cubic}. The classical Newton algorithm, thus, amounts to substituting the gradient update of Algorithm~\ref{alg:classical-gradient-ascent} with the Newton step, and its hybrid quantum--classical realization -- which employs the quantum estimator in Algorithm~\ref{alg:hessian-est-alg} -- is fully detailed in Algorithm~\ref{alg:q-stoch-newton}.

\section{Quantum algorithms for semidefinite programming with Bose--Einstein operators}
\label{sec:quantum_algorithms}

The classical gradient-ascent template of Algorithm~\ref{alg:classical-gradient-ascent} assumes exact access to the thermal traces $\Tr[X_T(\mu)Q_i]$ that define the gradient of the dual objective $f_T(\mu)$. For SDPs of practical interest, however, the Hilbert-space dimension $d$ is large, and these traces cannot be evaluated exactly on a classical computer. In this section, we develop hybrid quantum--classical algorithms that estimate the gradient and Hessian of $f_T(\mu)$ to a target additive precision $\varepsilon$, using only elementary quantum algorithmic primitives.

This section is organized as follows. Section~\ref{subsec:Quantum-algorithm-gradient} presents the quantum algorithm for estimating a single thermal trace $\Tr[X_T(\mu)Q_i]$ (Algorithm~\ref{alg:BE-thermal-alg}). The construction has three ingredients: (i) a geometric-series expansion of $X_T(\mu)=(e^{K_\mu/T}-I)^{-1}$ that exploits the strict positivity of the dual slack operator $K_\mu$; (ii) the observation that the characteristic function of a Cauchy distribution equals the desired exponential form, so that sampling a random evolution time $t$ converts the operator $e^{-mK_\mu/T}$ into the expectation of a unitary $e^{-itK_\mu}$ that a quantum computer can implement; and (iii) the Hadamard test, which extracts the real part of $\Tr[\rho_k e^{-itK_\mu}]$ under the standard linear-combination-of-states model for the input matrices. We prove the correctness of the resulting estimator and bound the per-call quantum gate runtime in closed form. Section~\ref{subsec:Quantum-algorithm-Hessian} extends the same ideas to a double geometric series for the elements of the Hessian matrix~\eqref{eq:1st-hessian-exp}, replacing the Hadamard test with a Hadamard test supplemented by a controlled-SWAP gate (Algorithm~\ref{alg:hessian-est-alg}). The end-to-end runtime guarantees for the full hybrid optimization, combining these per-call estimators with first- and second-order classical iteration counts, are deferred to Section~\ref{sec:overall-complexity}.

\subsection{Quantum algorithm for implementing Bose--Einstein thermal operators}
\label{subsec:Quantum-algorithm-gradient}

Evaluating the gradient of the dual objective function $f_T(\mu)$ in~\eqref{eq:gradient-BE-gen-obj-func} requires estimating the operator trace $\Tr[X_T(\mu) Q_i]$. 
Here, we present a quantum algorithm that achieves this using only basic quantum routines, specifically the Hadamard test and Hamiltonian simulation with the evolution time sampled randomly from a Cauchy distribution. By exploiting the fact that the operator $K_\mu = H - \mu \cdot Q$ is strictly positive definite, and hence has a strictly positive spectrum, we can expand the Bose--Einstein thermal operator $X_T(\mu)$ as a geometric series in unnormalized Boltzmann weights:
\begin{equation}
    X_T(\mu) = \left( e^{\frac{K_\mu}{T}} - I \right)^{-1} = \sum_{m=1}^{\infty} e^{-m K_\mu/T}. \label{eq:BE_geometric_series}
\end{equation}
By linearity, estimating the required trace thus reduces to evaluating a sum of individual components:
\begin{equation}
    \Tr[X_T(\mu) Q_i] = \sum_{m=1}^{\infty} \Tr\!\left[ Q_i e^{-m K_\mu/T} \right]. \label{eq:BE_trace_series}
\end{equation}

Because the eigenvalues of $K_\mu$ are strictly bounded from below by the minimum eigenvalue $\lambda_{\min} > 0$, the magnitude of the terms in the series decays exponentially as $\mathcal{O}(e^{-m \lambda_{\min} / T})$. Thus, to estimate the term $\Tr[X_T(\mu) Q_i]$, and in turn the gradient, to within a target precision $\varepsilon$, the infinite sum can be truncated at a finite order $M = \mathcal{O}(T \ln(1/\varepsilon) / \lambda_{\min})$ (see Appendix~\ref{app:complexity-analysis} for the derivation of the truncation order $M$), thus restricting the sum and having
\begin{equation}
    \Tr[X_T(\mu) Q_i] \approx \sum_{m=1}^{M} \Tr\!\left[ Q_i e^{-m K_\mu/T} \right]. \label{eq:BE_trace_series_truncated}
\end{equation}

To evaluate each term $\Tr[Q_i e^{-m K_\mu / T}]$ in~\eqref{eq:BE_trace_series_truncated} on a quantum computer, we assume that the matrices $H$ and $Q_i$ are given as a linear combination of states, also called quantum state model in the SDP literature~\cite{vanApeldoorn2019SDP}. Thus, we assume that
\begin{align}
    Q_{i} &= \sum_{k} \alpha_{i,k} \rho_{k}, \label{eq:state-model-2}\\
    K_\mu &= \sum_{k} h_{k} \rho_{k}, \label{eq:state-model-1}
\end{align}
where $h_k, \alpha_{i,k} \in \mathbb{R}$ and $\rho_k$ is a density matrix for all $k$ ($\rho_k \geq 0$, $\Tr[\rho_k] = 1$). 
Under this model, the terms in~\eqref{eq:BE_trace_series_truncated} can be decomposed by linearity as:
\begin{equation}
    \Tr\!\left[ Q_i e^{-m K_\mu/T} \right] = \sum_{k} \alpha_{i,k} \Tr\!\left[\rho_k e^{-m K_\mu/T} \right]. \label{eq:states_trace_linearity}
\end{equation}

Now, to transform the non-unitary weight $e^{-m K_\mu/T}$ in~\eqref{eq:states_trace_linearity} into a unitary operation implementable on a quantum computer, we exploit the fact that the characteristic function of a Cauchy distribution coincides with the non-unitary weight term. 
Let $t$ be a random variable drawn from a Cauchy distribution with scale parameter $\tau = m/T$, whose probability density function is 
\begin{equation}\label{eq:Cauchy_prob}
    p(t; \tau) \coloneqq \frac{\tau}{\pi(t^2 + \tau^2)}.
\end{equation}
For every strictly positive definite operator $K_\mu > 0$, the expected value of the unitary evolution operator $e^{-i t K_\mu}$ over this classical distribution evaluates exactly to the desired non-unitary thermal weight:
\begin{equation}
    \mathbb{E}_{t \sim p(t; \tau)} \left[ e^{-i t K_\mu} \right] = \int_{-\infty}^{\infty} dt \ p(t; \tau) e^{-i t K_\mu} = e^{-\tau K_\mu}. \label{eq:cauchy_integral}
\end{equation}

\begin{figure}[t]
    \begin{minipage}{\columnwidth}
    \centering
    \hspace{1cm}
    \scalebox{1.5}{
    \Qcircuit @C=1.2em @R=1.5em {
        \lstick{|0\rangle} & \gate{H} & \ctrl{1} & \gate{H} & \meter \\
        \lstick{\rho_k} & \qw & \gate{e^{-i t K_\mu}} & \qw & \qw 
        }
    }
    \end{minipage}
    \caption{The Hadamard test quantum circuit used in Algorithm~\ref{alg:BE-thermal-alg} to estimate the component $\operatorname{Re}\!\left\{\Tr \left[ \rho_k e^{-i t K_\mu}\right]\right\}$. The evolution time $t$ is sampled classically from a Cauchy distribution.}
    \label{fig:BE_circuit}
\end{figure}
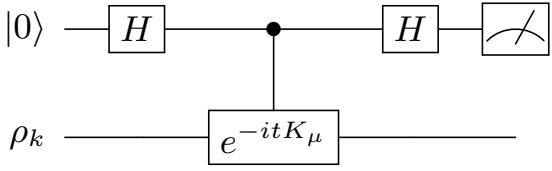

Thus, the desired trace term $\Tr[X_T(\mu) Q_i]$ can be estimated as
\begin{equation}
    \Tr[X_T(\mu) Q_i] \approx \sum_{m=1}^{M} \sum_{k} \alpha_{i,k} \mathbb{E}_{t \sim p(t; \tau)} \Tr\!\left[\rho_k e^{-i t K_\mu} \right],
\end{equation}
where each term $\Tr\!\left[\rho_k e^{-i t K_\mu}\right]$ can be estimated by the Hadamard test quantum circuit in Figure~\ref{fig:BE_circuit}.

The quantum algorithm for estimating the gradient component $\Tr[X_T(\mu) Q_i]$ proceeds as described in Algorithm~\ref{alg:BE-thermal-alg}.

\begin{algorithm}
\caption{Quantum algorithm for estimating $\Tr[X_T(\mu) Q_i]$}
\label{alg:BE-thermal-alg}
\KwData{Minimum eigenvalue bound $\lambda_{\min}$, target precision $\varepsilon$, state coefficients $\{\alpha_{i,k}\}$, Hadamard test circuit in Figure~\ref{fig:BE_circuit}}
\KwResult{Estimate of $\Tr[X_T(\mu) Q_i]$ to precision $\mathcal{O}(\varepsilon)$}
\textbf{Series Truncation:} Determine the cutoff order $M = \widetilde{\mathcal{O}}(T/\lambda_{\min})$ based on the target precision $\varepsilon$.\;
\textbf{Sampling:} For each integer $m \in \{1, \dots, M\}$, determine the required number of Joint Monte Carlo shots $N_{\operatorname{shots}}^{(m)} = \mathcal{O}(M^2 \left\|\alpha_i\right\|_1^2 / \varepsilon^2)$. Draw $N_{\operatorname{shots}}^{(m)}$ independent sample pairs $(t, k)$, where the evolution time $t$ is drawn from the Cauchy distribution $p(t; m/T)$ and the state index $k$ is drawn with probability $p_k = |\alpha_{i,k}| / \left\|\alpha_i\right\|_1$.\;
\textbf{Hadamard test:} For each sampled pair $(t, k)$, run the Hadamard test in Figure~\ref{fig:BE_circuit} with the controlled version of the unitary $e^{-i t K_\mu}$ to estimate the component $\operatorname{Re}\{\Tr[\rho_k e^{-i t K_\mu}]\}$.\;
\textbf{Classical Aggregation:} For each $m$, weight the measurement outcomes by the estimator factor $\left\|\alpha_i\right\|_1 \operatorname{sgn}(\alpha_{i,k})$, average over the $N_{\operatorname{shots}}^{(m)}$ samples, and sum the final averaged results over all $m \in \{1, \dots, M\}$.
\end{algorithm}

\begin{proposition}
\label{prop:BE-algorithm-correctness}
Let $K_\mu > 0$. In the limit of infinite Joint Monte Carlo shots ($N_{\operatorname{shots}}^{(m)} \to \infty$ for all $m$) and infinite series depth ($M \to \infty$), the expected value of the output of Algorithm~\ref{alg:BE-thermal-alg} converges exactly to the trace $\Tr[X_T(\mu) Q_i]$. 
\end{proposition}

\begin{proof}
See Appendix~\ref{sec:proof-BE-algorithm}. 
\end{proof}

In Appendix~\ref{app:complexity-analysis}, we analyze the computational complexity required to estimate the trace $\Tr[X_T(\mu) Q_i]$ to a target additive error $\varepsilon$. By synthesizing the sampling of the truncated geometric series, the Cauchy evolution times, and the sampled states into a single Joint Monte Carlo estimator, we establish the total quantum runtime, defined as the total number of two-qubit gates executed across all circuit repetitions. The complexity depends on the spectral properties of the dual slack operator and the 1-norms of the observables.

\begin{proposition}[Gradient Estimator Complexity]
\label{prop:gradient-complexity}
Fix a temperature $T>0$, and let $K_\mu > 0$ with minimum eigenvalue $\lambda_{\min} > 0$. To estimate the element $\Tr[X_T(\mu) Q_i]$ of the gradient defined in~\eqref{eq:gradient-BE-gen-obj-func} to an additive precision $\varepsilon$, the total quantum runtime of Algorithm~\ref{alg:BE-thermal-alg} scales asymptotically as:
\begin{equation}
    \text{Total Runtime} = \widetilde{\mathcal{O}}\!\left( \left\|\alpha_i\right\|_1^3 \left\|h\right\|_1 \frac{T^4}{\lambda_{\min}^5 \varepsilon^3} \right), \label{eq:main_complexity}
\end{equation}
where $T>0$ is the temperature parameter of the Bose--Einstein regularization in~\eqref{eq:sdp_primal_regularized}, $\left\|\alpha_i\right\|_1$ and $\left\|h\right\|_1$ are the 1-norms of the coefficient vectors for $Q_i$ and $K_\mu$ in the linear-combination-of-states decompositions~\eqref{eq:state-model-2} and~\eqref{eq:state-model-1} respectively, and the $\widetilde{\mathcal{O}}$ notation suppresses logarithmic factors.
\end{proposition}

\begin{proof}
See Appendix~\ref{app:complexity-analysis}.
\end{proof}

\subsection{Quantum algorithm for estimating Hessian matrix elements}
\label{subsec:Quantum-algorithm-Hessian}

Here we provide a quantum algorithm for estimating the elements of the Hessian matrix $H_{ij} = \frac{\partial^{2}}{\partial\mu_{i}\partial\mu_{j}}f_{T}(\mu)$ defined in Proposition~\ref{prop:hessian-BE-gen-obj-func}.

As in the quantum algorithm for estimating the gradient elements, by utilizing the strict positivity of the dual slack operator $K_\mu$, we can evaluate the Hessian elements expression in~\eqref{eq:1st-hessian-exp} via a double geometric series. Expanding the parameter-dependent operators, $X_T(\mu, s)$ and  $X_T(\mu, 1-s)$ defined in~\eqref{eq:BEop-t-depend}, using the geometric identity $(e^{K_\mu/T}-I)^{-1} = \sum_{m=1}^{\infty} e^{-m K_\mu / T}$ yields:
\begin{align}
    X_T(\mu, s) &= e^{\frac{s}{T} K_\mu} \sum_{m_1=1}^{\infty} e^{-m_1 \frac{K_\mu}{T}} \\
    & = \sum_{m_1=1}^{\infty} e^{-\tau_1(m_1, s) K_\mu},
\end{align}
where $\tau_1(m_1, s) = \frac{m_1-s}{T}$, and
\begin{align}
    X_T(\mu, 1-s) &= e^{\frac{1-s}{T} K_\mu} \sum_{m_2=1}^{\infty} e^{-m_2 \frac{K_\mu}{T}} \\
    & = \sum_{m_2=1}^{\infty} e^{-\tau_2(m_2, s) K_\mu},
\end{align}
where $\tau_2(m_2, s) = \frac{m_2-1+s}{T}$.
Under the quantum state model where $Q_i = \sum_k \alpha_{i,k}\rho_k$ and $Q_j = \sum_l \alpha_{j,l}\sigma_l$, the Hessian element defined in~\eqref{eq:1st-hessian-exp} becomes:
\begin{align}
    \begin{split}
        H_{ij} = -\frac{1}{T} & \int_{0}^{1}ds  \sum_{m_1=1}^{\infty}  \sum_{m_2=1}^{\infty} \sum_{k,l} \alpha_{i,k} \alpha_{j,l} \times \\
    &  \Tr\!\left[e^{-\tau_1(m_1,s) K_\mu}\rho_k e^{-\tau_2(m_2,s) K_\mu}\sigma_l\right].
    \end{split}
\end{align}

By sampling $s$ uniformly from $[0,1]$ and converting the non-unitary weights $e^{-\tau_1 K_\mu}$ and $e^{-\tau_2 K_\mu}$ into Cauchy-distributed unitary time-evolutions $U_1 = e^{-i t_1 K_\mu}$ and $U_2 = e^{-i t_2 K_\mu}$ respectively, the required unitary trace is $\operatorname{Re}\!\left\{\Tr[U_1 \rho_k U_2 \sigma_l]\right\}$. This can be estimated by a Hadamard test supplemented with a controlled-SWAP gate, as depicted in Figure~\ref{fig:hessian_circuit}.

\begin{figure}[t]
    \begin{minipage}{\columnwidth}
    \centering
    \scalebox{1.5}{
    \Qcircuit @C=0.45em @R=1.em {
        \lstick{|0\rangle} & \gate{H} & \ctrl{1} & \ctrl{1} & \ctrl{2} & \gate{H} & \meter \\
        \lstick{\rho_k} & \qw & \qswap & \gate{e^{-i t_2 K_\mu}} & \qw & \qw & \qw \\
        \lstick{\sigma_l} & \qw & \qswap \qwx & \qw & \gate{e^{-i t_1 K_\mu}} & \qw & \qw 
        }
    }
    \end{minipage}
    \caption{The quantum circuit for estimating the unitary trace overlap $\operatorname{Re}\!\left\{\Tr[e^{-i t_1 K_\mu} \rho_k e^{-i t_2 K_\mu} \sigma_l]\right\}$ required for the Hessian element. The control qubit dictates a SWAP gate followed by independent Hamiltonian time-evolutions on the two data registers.}
    \label{fig:hessian_circuit}
\end{figure}
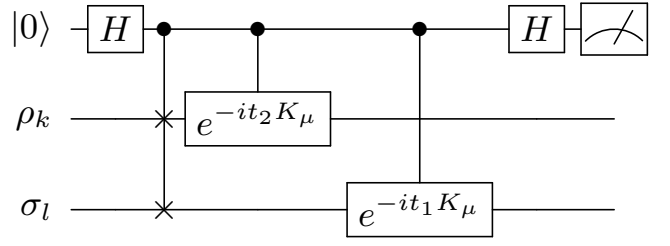

The full procedure is detailed in Algorithm~\ref{alg:hessian-est-alg}, and the details of the correctness of the algorithm are provided in Appendix~\ref{sec:proof-hessian-algorithm}.

\begin{algorithm}
\caption{Quantum algorithm for estimating the Hessian element $H_{ij}$}
\label{alg:hessian-est-alg}
\KwData{Minimum eigenvalue bound $\lambda_{\min}$, target precision $\varepsilon$, state coefficients $\{\alpha_{i,k}\}$ and $\{\alpha_{j,l}\}$, Hessian circuit in Figure~\ref{fig:hessian_circuit}}
\KwResult{Estimate of $H_{ij}$ to precision $\mathcal{O}(\varepsilon)$}
\textbf{Series Truncation:} Determine a cutoff order $M$ based on $\lambda_{\min}$ and $\varepsilon$ for the double sum over $m_1$ and $m_2$\;
\textbf{Joint Sampling:} For each pair $(m_1, m_2) \in \{1, \dots, M\}^2$, determine the required number of Joint Monte Carlo shots $N_{\operatorname{shots}}^{(m_1, m_2)}$. Draw $N_{\operatorname{shots}}^{(m_1, m_2)}$ independent tuples $(s, t_1, t_2, k, l)$ as follows:\;
\Indp
    Sample $s \sim \text{Uniform}(0, 1)$\;
    Compute scale parameters $\tau_1 = (m_1-s)/T$ and $\tau_2 = (m_2-1+s)/T$\;
    Sample evolution times $t_1 \sim p(t_1; \tau_1)$ and $t_2 \sim p(t_2; \tau_2)$ from Cauchy distributions as defined in~\eqref{eq:Cauchy_prob}\;
    Sample state indices $k$ and $l$ with probabilities $p_i(k) = |\alpha_{i,k}| / \left\|\alpha_i\right\|_1$ and $p_j(l) = |\alpha_{j,l}| / \left\|\alpha_j\right\|_1$\;
\Indm
\textbf{Quantum Evaluation:} For each sampled tuple $(s, t_1, t_2, k, l)$, execute the quantum circuit in Figure~\ref{fig:hessian_circuit}. Prepare the data registers in $\rho_k$ and $\sigma_l$, apply the controlled-SWAP and controlled-evolutions $e^{-i t_2 K_\mu}$ and $e^{-i t_1 K_\mu}$, and measure the control qubit in the $Z$-basis to obtain outcome $Z \in \{+1, -1\}$.\;
\textbf{Classical Aggregation:} For each shot, compute the unbiased estimator $Y = \left\|\alpha_{i}\right\|_1 \left\|\alpha_{j}\right\|_1 \operatorname{sgn}(\alpha_{i,k}\alpha_{j,l}) Z$. Average $Y$ over the $N_{\operatorname{shots}}^{(m_1, m_2)}$ samples, sum the averages over all $m_1, m_2 \in \{1, \dots, M\}$, and multiply the final total by $-1/T$.
\end{algorithm}

In Appendix~\ref{app:complexity_hessian}, we analyze the computational complexity required to estimate the Hessian matrix element $H_{ij}$ to a target additive error $\varepsilon$. By constructing a double Joint Monte Carlo estimator over the truncated geometric series and Cauchy distributions, we establish the total quantum runtime. The complexity depends on the spectral properties of the dual slack operator and the 1-norms of the observables.

\begin{proposition}[Hessian Estimator Complexity]
\label{prop:hessian-complexity}
Let $K_\mu > 0$ with minimum eigenvalue $\lambda_{\min} > 0$. To estimate the Hessian element $H_{ij}$ to an additive precision $\varepsilon$, the total quantum runtime of Algorithm~\ref{alg:hessian-est-alg} scales asymptotically as:
\begin{equation}
    \text{Total Runtime} = \widetilde{\mathcal{O}}\left( \left\|\alpha_i\right\|_1^3 \left\|\alpha_j\right\|_1^3 \left\|h\right\|_1 \frac{T^8}{\lambda_{\min}^9 \varepsilon^3} \right),
\end{equation}
where $\left\|\alpha_i\right\|_1$, $\left\|\alpha_j\right\|_1$, and $\left\|h\right\|_1$ are the 1-norms of the coefficient vectors for $Q_i$, $Q_j$, and $K_\mu$, respectively, and the $\widetilde{\mathcal{O}}$ notation suppresses logarithmic factors. 
\end{proposition}

This polynomial scaling guarantees that second-order continuous optimization methods, such as Newton's method or natural gradient descent, can be executed within this framework, provided the state is not driven too close to exact ground-state condensation ($\lambda_{\min} \to 0$).

\subsection{Hybrid quantum--classical algorithms for Bose--Einstein free energy minimization}
\label{sec:overall-complexity}

We are now in a position to combine the classical optimization template of Section~\ref{subsec:classical-gradient-ascent} with the quantum gradient and Hessian estimators of Sections~\ref{subsec:Quantum-algorithm-gradient} and~\ref{subsec:Quantum-algorithm-Hessian}, and to analyze the resulting end-to-end runtime needed to maximize the dual objective $f_T(\mu)$ to a target precision $\delta>0$, i.e., to produce an iterate $\mu^{J}$ satisfying $\mathbb{E}[f_T(\mu^\star) - f_T(\mu^{J})] \leq \delta$.

The quantum routines in Algorithms~\ref{alg:BE-thermal-alg} and~\ref{alg:hessian-est-alg} are Monte Carlo estimators: each call returns an unbiased estimate of a thermal trace, with a variance controlled by the Hadamard-test sampling, the Cauchy-distributed evolution time, and the state-model sampling. Consequently, classical gradient ascent and Newton's method, when fed by these estimators, become stochastic optimization algorithms~\cite{Robbins1951stochastic}. We therefore use the convergence theory for stochastic first- and second-order methods on $L$-smooth concave functions developed in the classical optimization literature~\cite{nemirovski2009robust,bubeck2015convex}.

For both first- and second-order methods, the total runtime decomposes into three factors:
\begin{enumerate}
    \item The number of classical iterations, $J$, dictated by the convergence theorem for the chosen optimization method on an $L_T$-smooth concave objective with unbiased noisy derivative oracles;
    \item The number of derivative components estimated per iteration ($c$ gradient components, plus $c^2$ Hessian elements for Newton's method);
    \item The per-call quantum runtime of Algorithm~\ref{alg:BE-thermal-alg} or Algorithm~\ref{alg:hessian-est-alg}, bounded in Proposition~\ref{prop:gradient-complexity} and Proposition~\ref{prop:hessian-complexity}, respectively.
\end{enumerate}
The next two subsections carry out the analysis of the total runtime for stochastic gradient ascent and stochastic Newton's method, respectively.

\subsubsection{First-order method: Stochastic gradient ascent}
\label{subsec:stochastic-GA}

The hybrid quantum--classical version of Algorithm~\ref{alg:classical-gradient-ascent} replaces the exact gradient component $q_{i}-\Tr[X_T(\mu)Q_i]$ with the unbiased stochastic estimate produced by Algorithm~\ref{alg:BE-thermal-alg}. The resulting procedure is summarized in Algorithm~\ref{alg:q-stoch-gradient-ascent}.

\begin{algorithm}
\caption{Hybrid quantum--classical stochastic gradient ascent for the bosonic free-energy dual}
\label{alg:q-stoch-gradient-ascent}
\KwData{Target accuracy $\delta>0$, learning rate $\eta=1/L_{T}$ with $L_{T}$ as in~\eqref{eq:smoothness-parameter}, total number of steps $J$, gradient estimator from Algorithm~\ref{alg:BE-thermal-alg}.}
\KwResult{Estimate $\tilde{f}_{T}(\mu^{J})$ of the optimal value $E(\mathcal{Q},q)$.}
Set the temperature $T>0$ small enough to guarantee an approximation error of at most $\delta/2$ via either Proposition~\ref{prop:no-BE-ent-simple-approx-err-bnd}, Remark~\ref{rem:dimension-only-bound}, or Proposition~\ref{prop:approx-error-spectral-gap}.\;
Initialize $\mu^{0}$ (for example, $\mu^{0}\leftarrow 0$) and set $J = \Theta(L_{T}\left\|\mu_{T}^{\star}\right\|^{2}/\delta)$.\;
\For{$j=1,\dots,J$}{
    \For{$i=1,\dots,c$}{
        Estimate $\Tr[X_{T}(\mu^{j-1})Q_{i}]$ to constant additive precision $\varepsilon_{j} = \Theta(\sqrt{L_T\delta/c})$ using Algorithm~\ref{alg:BE-thermal-alg}, obtaining $\widetilde{Q}_{i}^{(j-1)}$\;
        Set $g_{i}^{(j-1)}\leftarrow q_{i}-\widetilde{Q}_{i}^{(j-1)}$\;
    }
    Update $\mu^{j}\leftarrow\mu^{j-1}+\eta\,g^{(j-1)}$\;
}
Estimate $\Tr[HX_{T}(\mu^{J})]$ and each $\Tr[X_T(\mu^J)Q_i]$ ($i=1,\dots,c$) using Algorithm~\ref{alg:BE-thermal-alg} to combined additive precision $\delta/4$, obtaining $\widetilde{H}^{(J)}$ and $\widetilde{Q}^{(J)} \coloneqq (\widetilde{Q}_1^{(J)},\dots,\widetilde{Q}_c^{(J)})$\;
\Return $\tilde{f}_{T}(\mu^{J}) = \mu^{J}\cdot q+\widetilde{H}^{(J)}-\mu^{J}\cdot\widetilde{Q}^{(J)}$
\end{algorithm}

The output is an additive estimate of the optimal value $E(\mathcal{Q},q)$ of the original SDP. To bound its expected accuracy, we decompose the total error into three sources by inserting and subtracting the regularized dual optimum:
\begin{equation}
\begin{split}
\left|\tilde{f}_{T}(\mu^{J})-E(\mathcal{Q},q)\right| &\leq \underbrace{\left|\tilde{f}_{T}(\mu^{J})-f_{T}(\mu^{J})\right|}_{\text{entropy correction}} \\
& \quad + \underbrace{\left|f_{T}(\mu^{J})-F_{T}(\mathcal{Q},q)\right|}_{\text{dual suboptimality}} \\
& \quad + \underbrace{\left|F_{T}(\mathcal{Q},q)-E(\mathcal{Q},q)\right|}_{\text{approximation error}}. 
\end{split}
\label{eq:tot-error-decomp-stoch-GA}
\end{equation}
We allocate the budget $\delta/4$ to each of the three controllable contributions and a further $\delta/4$ to the final-state estimation error of $\Tr[H X_T(\mu^J)]$.
\begin{itemize}
    \item The \emph{entropy correction} is exactly bounded by $T\,S_{\mathrm{BE}}(X_T(\mu^J))$, which is controlled by the temperature schedule chosen in the first step of the algorithm.
    \item The \emph{dual suboptimality} is bounded by the standard convergence guarantees for stochastic gradient ascent on $L$-smooth concave objectives~\cite{bubeck2015convex, nemirovski2009robust}. For a fixed step size $\eta=1/L_T$ and an unbiased gradient estimator with expected variance bounded by $\sigma^2$, the expected suboptimality after $J$ steps is bounded by
    \begin{equation}
    \mathbb{E}[f_T(\mu^\star)-f_T(\mu^J)] \leq \mathcal{O}\!\left(\frac{L_T\left\|\mu_T^\star\right\|^2}{J} + \frac{\sigma^2}{L_T}\right),
    \end{equation}
    where the first term is the deterministic contraction rate of gradient ascent and the second is the persistent error floor caused by the stochastic noise. Because the algorithm estimates all $c$ gradient components independently to precision $\varepsilon_j$, the total variance is $\sigma^2 = c\varepsilon_j^2$. By choosing $J = \Theta(L_T\left\|\mu_T^\star\right\|^2/\delta)$ and a constant per-step precision $\varepsilon_j = \Theta(\sqrt{L_T \delta / c})$, both terms are strictly bounded by $\delta/8$, guaranteeing the required $\delta/4$ budget.
    \item The \emph{approximation error} is the gap between the regularized and the original SDP, controlled by the temperature schedule via Proposition~\ref{prop:no-BE-ent-simple-approx-err-bnd}, Remark~\ref{rem:dimension-only-bound}, or Proposition~\ref{prop:approx-error-spectral-gap}, depending on the available structural information.
\end{itemize}

The total runtime of Algorithm~\ref{alg:q-stoch-gradient-ascent} is therefore the product of three factors enumerated in Section~\ref{sec:overall-complexity}: (i) the iteration count $J=\Theta(L_{T}\left\|\mu_{T}^{\star}\right\|^{2}/\delta)$; (ii) the $c$ gradient components estimated per iteration; and (iii) the per-component quantum runtime
\begin{equation}
\widetilde{\mathcal{O}}\!\left(\left\|\alpha_{i}\right\|_{1}^{3}\,\left\|h\right\|_{1}\,\frac{T^{4}}{\lambda_{\min}^{5}\varepsilon_{j}^{3}}\right)
\end{equation}
established in Proposition~\ref{prop:gradient-complexity}. Multiplying the per-component runtime above by the $c$ gradient components per iteration and by the iteration count $J = \Theta(L_T\left\|\mu_T^\star\right\|^2/\delta)$, and substituting $\varepsilon_j = \Theta(\sqrt{L_T\delta/c})$, yields the end-to-end gate-count scaling
\begin{equation}
\widetilde{\mathcal{O}}\!\left(\frac{c^{5/2}\,\left\|\mu_T^\star\right\|^2\,\left\|\alpha\right\|_1^3\,\left\|h\right\|_1\,T^4}{\lambda_{\min}^5\,L_T^{1/2}\,\delta^{5/2}}\right).
\label{eq:end-to-end-SGA-runtime}
\end{equation}
This is polynomial in the number of constraints, $c$, in the input-model 1-norms $\left\|\alpha\right\|_1$ and $\left\|h\right\|_1$, in $T/\lambda_{\min}$, and in $1/\delta$. When the temperature $T$ is chosen via the spectral-gap bound (Proposition~\ref{prop:approx-error-spectral-gap}) and $\Delta^{-1}$ and $d_{0}$ are polynomial in $\ln d$, then $T$ scales polynomially in $\ln d$ and the overall runtime is polynomial in $\ln d$ (i.e., in the number of qubits), in $c$, and in $1/\delta$, modulo the dependence on the input-model 1-norms.

\subsubsection{Second-order method: Stochastic Newton's method}
\label{subsec:stochastic-Newton}

The smoothness parameter $L_T$ in~\eqref{eq:smoothness-parameter} diverges as the dual iterate approaches the boundary of the PSD cone ($\lambda_{\min}\to 0$). This ill-conditioning forces stochastic gradient ascent to take ever-smaller step sizes and inflates the iteration count $J\propto L_T$. A standard remedy in classical convex optimization is to precondition the update by the inverse Hessian of the dual objective, yielding a Newton step; the resulting method exhibits local quadratic convergence~\cite{nesterov2018Lectures, Nesterov2006Cubic}. Each Newton step replaces the gradient update $\mu^{j}\leftarrow\mu^{j-1}+\eta\,g^{(j-1)}$ in Algorithm~\ref{alg:q-stoch-gradient-ascent} by
\begin{equation}
\mu^{j}\leftarrow\mu^{j-1}-\eta\,[\nabla^{2}f_{T}(\mu^{j-1})]^{-1}\,g^{(j-1)},
\end{equation}
where the gradient is again estimated by Algorithm~\ref{alg:BE-thermal-alg} and each element of the Hessian is estimated by Algorithm~\ref{alg:hessian-est-alg}. Because both estimators are unbiased Monte Carlo samples, the resulting algorithm is a stochastic Newton's method~\cite{Roosta-Khorasani2019sub}.

\begin{algorithm}
\caption{Hybrid quantum--classical stochastic Newton's method for the bosonic free-energy dual}
\label{alg:q-stoch-newton}
\KwData{Target accuracy $\delta>0$, learning rate $\eta>0$, total number of steps $J$, gradient and Hessian estimators from Algorithms~\ref{alg:BE-thermal-alg} and~\ref{alg:hessian-est-alg}.}
\KwResult{Estimate $\tilde{f}_{T}(\mu^{J})$ of the optimal value $E(\mathcal{Q},q)$.}
Set the temperature $T>0$ small enough to guarantee an approximation error of at most $\delta/2$ via Proposition~\ref{prop:approx-error-spectral-gap} (or Remark~\ref{rem:dimension-only-bound}).\;
Initialize $\mu^{0}\in\mathbb{R}^{c}$ (for example, $\mu^{0}\leftarrow 0$) and set the total number of steps $J$ according to standard convergence theorems for stochastic Newton's method (see, e.g.,~\cite{Nesterov2006Cubic, Roosta-Khorasani2019sub}).\;
\For{$j=1,\dots,J$}{
    \For{$i,k=1,\dots,c$}{
        Estimate $[\nabla^{2}f_{T}(\mu^{j-1})]_{i,k}$ to additive precision $\varepsilon_{j}$ using Algorithm~\ref{alg:hessian-est-alg}, obtaining the entry $[\widehat{\nabla^{2}f_T}(\mu^{j-1})]_{i,k}$\;
    }
    \For{$i=1,\dots,c$}{
        Estimate $\Tr[X_{T}(\mu^{j-1})Q_{i}]$ to additive precision $\varepsilon_{j}$ using Algorithm~\ref{alg:BE-thermal-alg}, obtaining $\widetilde{Q}_{i}^{(j-1)}$\;
        Set $g_{i}^{(j-1)}\leftarrow q_{i}-\widetilde{Q}_{i}^{(j-1)}$\;
    }
    Solve the linear system $\widehat{\nabla^{2}f_T}(\mu^{j-1})\,\Lambda^{(j-1)}=g^{(j-1)}$\;
    Update $\mu^{j}\leftarrow\mu^{j-1}-\eta\,\Lambda^{(j-1)}$\;
}
Estimate $\Tr[HX_{T}(\mu^{J})]$ and each $\Tr[X_T(\mu^J)Q_i]$ ($i=1,\dots,c$) using Algorithm~\ref{alg:BE-thermal-alg} to combined additive precision $\delta/4$, obtaining $\widetilde{H}^{(J)}$ and $\widetilde{Q}^{(J)} \coloneqq (\widetilde{Q}_1^{(J)},\dots,\widetilde{Q}_c^{(J)})$\;
\Return $\tilde{f}_{T}(\mu^{J}) = \mu^{J}\cdot q+\widetilde{H}^{(J)}-\mu^{J}\cdot\widetilde{Q}^{(J)}$
\end{algorithm}

A full end-to-end runtime analysis of stochastic Newton's method on the Bose--Einstein regularized dual would require a careful variance-budget analysis to handle the propagation of error through the inversion of the noisy Hessian, which is beyond the scope of this paper; we therefore restrict ourselves here to a simple accounting of the per-iteration cost, deferring a complete end-to-end statement to future work. The total runtime of Algorithm~\ref{alg:q-stoch-newton} is the product of three factors:
\begin{enumerate}
    \item The \emph{number of Newton iterations}, $J$, which scales as $\mathcal{O}(\log\log(1/\delta))$ in the local quadratic-convergence regime, plus a polynomial number of damped steps to enter that regime~\cite{Nesterov2006Cubic}.
    \item The \emph{number of derivative components per iteration}: $c$ gradient components estimated by Algorithm~\ref{alg:BE-thermal-alg}, plus $c^2$ Hessian elements estimated by Algorithm~\ref{alg:hessian-est-alg}.
    \item The \emph{per-element quantum runtime}, given by Proposition~\ref{prop:gradient-complexity} for gradient components and, for Hessian elements, by Proposition~\ref{prop:hessian-complexity},
    \begin{equation}
    \widetilde{\mathcal{O}}\!\left(\left\|\alpha_{i}\right\|_{1}^{3}\,\left\|\alpha_{j}\right\|_{1}^{3}\,\left\|h\right\|_{1}\,\frac{T^{8}}{\lambda_{\min}^{9}\,\varepsilon^{3}}\right).
    \end{equation}
\end{enumerate}

This scaling behavior reveals a trade-off between first- and second-order methods. The advantage of stochastic Newton is that the iteration count is essentially independent of $L_T$ (and hence of $\lambda_{\min}$), which is the dominant source of slow convergence for gradient ascent near the unregularized optimum. The disadvantage is that each iteration is significantly more expensive: $c^2$ Hessian calls per step, each scaling as $T^8/\lambda_{\min}^9$, versus $c$ gradient calls scaling as $T^4/\lambda_{\min}^5$. Which method performs better in practice depends on the specific SDP, on the constants hidden in the asymptotic bounds, and on the spectral structure of $K_{\mu}$ along the optimization trajectory. A detailed end-to-end comparison remains a subject for future investigation.

\subsection{Comparison with other SDP solvers}
\label{sec:comparison}

We position the Bose--Einstein framework within the broader landscape of SDP solvers along three axes: the dependence of the duality gap on the Hilbert-space dimension $d$; the quantum primitives required; and the role of an a priori upper bound $R$ on the trace of the primal variable. The relevant comparison points are classical interior-point methods (IPMs)~\cite{Boyd2004convex,vandenberghe1996semidefinite}, the matrix multiplicative weights update method (MMW)~\cite{Tsuda2005matrix,Dhillon2007Matrix}, the quantum MMW SDP solver~\cite{vanApeldoorn2019SDP} (which we will refer to as the quantum MMW solver), and the two prior thermodynamic frameworks~\cite{liu2025sdp,liu2026sdp_fermi}.

\subsubsection{Classical IPMs and dimension dependence of  duality gap}

Classical IPMs evaluated on the central path produce a duality gap of $T\cdot d$, where $T$ is the barrier penalty parameter (often denoted $\mu$ or $1/t$ in the IPM literature), so the barrier must scale as $\mathcal{O}(\varepsilon/d)$ to reach precision $\varepsilon$. This is the scaling exploited by mature SDP solvers such as SDPA~\cite{Yamashita2010SDPA}, MOSEK~\cite{MOSEK}, and the CVX/CVXPY modeling frameworks~\cite{CVX2014,Diamond2016CVXPY}. Our dimension-only bound in Remark~\ref{rem:dimension-only-bound} recovers exactly this scaling in the worst case, confirming the consistency of the two pictures. The structural advantage of the Bose--Einstein framework lies in Proposition~\ref{prop:approx-error-spectral-gap}: whenever the dual slack operator $K_{\mu_T^\star}$ has ground-space degeneracy $d_0$ and spectral gap $\Delta>0$, the duality gap is bounded by 
\begin{equation}
T\,d_0 + (d-d_0)(\lambda_{\min}+\Delta)/(e^{(\lambda_{\min}+\Delta)/T}-1),    
\end{equation}
whose excited-state contribution decays exponentially in $\Delta/T$. In the regime $\Delta\gg T$ and $d_0\ll d$, the temperature is governed by $d_0$ and $\Delta$ instead of $d$; if $d_0$ and $\Delta^{-1}$ scale polynomially in $\ln d$, the temperature scales polynomially in $\ln d$, a substantial asymptotic advantage that the dimension-insensitive logarithmic barrier of classical IPMs cannot leverage.

\subsubsection{Classical MMW}

First-order methods based on matrix multiplicative weights and matrix exponentiated gradient updates~\cite{Tsuda2005matrix,Dhillon2007Matrix} operate natively on the trace-normalized cone and so do not directly handle the unbounded constraint $X\geq 0$ without an additional trace bound. The Bose--Einstein framework supplies precisely such a generalization: entropic regularization on the unbounded cone is a special case of a Bregman divergence (Section~\ref{subsec:BE-convexity-additivity-geometry}), and the trace-normalized Boltzmann setting is recovered when an additional constraint $\Tr[X]=1$ is imposed. We expect that algorithmic ideas from the MMW literature~\cite{Garber2016Sublinear} can be carried over to the Bose--Einstein framework using $D_{\mathrm{BE}}$ as the prox-function, although a detailed analysis remains a subject for future work.

\subsubsection{Quantum MMW solver}

The quantum MMW solver of~\cite{vanApeldoorn2019SDP}, building on the earlier quantum SDP algorithms of~\cite{BrandaoSvore2017,vanApeldoorn2017quantumSDP,Brandao2019QuantumSDP} and complemented more recently by quantum interior-point variants~\cite{KerenidisPrakash2020QIPM,AugustinoNannicini2023QIPM}, is based on a quantum implementation of MMW with block-encoded inputs and Gibbs-state preparation as primitives. For SDPs with $n$-dimensional matrix variables and $m$ constraints, it achieves a query complexity of $\widetilde{\mathcal{O}}((\sqrt{m} + \sqrt{n}\gamma)\alpha\gamma^4)$, where $\gamma = Rr/\varepsilon$. Here, $R$ and $r$ are the dual and primal feasibility radii, respectively, and $\alpha$ is the input-model normalization~\cite[Theorem 17]{vanApeldoorn2019SDP}.

To compare the bound above with our bound for the stochastic gradient ascent algorithm shown in~\eqref{eq:end-to-end-SGA-runtime}, one must account for the difference in cost metrics: the MMW solver measures oracle queries, whereas we report total quantum gate runtime. However, even when converting MMW queries to gate counts (which may introduce only a $\mathrm{polylog}(n)$ overhead for structured inputs), the MMW solver retains an explicit $\sqrt{n}$ dependence. In contrast, our framework's gate count has no explicit polynomial dependence on the dimension $d \equiv n$; rather, $d$ may only enter implicitly through the input-model 1-norms ($\left\|\alpha_i\right\|_1, \left\|h\right\|_1$) and the temperature $T$. 
This leads to two distinct comparative regimes. In the worst case (arbitrary inputs, no spectral gap, $\lambda_{\min}\to 0$), both frameworks scale polynomially in $d$, and the MMW solver is likely more robust because its complexity does not rely on the spectral structure of the dual slack operator. However, in the favorable regime where the inputs are structured (with $\left\|\alpha_i\right\|_1, \left\|h\right\|_1 = \mathrm{polylog}(d)$), the slack operator is gapped (with $d_0, \Delta^{-1} = \mathrm{polylog}(d)$), and $\lambda_{\min}$ is uniformly bounded away from zero, our framework's end-to-end gate count is $\mathrm{polylog}(d) \cdot c^{5/2} \cdot \mathrm{poly}(1/\delta)$, while the quantum MMW solver retains its $\sqrt{n}=\sqrt{d}$ factor. Our framework therefore achieves a polynomial-in-$\log d$ runtime -- polynomial in the qubit count -- against the $\sqrt{n}$ scaling of the quantum MMW solver, at the price of a worse polynomial dependence on the number of constraints $c \equiv m$ (specifically, $c^{5/2}$ versus $\sqrt{m}$).

Beyond the dimension comparison, the two frameworks differ in two further ways. The first concerns the quantum primitives required: the quantum MMW solver assumes block-encoded access to the input matrices and quantum Gibbs-state preparation, while Algorithms~\ref{alg:BE-thermal-alg} and~\ref{alg:hessian-est-alg} of the present work instead use only Hamiltonian simulation, the Hadamard test, and classical Cauchy sampling. The second difference concerns the dependence on an a priori upper bound on the primal trace: the quantum MMW solver's runtime scales as $R^4$, where $R$ is such a bound, and when $R$ is not known it must be located by binary search at an additional multiplicative cost, whereas the Bose--Einstein framework requires no such a priori trace bound: the variable $X\geq 0$ is unbounded in the optimization domain, and the runtime is controlled instead by the spectral parameters $\lambda_{\min},\Delta,d_0$ of the dual slack operator at the regularized optimum, sidestepping the dependence on $R$ entirely. We note here that in our framework, as $T\to 0$, complementary slackness forces $\lambda_{\min}\to 0$, so that obtaining a fully comparable $\varepsilon$-complexity with the $1/\varepsilon^4$ scaling of the quantum MMW solver would require characterizing the rate at which $\lambda_{\min}$ vanishes along the optimization trajectory -- an open problem to which we return in Section~\ref{sec:conclusion}.

\subsubsection{Prior thermodynamic frameworks}

The Bose--Einstein framework completes the thermodynamic trilogy of SDP solvers, complementing the Boltzmann framework~\cite{liu2025sdp} (which handles trace-one density operators) and the Fermi--Dirac framework~\cite{liu2026sdp_fermi} (which handles measurement operators bounded by the identity). The three frameworks share a common skeleton: a physically motivated entropy, a free-energy regularization, an unconstrained concave dual, gradient and Hessian formulas as thermal expectation values, and hybrid quantum--classical algorithms based on Hadamard tests and Hamiltonian simulation. They differ in the operator domain, the form of the optimal primal variable (a quantum Boltzmann state, a Fermi--Dirac thermal measurement, or a Bose--Einstein thermal operator), and the scaling of the smoothness parameter $L_T$. The Bose--Einstein quantum relative entropy of Section~\ref{sec:be-relative-entropy} is the canonical divergence on the unbounded cone, generalizing the von Neumann and Fermi--Dirac relative entropies (which are the Bregman divergences generated by $-S(\rho)$ on the set of density operators and by $-S_{\mathrm{FD}}(M)$ on $[0,I]$, respectively).

A specific point worth emphasizing is the relationship between the Bose--Einstein framework and the Boltzmann framework of~\cite{liu2025sdp}. Both target general SDPs over PSD operators, but the Boltzmann framework only handles trace-normalized variables; to apply it to a standard SDP, one must first guess an upper bound $R$ on $\Tr[X]$, rescale the problem so that the variable lives on the set of density operators, and then solve. If $R$ is not known, a binary search over $R$ adds a multiplicative cost analogous to that incurred by the quantum MMW solver. The Bose--Einstein framework, in contrast, regularizes on the unbounded cone; thus it does not require an explicit trace bound, although compactness assumptions on the feasible set are still necessary for the uniform entropy-based approximation bounds. The trade-off for this independence is that the relevant complexity parameter is no longer the trace bound but the spectral structure ($\lambda_{\min}$, $\Delta$, $d_0$) of the dual slack operator at the regularized optimum, which translates into a runtime that is favorable when this spectral structure is benign and unfavorable when it is not.

\subsubsection{Summary}

To summarize, the Bose--Einstein framework may be viewed as a complementary approach to both the quantum MMW solver~\cite{vanApeldoorn2019SDP} and the prior Boltzmann thermodynamic framework~\cite{liu2025sdp}, rather than a direct competitor to either. While these frameworks all address general SDPs, they trade off different complexity parameters: our framework dispenses with the binary search over the trace bound $R$ required by the other two, replacing it with a dependence on the spectral structure of the dual slack operator. It also uses arguably simpler quantum primitives, that is Hadamard test plus Hamiltonian simulation rather than block-encoded Gibbs-state preparation. For SDPs with structured input matrices, a small ground-space degeneracy, and a non-vanishing spectral gap at the regularized optimum, the Bose--Einstein framework may be the natural choice; for general worst-case SDPs without these structural assumptions, the quantum MMW solver may be more efficient.

\section{Bose--Einstein quantum relative entropy}
\label{sec:be-relative-entropy}

In Section~\ref{sec:general_sdp}, we demonstrated that regularizing a semidefinite program with the Bose--Einstein entropy naturally maps the optimal primal variable to a thermal operator of independent bosonic modes. In Appendix~\ref{app:proof_lemma_rel_entropy}, we show that this proof relies on a specific divergence: the Bose--Einstein quantum relative entropy, $D_{\mathrm{BE}}(X\|Y)$.

In the standard quantum information literature, the distinguishability of two quantum states is commonly measured by the Umegaki relative entropy \cite{Umegaki1962}, defined as $D(X\|Y) \coloneqq \Tr[X\ln X - X\ln Y]$. However, this standard divergence is mathematically tailored for normalized density matrices ($\Tr[\rho]=1$). When applied generally to the unbounded positive semidefinite cone ($X \geq 0$), it can evaluate to a strictly negative number, violating the axiom of positivity of a distance metric. While one can patch the standard divergence to restore positivity, typically by adopting the generalized quantum relative entropy $D(X\|Y) = \Tr[X \ln X - X \ln Y - X + Y]$ derived as the matrix Bregman divergence of $\Tr[X \ln X - X]$~\cite[Section 2.2]{Tsuda2005matrix}, we introduce the Bose--Einstein relative entropy as a general-purpose divergence natively designed for unbounded positive operators. As a canonical Bregman divergence generated by the strictly convex trace function $-S_{\mathrm{BE}}(X)$, its mathematical structure, specifically the inclusion of the $(X+I)$ terms, ensures it remains non-negative, well-defined, and geometrically consistent across the entire positive semidefinite cone. This makes it a natural mathematical choice not only for our thermodynamic framework, but for general information-theoretic and continuous optimization tasks involving unnormalized operators.

In this section, we define the Bose--Einstein relative entropy and establish its fundamental properties, including faithfulness, unitary invariance, and spectral expansion. We also analyze its convexity properties and discuss how its geometry naturally accommodates continuous optimization algorithms over the unbounded positive semidefinite cone.

\subsection{Definition and properties}
\label{subsec:BE-properties}

\begin{definition}[Bose--Einstein quantum relative entropy]
\label{def:BE-rel-entropy}
For positive semidefinite operators $X, Y \geq 0$, the Bose--Einstein quantum relative entropy is defined as
\begin{equation}
    D_{\mathrm{BE}}(X\|Y) \coloneqq -S_{\mathrm{BE}}(X) + \Tr[(X+I)\ln(Y+I) - X\ln Y], \label{eq:BE_rel_entropy_main}
\end{equation}
where $S_{\mathrm{BE}}(X) = \Tr[(X+I)\ln(X+I) - X\ln X]$. Henceforth, we refer to it simply as the Bose--Einstein relative entropy whenever the context is unambiguous.
\end{definition}

\begin{remark}[Decomposition into Umegaki relative entropies]
\label{rem:BE-umegaki-decomposition}
A direct regrouping of the terms in~\eqref{eq:BE_rel_entropy_main} expresses $D_{\mathrm{BE}}$ as a difference of two Umegaki-type relative entropies:
\begin{equation}\label{eq:BE-umegaki-decomp}
    D_{\mathrm{BE}}(X\|Y) = D(X\|Y) - D(X+I\,\|\,Y+I),
\end{equation}
where $D(A\|B) \coloneqq \Tr[A\ln A - A\ln B]$ is the natural extension of the Umegaki relative entropy to unnormalized $A, B \geq 0$. Although either term on the right can be negative on unnormalized inputs -- as noted at the beginning of this section -- their difference is always non-negative (due to the Bregman-divergence structure of $-S_{\mathrm{BE}}$, as shown in the following Proposition~\ref{prop:BE-faithfulness}). The shift $X \mapsto X+I$ encodes the bosonic occupation-fluctuation factor $n(n+1)$, and it is the counterpart of the Fermi--Dirac decomposition $D_{\mathrm{FD}}(M\|N) = D(M\|N) + D(I-M\,\|\,I-N)$ (see~\cite[Definition 27]{liu2026sdp_fermi}), in which $I-M$ plays the role of the ``hole'' density.
\end{remark}

$D_{\mathrm{BE}}$ satisfies the foundational requirements of a divergence as described in the following propositions.

\begin{proposition}[Faithfulness and Support]
\label{prop:BE-faithfulness}
For all $X, Y \geq 0$, the Bose--Einstein relative entropy is non-negative: $D_{\mathrm{BE}}(X\|Y) \geq 0$, with equality if and only if $X=Y$. Furthermore, $D_{\mathrm{BE}}(X\|Y)$ is finite if and only if $\operatorname{supp}(X) \subseteq \operatorname{supp}(Y)$.
\end{proposition}
\begin{proof}
See Appendix~\ref{app:proof_BE_faithfulness}. The non-negativity follows directly from its structure as a Bregman divergence (see Remark~\ref{rem:bregman-formulation}) via Klein's inequality.  The support condition arises from the singular nature of the $-X \ln Y$ term that occurs when the support condition does not hold.
\end{proof}

\begin{proposition}[Unitary Invariance]
\label{prop:BE-isometric}
For every unitary $U$ (i.e., $U^\dagger U = UU^\dagger = I$), the divergence is invariant:
\begin{equation}
    D_{\mathrm{BE}}(UXU^\dagger \| UYU^\dagger) = D_{\mathrm{BE}}(X\|Y).
\end{equation}
\end{proposition}
\begin{proof}
See Appendix~\ref{app:proof_BE_isometric}. 
\end{proof}

Unitary invariance guarantees that the metric depends only on the intrinsic mode occupations and their overlaps, independent of the chosen physical representation (e.g., spatial modes versus momentum modes).

Furthermore, this divergence decomposes into the classical scalar relative entropy, weighted by the quantum transition probabilities between the eigenbases of the two operators.

\begin{proposition}[Spectral Expansion]
\label{prop:BE-spectral}
Let $X$ and $Y$ have spectral decompositions $X = \sum_i x_i |\psi_i\rangle\!\langle\psi_i|$ and $Y = \sum_j y_j |\phi_j\rangle\!\langle\phi_j|$. The Bose--Einstein relative entropy defined in~\eqref{eq:BE_rel_entropy_main} expands as
\begin{equation}
    D_{\mathrm{BE}}(X\|Y) = \sum_{i,j} d_{\mathrm{BE}}(x_i \| y_j) \left| \langle \psi_i | \phi_j \rangle \right|^2, \label{eq:BE-spectral-expansion}
\end{equation}
where $d_{\mathrm{BE}}(x\|y)$ is the classical scalar Bose--Einstein divergence, defined as
\begin{equation}
    d_{\mathrm{BE}}(x\|y) \coloneqq x \ln \!\left(\frac{x}{y}\right) + (x+1) \ln\!\left( \frac{y+1}{x+1}\right).\label{eq:scalar_BE_rel_entr}
\end{equation}
\end{proposition}

\begin{proof}
See Appendix~\ref{app:proof_BE_spectral}.
\end{proof}

\subsection{Convexity, additivity, and algorithmic geometry}
\label{subsec:BE-convexity-additivity-geometry}

The Bose--Einstein relative entropy defined in~\eqref{eq:BE_rel_entropy_main} belongs to the class of matrix Bregman divergences~\cite{Dhillon2007Matrix}.

\begin{remark}[Bregman Divergence Formulation]
\label{rem:bregman-formulation}
The Bose--Einstein relative entropy coincides with the matrix Bregman divergence generated by the negative bosonic entropy, $F(X) = -S_{\mathrm{BE}}(X)$, where $S_{\mathrm{BE}}(X)$ is defined in~\eqref{eq:BE-entropy-def}. Using the gradient $\nabla F(Y) = \ln Y - \ln(Y+I)$, the divergence takes the canonical Bregman form:
\begin{equation}
    D_{\mathrm{BE}}(X\|Y) = F(X) - F(Y) - \Tr[\nabla F(Y)(X-Y)].
\end{equation}
\end{remark}

This mathematical classification has consequences for both its physical behavior and its utility in continuous optimization. In standard projected gradient descent over the positive semidefinite cone, the variable update relies on minimizing the Euclidean distance $\|X - Y\|_2^2$. This isotropic metric often pushes the algorithm outside the feasible region, requiring computationally expensive projection subroutines to restore positivity~\cite{Boyd2004convex,bubeck2015convex}. 

To overcome this computational bottleneck, the optimization can transition into the Quantum Mirror Descent framework by substituting the Euclidean proximity penalty with the Bose--Einstein relative entropy~\cite{bubeck2015convex}. By using $D_{\mathrm{BE}}(X\|X_k)$, where $X_k$ denotes the current iterate of the optimization, as the distance-generating function, the optimization operates in a non-Euclidean geometry that naturally mirrors the feasible space.

This algorithmic geometry underpins the thermodynamic framework developed in Section~\ref{sec:general_sdp}. In optimization theory, it is a well-established equivalence that regularizing a primal objective with a strictly convex function is mathematically identical to using its generated Bregman divergence as the proximity penalty in Mirror Descent~\cite{Beck2003Mirror, bubeck2015convex}. 
By regularizing the primal semidefinite program with the Bose--Einstein entropy, our framework inherently adopts this non-Euclidean geometry. Consequently, solving the regularized dual problem naturally executes Quantum Mirror Descent: the updates natively respect the boundary constraints, mapping directly to the valid Bose--Einstein thermal operators $X_T(\mu)$ derived earlier~\cite{Nesterov2009PrimalDual}. This circumvents the need for hard spectral projections, structurally adapting the algorithm to the exact physical geometry of the unbounded optimization space and reducing the computational overhead per iteration.

While $D_{\mathrm{BE}}(X\|Y)$ is a strictly convex Bregman divergence, this strict convexity applies only to its first argument; it does not satisfy the stronger property of joint convexity in both arguments. This feature has consequences for the data-processing inequality, a fundamental property of distinguishability measures in quantum information theory~\cite{Wilde2017quantum}. Because monotonicity under quantum channels requires joint convexity, the failure of joint convexity precludes the generalized data-processing inequality.

\begin{proposition}[Convexity and Data Processing]
\label{prop:BE-convexity-main}
$D_{\mathrm{BE}}(X\|Y)$ is strictly convex in its first argument $X$. However, it is not jointly convex in $(X,Y)$. Consequently, it does not satisfy the generalized data-processing inequality under arbitrary Completely Positive Trace-Preserving (CPTP) maps.
\end{proposition}
\begin{proof}
See Appendix~\ref{app:proof_BE_convexity}. We prove the failure of joint convexity by demonstrating that the Hessian of the underlying scalar function $d_{\mathrm{BE}}(x\|y)$ defined in~\eqref{eq:scalar_BE_rel_entr} is indefinite. As detailed in the appendix, because monotonicity under the partial trace requires joint convexity, the absence of joint convexity causes the violation of the generalized data-processing inequality.
\end{proof}

The lack of joint convexity and the corresponding failure of the data-processing inequality highlight the distinction between normalized density matrices and unbounded positive operators. The Umegaki relative entropy satisfies the data-processing inequality because quantum channels monotonically contract the distinguishability of strict probability distributions ($\Tr[\rho]=1$). In our framework, however, $X$ and $Y$ represent unnormalized, macroscopic quantities, whose distinguishability is not guaranteed to contract under arbitrary quantum operations: physically, macroscopic occupation numbers can grow or shrink under physical channels such as amplifiers or attenuators, so contraction in general is not required. A restricted form of monotonicity does hold, however, under the affine maps that model bosonic Gaussian channels; we develop this in Section~\ref{subsec:BE-affine-monotonicity}.

Correspondingly, the divergence is not invariant under the standard tensor product ($D_{\mathrm{BE}}(X \otimes Z \| Y \otimes Z) \neq D_{\mathrm{BE}}(X\|Y)$). Instead, the natural geometric operation for combining independent variables -- whether they represent disjoint block-diagonal matrices in a semidefinite program or uncoupled bosonic harmonic oscillators -- is the direct sum. Under this operation, the Bose--Einstein relative entropy decomposes additively.

\begin{proposition}[Additivity under Direct Sums]
\label{prop:BE-additivity}
For independent uncoupled systems, the divergence decomposes additively:
\begin{equation}
    D_{\mathrm{BE}}(X \oplus Z \| Y \oplus W) = D_{\mathrm{BE}}(X\|Y) + D_{\mathrm{BE}}(Z\|W).
\end{equation}
\end{proposition}
\begin{proof}
See Appendix~\ref{app:proof_BE_additivity}. This follows directly from the block-diagonal structure of the matrix logarithm.
\end{proof}

\subsection{Monotonicity under affine maps and bosonic Gaussian channels}
\label{subsec:BE-affine-monotonicity}

The failure of general data processing established in Proposition~\ref{prop:BE-convexity-main} does not preclude a restricted monotonicity. We show here that $D_{\mathrm{BE}}$ contracts under a physically natural class of affine maps that includes the attenuator, amplifier, and additive-noise channels of bosonic Gaussian quantum information theory. The argument rests on the Bregman integral representation of the scalar divergence $d_{\mathrm{BE}}$ in~\eqref{eq:scalar_BE_rel_entr}.

\begin{lemma}[Integral representation]
\label{lem:BE-integral-representation}
For all $x, y > 0$,
\begin{equation}\label{eq:BE-integral-rep}
    d_{\mathrm{BE}}(x\|y) = (x-y)^2 \int_0^1 \frac{1-t}{z_t(z_t+1)}\,dt, 
\end{equation}
where
\begin{equation}
    z_t \coloneqq y + t(x-y).
\end{equation}
\end{lemma}

\begin{proof}
See Appendix~\ref{app:proof_BE_affine_monotonicity}. This is the standard Bregman integral form applied to the generating function $f(x) = x \ln x - (x+1)\ln(x+1)$, whose second derivative is $f''(x) = 1/[x(x+1)]$.
\end{proof}

\begin{theorem}[Affine monotonicity]
\label{thm:BE-affine-monotonicity}
Let $a, b \geq 0$ satisfy
\begin{equation}\label{eq:BE-affine-condition}
    2b + 1 \geq a.
\end{equation}
Then for all $X, Y \geq 0$,
\begin{equation}\label{eq:BE-affine-monotonicity-bnd}
    D_{\mathrm{BE}}(X\|Y) \geq D_{\mathrm{BE}}(aX + bI\,\|\,aY + bI).
\end{equation}
\end{theorem}
\begin{proof}
See Appendix~\ref{app:proof_BE_affine_monotonicity}. The scalar inequality $d_{\mathrm{BE}}(x\|y) \geq d_{\mathrm{BE}}(ax+b\|ay+b)$ follows from comparing the integrands in~\eqref{eq:BE-integral-rep} pointwise in $t$, and lifts to the operator level via the spectral expansion of Proposition~\ref{prop:BE-spectral}, since the affine map $Z \mapsto aZ + bI$ preserves the eigenbases of $X$ and $Y$ and shifts their eigenvalues by $z \mapsto az + b$.
\end{proof}

Three specializations of Theorem~\ref{thm:BE-affine-monotonicity} correspond to the canonical single-mode bosonic Gaussian channels:
\begin{itemize}
    \item \emph{Attenuator} ($a = \eta$, $b = (1-\eta)N$, with $\eta \in [0,1]$ and $N \geq 0$): models a lossy channel with thermal environment of mean photon number $N$.
    \item \emph{Amplifier} ($a = G$, $b = (G-1)(N+1)$, with $G \geq 1$ and $N \geq 0$): models a phase-insensitive amplifier with thermal noise.
    \item \emph{Additive noise} ($a = 1$, $b = N$, with $N \geq 0$): models a classical additive-noise channel.
\end{itemize}
In each case the algebraic condition~\eqref{eq:BE-affine-condition} is satisfied; the explicit verifications of $2b+1-a \geq 0$ are carried out in Appendix~\ref{app:proof_BE_affine_monotonicity}.

The Bose--Einstein relative entropy thus exhibits a dichotomy: it fails the general data-processing inequality (Proposition~\ref{prop:BE-convexity-main}) but contracts under the affine maps that govern the lossy, amplifying, and noisy transformations of bosonic occupation numbers. This identifies $D_{\mathrm{BE}}$ as the natural divergence for tracking the distinguishability of bosonic-mode populations across the bosonic Gaussian channels of continuous-variable quantum information.

\subsection{Bose--Einstein Fisher information matrix}
\label{subsec:BE-fisher-info}

In information geometry, the distance between infinitesimally close parameterized distributions is governed by the Fisher information matrix~\cite{Amari2000methods}. Using the Bose--Einstein relative entropy, we can formally define the Bose--Einstein Fisher information matrix for the family of thermal operators $X_T(\mu)$ parameterized by the dual vector $\mu \in \mathbb{R}^c$:
\begin{equation}
    \left[I(\mu)\right]_{i,j} \coloneqq \left.\frac{\partial^{2}}{\partial\varepsilon_{i}\partial\varepsilon_{j}}D_{\mathrm{BE}}(X_{T}(\mu)\|X_{T}(\mu+\varepsilon))\right|_{\varepsilon=0}.
\end{equation}
This metric is connected to the algorithmic geometry discussed in the previous subsection, establishing the geometric curvature of the optimization space.

\begin{proposition}[Bose--Einstein Fisher Information]
\label{prop:BE-fisher-info}
The elements of the Bose--Einstein Fisher information matrix evaluate exactly to:
\begin{align}
    \left[I(\mu)\right]_{i,j} &= \frac{1}{T^{2}}\int_{0}^{1}ds\,\Tr\!\left[X_{T}(\mu,s)Q_{i}X_{T}(\mu,1-s)Q_{j}\right] \\
    &= \frac{1}{T^{2}}\operatorname{Re}\!\left\{\Tr\!\left[X_{T}(\mu)\Phi_{\mu}(Q_{i})\left(X_{T}(\mu)+I\right)Q_{j}\right]\right\}. \label{eq:BE-alt-exp-Fisher-info}
\end{align}
Consequently, the Fisher information matrix is strictly proportional to the negative Hessian of the dual objective function: $I(\mu) = -\frac{1}{T} \nabla^2 f_T(\mu)$.
\end{proposition}
\begin{proof}
See Appendix~\ref{sec:proof-be-fisher}.
\end{proof}

The presence of the term $X_T(\mu)(X_T(\mu)+I)$ in the Fisher information matrix, as opposed to the $M_T(\mu)(I-M_T(\mu))$ term found in the Fermi--Dirac Fisher metric~\cite[Proposition 34]{liu2026sdp_fermi}, is a direct manifestation of the different statistics of bosons and fermions.

\section{Conclusion}
\label{sec:conclusion}

We have established a thermodynamic framework for solving general SDPs by interpreting unbounded positive semidefinite operators as the expected occupation numbers of independent bosonic modes. By regularizing the standard SDP with the Bose--Einstein entropy, the problem maps to a bosonic free-energy minimization at strictly positive temperature, with optimal primal variable taking the explicit form of a Bose--Einstein thermal operator $X_{T}(\mu) = (e^{(H-\mu\cdot Q)/T}-I)^{-1}$. This completes a thermodynamic trilogy of SDP frameworks: the Boltzmann framework~\cite{liu2025sdp} for trace-normalized variables, the Fermi--Dirac framework~\cite{liu2026sdp_fermi} for measurement operators bounded by the identity, and the Bose--Einstein framework of the present paper for unbounded positive operators. The three frameworks share a common skeleton consisting of a physically motivated entropy, a free-energy regularization, an unconstrained concave dual, and gradient/Hessian formulas as thermal expectations, but differ in the operator domain they handle.

In this work we showed that the Bose--Einstein framework allows us to obtain the following results: (i) an approximation-error bound (Proposition~\ref{prop:approx-error-spectral-gap}) that depends only on the ground-space degeneracy $d_{0}$ and the spectral gap $\Delta$ of the dual slack operator, thereby decoupling the duality gap from the global Hilbert-space dimension $d$ in the gapped regime; (ii) the introduction of the Bose--Einstein quantum relative entropy $D_{\mathrm{BE}}(X\|Y)$ as a canonical Bregman divergence on the unbounded positive semidefinite cone, together with a restricted monotonicity (Theorem~\ref{thm:BE-affine-monotonicity}) under the affine maps modeling bosonic Gaussian channels; and (iii) hybrid quantum--classical algorithms (Algorithms~\ref{alg:BE-thermal-alg} and~\ref{alg:hessian-est-alg}) that estimate the gradient and Hessian of the dual objective using only Hamiltonian simulation, Hadamard tests, and classical Cauchy sampling. These quantum primitives avoid quantum Gibbs-state preparation typical of other quantum SDP solvers, at the cost of a polynomial dependence on the spectral parameter $T/\lambda_{\min}$.

Several natural directions arise as amenable paths of further investigation. An open question that remains is the relationship between the temperature $T$ and the smallest eigenvalue $\lambda_{\min}$ of the dual slack operator $K_{\mu_{T}^{\star}}$ along the optimization trajectory. Complementary slackness forces $\lambda_{\min}\to 0$ at the unregularized optimum, and our runtime bounds for both Algorithms~\ref{alg:BE-thermal-alg} and~\ref{alg:hessian-est-alg} depend inverse-polynomially on $\lambda_{\min}$. A precise characterization of how $\lambda_{\min}$ depends on $T$ in the small-$T$ limit, and on the structural parameters $d_{0}$ and $\Delta$ of the unregularized optimum, would close the gap to a clean end-to-end $\varepsilon$-complexity statement. We expect this characterization to depend on the specific class of SDP under consideration and to be amenable to a perturbative analysis around $\mu^{\star}$. 

On the theoretical side, the Bose--Einstein quantum relative entropy is a stand-alone information-theoretic primitive whose properties on the unbounded cone deserve further investigation; potential applications include continuity bounds for non-Gaussian states in continuous-variable quantum information, entanglement quantification for unnormalized operators, and matrix-Bregman optimization in classical statistics. 

On the algorithmic side, the Bose--Einstein Mirror Descent perspective discussed in Section~\ref{subsec:BE-convexity-additivity-geometry} suggests carrying over ideas from the classical MMW literature into our framework.  These include width reduction~\cite{Garber2016Sublinear}, where the polynomial dependence on a problem-specific quantity is reduced by carefully balancing primal and dual updates, and variance-reduction techniques~\cite{Allen-Zhu2017variance,Carmon2020Variance}, where past stochastic gradients are reused to lower the variance of the Monte Carlo estimators. Both ideas have the potential to improve the polynomial dependence on $T/\lambda_{\min}$ in our runtime bounds. 

On the application side, identifying classes of SDPs that fit natively into this unbounded-cone setting -- such as bosonic moment problems, quantum optical hypothesis testing, or the optimization of molecular vibrational modes in quantum chemistry -- would showcase the broad practical relevance of the framework.

Finally, numerically simulating the proposed hybrid quantum--classical gradient and Newton methods represents a natural next step to validate their performance empirically, following a similar methodology to the numerical investigations recently conducted for the analogous Boltzmann framework~\cite{Minervini_2026}.

Overall, the Bose--Einstein framework opens a new avenue for both theoretical and practical advances in quantum optimization, with rich connections to quantum information theory, statistical mechanics, and algorithm design.

\section*{Acknowledgments}

MM and MMW
acknowledge support from the Cornell School of Electrical and Computer Engineering. NL acknowledges funding from the
Science and Technology Commission of Shanghai Municipality (STCSM)
grant no.~24LZ1401200 (21JC1402900), NSFC grants no.~12471411 and
no.~12341104, the Shanghai Jiao Tong University 2030 Initiative, the Shanghai Pilot Program for Basic Research, 
and the Fundamental Research Funds for the Central Universities. 

\bibliography{references}

\appendix
\onecolumngrid

\section{Bose--Einstein entropy}
\label{app:be-entropy}

\subsection{Thermal-state interpretation}
\label{app:be_entropy_thermal_state}

In this appendix, we show that the Bose--Einstein entropy $S_{\mathrm{BE}}(X)$ defined in~\eqref{eq:BE-entropy-def} coincides with the von Neumann entropy of a non-interacting bosonic thermal state whose mean occupation numbers, in the eigenmodes of $X$, are the eigenvalues of $X$. This justifies the physical interpretation invoked throughout Section~\ref{sec:free_energy_min}.

\paragraph*{Step 1: Maximum-entropy distribution at fixed mean occupation.}
Consider the constrained optimization problem
\begin{equation}\label{eq:max-ent-problem}
\max_{\{p(n)\}_{n\geq0}}\,\left\{-\sum_{n=0}^{\infty} p(n)\ln p(n) \;:\;\sum_{n=0}^{\infty} p(n) = 1,\;\;\sum_{n=0}^{\infty} n\,p(n) = x\right\},
\end{equation}
where $\{p(n)\}_{n\geq0}$ ranges over discrete probability distributions on the bosonic occupation numbers $n\in\{0,1,2,\dots\}$ and $x\geq0$ is the prescribed mean occupation. Introducing Lagrange multipliers $\alpha$ and $\beta$ for the normalization and mean-occupation constraints, respectively, the Lagrangian is
\begin{equation}
\mathcal{L}(p,\alpha,\beta) = -\sum_{n=0}^{\infty}p(n)\ln p(n) - \alpha\!\left[\sum_{n=0}^{\infty}p(n) - 1\right] - \beta\!\left[\sum_{n=0}^{\infty}n\,p(n) - x\right].
\end{equation}
Setting $\partial\mathcal{L}/\partial p(n) = 0$ for every $n\geq0$ yields the stationarity condition $-\ln p(n) - 1 - \alpha - \beta n = 0$, so the optimal distribution has the exponential form
\begin{equation}
p(n) = C\,q^{n}, \qquad C \coloneqq e^{-(1+\alpha)},\quad q \coloneqq e^{-\beta}.
\end{equation}
The normalization constraint $\sum_{n=0}^{\infty}C\,q^n = C/(1-q) = 1$ gives $C = 1-q$, and the mean-occupation constraint $\sum_{n=0}^{\infty} n\,(1-q)\,q^n = q/(1-q) = x$ gives $q = x/(x+1)$ and hence $1-q = 1/(x+1)$. The optimal distribution is therefore the geometric distribution
\begin{equation}\label{eq:geom-dist}
p(n) = (1-q)\,q^{n}, \qquad q = \frac{x}{x+1},\quad 1-q = \frac{1}{x+1},
\end{equation}
on the bosonic occupation numbers. Strict concavity of the Shannon entropy in $p$ confirms that this stationary point is the unique global maximum.

\paragraph*{Step 2: Single-mode bosonic thermal state.}
The maximum-entropy distribution~\eqref{eq:geom-dist} can be realized as the diagonal of a quantum density operator on the bosonic Fock space $\{|n\rangle\}_{n=0}^{\infty}$:
\begin{equation}\label{eq:single-mode-thermal}
\rho_{x} \coloneqq \sum_{n=0}^{\infty} p(n)\,|n\rangle\!\langle n|.
\end{equation}
This is the single-mode bosonic thermal state with mean occupation $x$. Its von Neumann entropy equals the Shannon entropy of $p$:
\begin{align}
S(\rho_{x})
&= -\sum_{n=0}^{\infty} p(n)\ln p(n) \\
&= -\sum_{n=0}^{\infty} (1-q)\,q^{n}\bigl[\ln(1-q) + n\ln q\bigr] \\
&= -\ln(1-q) - x\,\ln q \\
&= \ln(x+1) + x\bigl[\ln(x+1) - \ln x\bigr] \\
&= (x+1)\ln(x+1) - x\ln x \;=\; g(x),
\end{align}
where in the third equality we used the normalization and mean-occupation identities from Step~1, and in the fourth we substituted $1-q=1/(x+1)$ and $q=x/(x+1)$. Hence the scalar bosonic entropy $g(x)$ defined in~\eqref{eq:bosonic-entropy-def} is exactly the von Neumann entropy of the single-mode bosonic thermal state with mean occupation $x$.

\paragraph*{Step 3: Multimode product thermal state.}
Let $X = \sum_{j=1}^{d} x_{j}\,|\phi_{j}\rangle\!\langle\phi_{j}|$ be the eigendecomposition of a positive semidefinite operator on a $d$-dimensional Hilbert space, with eigenvalues $x_{j} \geq 0$. Associating to each eigenmode $|\phi_{j}\rangle$ an independent bosonic mode, define the product state
\begin{equation}\label{eq:multimode-thermal}
\rho_{X} \coloneqq \bigotimes_{j=1}^{d} \rho_{x_{j}},
\end{equation}
where $\rho_{x_{j}}$ is the single-mode thermal state of~\eqref{eq:single-mode-thermal} with mean occupation $x_{j}$. By construction, $\rho_{X}$ is a non-interacting bosonic thermal state whose mean occupation in the $j$-th eigenmode of $X$ equals the eigenvalue $x_{j}$. Additivity of the von Neumann entropy over tensor products then gives
\begin{equation}\label{eq:multimode-entropy}
S(\rho_{X}) = \sum_{j=1}^{d} S(\rho_{x_{j}}) = \sum_{j=1}^{d} g(x_{j}).
\end{equation}

\paragraph*{Step 4: Matching $S_{\mathrm{BE}}(X)$.}
The Bose--Einstein entropy $S_{\mathrm{BE}}(X) = \Tr[(X+I)\ln(X+I) - X\ln X]$ is a trace function applied to $X$ and therefore depends only on the eigenvalues of $X$:
\begin{equation}\label{eq:S_BE-eigvalue-sum}
S_{\mathrm{BE}}(X) = \sum_{j=1}^{d} g(x_{j}).
\end{equation}
Combining~\eqref{eq:multimode-entropy} and~\eqref{eq:S_BE-eigvalue-sum} yields
\begin{equation}\label{eq:S_BE-equals-S-rho_X}
S_{\mathrm{BE}}(X) = S(\rho_{X}),
\end{equation}
establishing that $S_{\mathrm{BE}}(X)$ is the von Neumann entropy of the non-interacting bosonic thermal state whose mean occupations in the eigenmodes of $X$ are the eigenvalues of $X$. Equivalently, by the Gibbs variational principle, $S_{\mathrm{BE}}(X)$ is the maximum von Neumann entropy attainable by any bosonic state whose mean occupations in the eigenbasis of $X$ are $\{x_{j}\}_{j=1}^{d}$.

\subsection{Concavity}
\label{app:concavity_BEentropy}
Observe that the Bose--Einstein entropy, defined in~\eqref{eq:BE-entropy-def}, takes the form of a trace function $S_{\mathrm{BE}}(X) = \Tr[g(X)]$, where the scalar function is the scalar bosonic entropy, defined as
\begin{equation}
    g(x) \coloneqq (x+1)\ln(x+1) - x\ln x.
\end{equation}
Let us first prove the concavity of the scalar function $g(x)$. To do so, we evaluate its derivatives for $x > 0$. The first derivative is
\begin{align}
    g'(x) &= \frac{d}{dx} (x+1)\ln(x+1) - \frac{d}{dx} x\ln x \\
    &= \left( \ln(x+1) + 1 \right) - \left( \ln x + 1 \right) \\
    &= \ln(x+1) - \ln x.
\end{align}
The second derivative is then given by
\begin{equation}\label{step:second_der_g(x)}
    g''(x) = \frac{1}{x+1} - \frac{1}{x} = \frac{x - (x+1)}{x(x+1)} = -\frac{1}{x(x+1)}.
\end{equation}
Because $x(x+1) > 0$ for all $x > 0$, it follows that $g''(x) < 0$. Therefore, the scalar function $g(x)$ is strictly concave on the interval $(0, \infty)$ and concave on $[0, \infty)$.

By a foundational theorem of matrix analysis~\cite[Theorem~2.10]{Carlen2010trace}, if a real-valued scalar function $g(x)$ is concave on an interval $\mathcal{I}$, then the corresponding trace function $X \mapsto \Tr[g(X)]$ is concave on the set of Hermitian matrices with eigenvalues in $\mathcal{I}$. Consequently, the Bose--Einstein entropy for positive semidefinite operators $X \geq 0$ is concave.

\section{Duality theory and approximation bounds}
\label{app:duality-and-approximation}

\subsection{Proof of Proposition~\ref{prop:dual_general_sdp}}
\label{app:Proof_dual_gen_sdp}

In this appendix, we prove the duality relation for the standard, unregularized semidefinite optimization problem defined in~\eqref{eq:general_SDP}. Consider the primal problem:
\begin{align}
    E(\mathcal{Q}, q) & =\min_{X\geq0}\left\{ \Tr\!\left[HX\right]:\Tr\!\left[Q_{i}X\right]=q_{i}\,\forall i\in\left[c\right]\right\} \\
    & =\min_{X\geq0}\left\{ \Tr\!\left[HX\right]+\sup_{\mu\in\mathbb{R}^{c}}\sum_{i\in\left[c\right]}\mu_{i}\left(q_{i}-\Tr\!\left[Q_{i}X\right]\right)\right\} \label{step:lagrange_mult}\\
    & =\min_{X\geq0}\sup_{\mu\in\mathbb{R}^{c}}\left\{ \Tr\!\left[HX\right]+\sum_{i\in\left[c\right]}\mu_{i}\left(q_{i}-\Tr\!\left[Q_{i}X\right]\right)\right\} \\
    & =\min_{X\geq0}\sup_{\mu\in\mathbb{R}^{c}}\left\{ \mu\cdot q+\Tr\!\left[\left(H-\mu\cdot Q\right)X\right]\right\} \\
    & =\sup_{\mu\in\mathbb{R}^{c}}\min_{X\geq0}\left\{ \mu\cdot q+\Tr\!\left[\left(H-\mu\cdot Q\right)X\right]\right\} \label{step:minimax}\\
    & =\sup_{\mu\in\mathbb{R}^{c}}\left\{ \mu\cdot q+\min_{X\geq0}\Tr\!\left[\left(H-\mu\cdot Q\right)X\right]\right\} \label{eq:sup-and-min-over-X}\\
    & =\sup_{\mu\in\mathbb{R}^{c}}\left\{ \mu\cdot q:H-\mu\cdot Q\geq0\right\} \label{step:optimal_min}.
\end{align}
The equality in~\eqref{step:lagrange_mult} follows from introducing the Lagrange multiplier $\mu_{i}\in\mathbb{R}$ for the $i$th constraint. The equality in~\eqref{step:minimax} follows from strong duality for semidefinite programming. The feasible domains $\{X : X \geq 0\}$ and $\mathbb{R}^c$ are convex, and the objective function is linear in both $\mu$ and $X$. The assumed existence of a strictly feasible point (Assumption~\ref{ass:Slater}) ensures that the duality gap is zero, allowing the exchange of the minimum and supremum~\cite[Section 5]{Boyd2004convex}. The last equality in~\eqref{step:optimal_min} follows from applying Lemma~\ref{lem:unbounded-trace-min} below with the dual slack operator $A = K_\mu \coloneqq H - \mu \cdot Q$, which enforces the linear matrix inequality $K_\mu \geq 0$ to prevent the inner minimization from diverging to $-\infty$.

Furthermore, because the objective function $f(\mu) \coloneqq \mu\cdot q$ is linear in $\mu$ and the feasible region defined by the linear matrix inequality $K_\mu \geq 0$ forms a convex set, the dual optimization problem constitutes a convex optimization problem in dual form.

\begin{lemma}
\label{lem:unbounded-trace-min}
For a $d\times d$ Hermitian matrix $A$, where $d\in\mathbb{N}$, the following equality holds:
\begin{equation}
\min_{X\geq0}\Tr\!\left[AX\right]=\begin{cases}
0 & \text{if }A\geq0,\\
-\infty & \text{otherwise}.
\end{cases}\label{eq:unbounded-trace-lemma}
\end{equation}
\end{lemma}

\begin{proof}
Let $A=\left(A\right)_{+}-\left(A\right)_{-}$ be the Jordan decomposition of $A$ into its positive and negative parts~\cite{Bhatia1997matrix}. 

First, consider the case where $A\geq0$, which implies $\left(A\right)_{-}=0$. Since $X\geq0$, the product of two positive semidefinite matrices has a non-negative trace, yielding $\Tr\!\left[AX\right]\geq0$. Thus, the minimum of $\Tr\!\left[AX\right]$ over all $X\geq0$ is exactly $0$, which is achieved by choosing $X=0$.

Let us now consider the case where $A\not\geq 0$, meaning $\left(A\right)_{-}\neq0$. This implies that $A$ possesses at least one strictly negative eigenvalue, $-\lambda$ (where $\lambda>0$), with a corresponding normalized eigenvector $|v\rangle$. We can construct a positive semidefinite operator $X=c|v\rangle\!\langle v|$ for any real constant $c>0$. Evaluating the objective yields:
\begin{equation}
\Tr\!\left[AX\right]=c\langle v|A|v\rangle=-c\lambda.
\end{equation}
By taking the limit as $c\to\infty$, we find that $\Tr\!\left[AX\right]\to-\infty$. Thus, the minimum over all $X\geq0$ is unbounded from below.
\end{proof}

\subsection{Proof of Theorem~\ref{prop:dual_sdp_regularized}}
\label{sec:Proof-of-dual-be-reg}

In this appendix, we prove the dual representation of the Bose--Einstein regularized semidefinite program. Let $F_T(\mathcal{Q}, q)$ denote the optimal primal value. We proceed by expressing the problem using a minimax formulation:
\begin{align}
 F_T(\mathcal{Q}, q) & =\min_{X\geq0}\left\{ \Tr\!\left[HX\right]-T S_{\mathrm{BE}}(X): \Tr\!\left[Q_{i}X\right]=q_{i}\,\forall i\in\left[c\right] \right\} \\
 & = \min_{X\geq0}\left\{ \Tr\!\left[HX\right]-T S_{\mathrm{BE}}(X) +\sup_{\mu\in\mathbb{R}^{c}}\sum_{i\in\left[c\right]}\mu_{i}\left(q_{i}-\Tr\!\left[Q_{i}X\right]\right)\right\} \label{step:lagrange_mult_reg}\\
 & =\min_{X\geq0}\sup_{\mu\in\mathbb{R}^{c}}\left\{ \mu\cdot q+\Tr\!\left[\left(H-\mu\cdot Q\right)X\right]-T S_{\mathrm{BE}}(X) \right\} \\
 & = \sup_{\mu\in\mathbb{R}^{c}}\min_{X\geq0}\left\{ \mu\cdot q+\Tr\!\left[\left(H-\mu\cdot Q\right)X\right] -T S_{\mathrm{BE}}(X) \right\} \label{step:strong_duality_reg}\\
 & =\sup_{\mu\in\mathbb{R}^{c}}\left\{ \mu\cdot q+\min_{X\geq0}\left\{ \Tr\!\left[K_\mu X\right]-T S_{\mathrm{BE}}(X) \right\} \right\} \label{eq:inner_min_step}\\
 & = \sup_{\mu\in\mathbb{R}^{c}}\left\{ \mu\cdot q+T\Tr\!\left[\ln\!\left(I-e^{-\frac{1}{T}K_\mu}\right)\right]\right\}, \label{step:final_dual_form}
\end{align}
where $K_\mu \coloneqq H - \mu \cdot Q$ is the dual slack operator, also called grand canonical Hamiltonian. 
The equality in~\eqref{step:lagrange_mult_reg} follows by introducing the Lagrange multipliers $\mu_{i}\in\mathbb{R}$ for the constraints $\Tr[Q_{i}X]=q_{i}$. The equality in~\eqref{step:strong_duality_reg} follows from strong duality for convex optimization. The feasible domains $\{X : X \geq 0\}$ and $\mathbb{R}^c$ are convex, the objective function is linear in $\mu$, and it is strictly convex in $X$ due to the strict concavity of the Bose--Einstein entropy $S_{\mathrm{BE}}(X)$ (see Appendix~\ref{app:concavity_BEentropy}). Because the constraints are identical to those of the unregularized problem (see Appendix~\ref{app:Proof_dual_gen_sdp}), the assumed existence of a strictly feasible point (Assumption~\ref{ass:Slater}) ensures that the duality gap is zero, allowing the exchange of the minimum and supremum~\cite[Section 5]{Boyd2004convex}. Finally, the last equality in~\eqref{step:final_dual_form} follows directly from evaluating the inner minimization over $X$, as established in Lemma~\ref{lem:opt-BE-free-energy} below.

\begin{lemma}
\label{lem:opt-BE-free-energy}
Let $A > 0$ be a $d\times d$ strictly positive definite Hermitian matrix, and let $T>0$. Then,
\begin{equation}
\min_{X\geq 0} \left( \Tr\!\left[AX\right]-T S_{\mathrm{BE}}(X) \right) = T\Tr\!\left[\ln\!\left(I-e^{-\frac{A}{T}}\right)\right].\label{eq:BE-opt-prob}
\end{equation}
Furthermore, the unique operator achieving this minimum is the Bose--Einstein thermal operator $X=\left(e^{\frac{A}{T}}-I\right)^{-1}$.
\end{lemma}

We provide two independent proofs for Lemma~\ref{lem:opt-BE-free-energy}. The first relies on standard first-order optimality conditions, and the second leverages the non-negativity of the Bose--Einstein relative entropy, defined in~\eqref{eq:BE_rel_entropy_main}.

\subsubsection{First Proof of Lemma~\ref{lem:opt-BE-free-energy} (First-Order Conditions)}\label{app:proof_lemma_first_order}

\begin{proof}
Because $X \mapsto -S_{\mathrm{BE}}(X)$ is strictly convex (see Appendix~\ref{app:concavity_BEentropy} for the proof of the concavity of $S_{\mathrm{BE}}(X)$), the objective function $h(A,T,X) \coloneqq \Tr[AX]-T S_{\mathrm{BE}}(X)$ of the minimization problem in~\eqref{eq:BE-opt-prob} is strictly convex in $X$. Thus, the first-order stationarity condition is both necessary and sufficient for global optimality. The matrix derivative with respect to $X$ is:
\begin{align}\label{step:der_reg_obj_fun}
 \frac{\partial h}{\partial X} = A-T\frac{\partial S_{\mathrm{BE}}(X)}{\partial X} = A-T\left(\ln\!\left(X+I\right)-\ln X\right),
\end{align}
where we used the scalar derivative $g'(x) = \ln(x+1)-\ln x$, with $g(x)$ as defined in~\eqref{eq:bosonic-entropy-def}. Setting the gradient to zero yields the optimal operator:
\begin{align}
0 = A-T\left(\ln\!\left(X+I\right)-\ln X\right) \implies \frac{A}{T} = \ln\!\left(I+X^{-1}\right) \implies X = \left(e^{\frac{A}{T}}-I\right)^{-1}.
\end{align}
To evaluate the minimum value, we plug this optimal $X$ back into $h(A,T,X)$. Because the operators commute, we can evaluate this equivalently using the scalar eigenvalues $a>0$ of $A$. For a single eigenmode, the scalar objective $a x -T s(x)$ evaluated at the optimum $x = (e^{a/T}-1)^{-1}$ is:
\begin{align}
 & \frac{a}{e^{\frac{a}{T}}-1}-T\left(\frac{e^{\frac{a}{T}}}{e^{\frac{a}{T}}-1}\ln\!\left(\frac{e^{\frac{a}{T}}}{e^{\frac{a}{T}}-1}\right)-\frac{1}{e^{\frac{a}{T}}-1}\ln\!\left(\frac{1}{e^{\frac{a}{T}}-1}\right)\right)\nonumber \\
 & =\frac{a}{e^{\frac{a}{T}}-1}-T\left(\frac{e^{\frac{a}{T}}}{e^{\frac{a}{T}}-1}\left(\frac{a}{T}-\ln\!\left(e^{\frac{a}{T}}-1\right)\right)+\frac{1}{e^{\frac{a}{T}}-1}\ln\!\left(e^{\frac{a}{T}}-1\right)\right)\\
 & =\frac{a - a e^{\frac{a}{T}}}{e^{\frac{a}{T}}-1}+T\left(\frac{e^{\frac{a}{T}} - 1}{e^{\frac{a}{T}}-1}\right)\ln\!\left(e^{\frac{a}{T}}-1\right)\\
 & =-a+T\ln\!\left(e^{\frac{a}{T}}-1\right) \\
 & = T \ln \!\left( e ^{-\frac{a}{T}} \right)+T\ln\!\left(e^{\frac{a}{T}}-1\right)\\
 & = T\ln\!\left(\frac{e^{\frac{a}{T}}-1}{e^{\frac{a}{T}}}\right) \\
 & = T\ln\!\left(1-e^{-\frac{a}{T}}\right).
\end{align}
Summing this scalar result over all eigenvalues establishes the final matrix trace expression $T\Tr[\ln(I-e^{-A/T})]$, completing the proof.
\end{proof}

\subsubsection{Second Proof of Lemma~\ref{lem:opt-BE-free-energy} (Relative Entropy)}\label{app:proof_lemma_rel_entropy}

\begin{proof}
Recall the definition of the Bose--Einstein relative entropy for positive semidefinite operators $X, Y \geq 0$:
\begin{equation}\label{eq:BE_rel_entropy_appendix}
    D_{\mathrm{BE}}(X\|Y) \coloneqq -S_{\mathrm{BE}}(X) + \Tr[(X+I)\ln(Y+I) - X\ln Y].
\end{equation}
By Proposition~\ref{prop:BE-faithfulness}, $D_{\mathrm{BE}}(X\|Y) \geq 0$, with equality holding if and only if $X=Y$. Let us define the target operator $Y \coloneqq (e^{A/T}-I)^{-1}$. We can express the matrix exponential $e^{A/T}$ strictly in terms of $Y$ as $e^{A/T} = Y^{-1} + I = (I+Y)Y^{-1}$. Taking the matrix logarithm of both sides yields:
\begin{equation}
    \frac{A}{T} = \ln(I+Y) - \ln Y.
\end{equation}
Substituting this expression into the energy term $\Tr[AX]$, the objective function of the minimization in~\eqref{eq:BE-opt-prob} becomes:
\begin{align}
    \Tr[AX] - T S_{\mathrm{BE}}(X) &= T \Tr\!\left[X \left(\frac{A}{T}\right)\right] - T S_{\mathrm{BE}}(X) \\
    &= T \Tr[X \ln(I+Y) - X \ln Y] - T S_{\mathrm{BE}}(X). \label{eq:expanded_objective}
\end{align}
Now, observe that multiplying the relative entropy~\eqref{eq:BE_rel_entropy_appendix} by $T$ and rearranging terms gives:
\begin{align}
    T D_{\mathrm{BE}}(X\|Y) &= -T S_{\mathrm{BE}}(X) + T \Tr[X\ln(Y+I) + \ln(Y+I) - X\ln Y] \\
    &= \left( -T S_{\mathrm{BE}}(X) + T \Tr[X \ln(Y+I) - X \ln Y] \right) + T \Tr[\ln(Y+I)]. 
\end{align}
Recognizing the bracketed term as our expanded objective in~\eqref{eq:expanded_objective}, we can rewrite the objective function of the minimization in~\eqref{eq:BE-opt-prob} entirely in terms of the relative entropy:
\begin{equation}\label{eq:obj_fun_RelEntropy}
    \Tr[AX] - T S_{\mathrm{BE}}(X) = T D_{\mathrm{BE}}(X\|Y) - T \Tr[\ln(Y+I)].
\end{equation}
Because $D_{\mathrm{BE}}(X\|Y) \geq 0$, the global minimum is achieved by setting $X=Y$, which zeroes out the relative entropy term. It remains only to simplify the residual trace term. To do so, let us first evaluate
\begin{equation}
    Y + I = (e^{A/T}-I)^{-1} + I = \frac{I + e^{A/T} - I}{e^{A/T}-I} = \frac{e^{A/T}}{e^{A/T}-I} = (I - e^{-A/T})^{-1}.
\end{equation}
Therefore, $-T\Tr[\ln(Y+I)] = T\Tr[\ln(I - e^{-A/T})]$, completing the proof.
\end{proof}

\subsection{Proof of Proposition~\ref{prop:simple-approx-bnd} (simple bound on approximation error)}
\label{sec:proof-simple-approx-bnd}

We first prove the upper bound inequality:
\begin{equation}
E(\mathcal{Q},q) \geq F_{T}(\mathcal{Q},q). \label{eq:first-simple-ineq-en-to-free-en}
\end{equation}
Suppose that $X \geq 0$ is a feasible operator for the original optimization problem (i.e., it satisfies $\Tr[Q_{i}X]=q_{i}$ for all $i \in [c]$). Because the scalar function $s_{\mathrm{BE}}(x) = (x+1)\ln(x+1) - x\ln x \geq 0$ for all $x \geq 0$, the operator entropy defined in~\eqref{eq:BE-entropy-def} is strictly non-negative: $S_{\mathrm{BE}}(X) \geq 0$. Since $T>0$, we have:
\begin{align}
\Tr\!\left[HX\right] & \geq \Tr\!\left[HX\right] - T S_{\mathrm{BE}}(X) \\
 & \geq F_{T}(\mathcal{Q},q). \label{eq:proof_F_t_lower}
\end{align}
The inequality in~\eqref{eq:proof_F_t_lower} follows from the definition of $F_{T}(\mathcal{Q},q)$ in~\eqref{eq:sdp_primal_regularized} as the minimum over all feasible operators. Because this chain of inequalities holds for any feasible $X$, it must hold for the optimal $X^\star$, thus establishing $E(\mathcal{Q},q) \geq F_{T}(\mathcal{Q},q)$, where $E(\mathcal{Q},q)$ is the global minimum of $\Tr[HX]$ over the feasible set.

We now prove the lower bound inequality:
\begin{equation}
F_{T}(\mathcal{Q},q) \geq E(\mathcal{Q},q) - T S_{\max}, \label{eq:second-simple-ineq-en-to-free-en}
\end{equation}
which implies the desired result $F_T(\mathcal{Q},q) \geq E(\mathcal{Q},q) - \varepsilon$ when choosing a temperature $T \leq \frac{\varepsilon}{S_{\max}}$. Suppose $X \geq 0$ is a feasible operator. By the definition of $S_{\max}$ in~\eqref{eq:S_max_def}, its entropy is bounded by $S_{\max}$. Therefore:
\begin{align}
    \Tr\!\left[HX\right] - T S_{\mathrm{BE}}(X) & \geq \Tr\!\left[HX\right] - T S_{\max} \\
    & \geq E(\mathcal{Q},q) - T S_{\max}. \label{eq:proof_E_lower}
\end{align}
The inequality in~\eqref{eq:proof_E_lower} follows because $E(\mathcal{Q},q)$ is the global minimum of $\Tr[HX]$ over the feasible set. Taking the infimum of the left-hand side over all feasible $X$ yields $F_{T}(\mathcal{Q},q)$, thus concluding the proof.

\subsection{Proof of Proposition~\ref{prop:no-BE-ent-simple-approx-err-bnd}}
\label{sec:no-BE-ent-simple-approx-err-bnd}

We bound the quantity $\tilde{f}_{T}(\mu_{T}^{\star})$, defined in~\eqref{eq:f_tilde}, from both below and above. 

First, we establish the lower bound. Because $\mu_T^\star$ perfectly maximizes the regularized dual objective $f_T(\mu)$, the gradient of the objective with respect to $\mu$, defined in~\eqref{eq:gradient-BE-gen-obj-func}, must vanish at optimality:
\begin{equation}
    \nabla_\mu f_T(\mu_T^\star) = q - \Tr[Q X_T(\mu_T^\star)] = 0.
\end{equation}
Therefore, the optimal thermal operator $X_T(\mu_T^\star)$ exactly satisfies the primal equality constraints $\Tr[Q_i X_T(\mu_T^\star)] = q_i$ for all $i \in [c]$. Substituting this into the definition of $\tilde{f}_{T}(\mu_{T}^{\star})$, the dual penalty terms perfectly cancel, yielding exactly the unregularized physical energy:
\begin{equation}
    \tilde{f}_{T}(\mu_{T}^{\star}) = \Tr[H X_T(\mu_T^\star)]. \label{eq:f_tilde_trace_simplification}
\end{equation}
Since $X_T(\mu_T^\star)$ is an exactly feasible, positive semidefinite operator for the unregularized primal problem, its expected energy cannot be lower than the exact global minimum:
\begin{equation}
    \tilde{f}_{T}(\mu_{T}^{\star}) \geq E(\mathcal{Q},q).
\end{equation}
This establishes the lower bound.

Next, we establish the upper bound by relating the physical energy back to the regularized free energy. From the definition of the regularized primal problem, the optimal free energy is exactly:
\begin{equation}
    F_T(\mathcal{Q},q) = \Tr[H X_T(\mu_T^\star)] - T S_{\mathrm{BE}}(X_T(\mu_T^\star)).
\end{equation}
Rearranging this equality and substituting~\eqref{eq:f_tilde_trace_simplification} yields:
\begin{align}
    \tilde{f}_{T}(\mu_{T}^{\star}) &= F_T(\mathcal{Q},q) + T S_{\mathrm{BE}}(X_T(\mu_T^\star)) \\
    &\leq E(\mathcal{Q},q) + T S_{\mathrm{BE}}(X_T(\mu_T^\star)). \label{step:upper_bound_intermediate}
\end{align}
The inequality in~\eqref{step:upper_bound_intermediate} is a direct consequence of Proposition~\ref{prop:simple-approx-bnd}, which guarantees $F_T(\mathcal{Q},q) \leq E(\mathcal{Q},q)$. 

Because the optimal thermal operator $X_T(\mu_T^\star)$ lies within the feasible set, its entropy is strictly bounded by the maximum feasible entropy $S_{\max}$ defined in~\eqref{eq:S_max_def}. By applying the temperature condition $T \leq \frac{\varepsilon}{S_{\max}}$, we bound the entropy penalty:
\begin{equation}
    T S_{\mathrm{BE}}(X_{T}(\mu_T^\star)) \leq T S_{\max} \leq \varepsilon.
\end{equation}
Substituting this bound into~\eqref{step:upper_bound_intermediate} establishes the final upper bound $\tilde{f}_{T}(\mu_{T}^{\star}) \leq E(\mathcal{Q},q) + \varepsilon$, completing the proof.

\subsection{Proof of Proposition~\ref{prop:approx-error-spectral-gap} (spectral gap bound on approximation error)}
\label{sec:proof-approx-error-spectral-gap}

To bound the approximation error of the unregularized energy evaluated at the thermal operator, we must bound the quantity $\tilde{f}_{T}(\mu_{T}^{\star}) \coloneqq \mu_{T}^{\star} \cdot q + \Tr[(H-\mu_{T}^{\star}\cdot Q)X_T(\mu_T^\star)]$.

The lower bound, $\tilde{f}_{T}(\mu_{T}^{\star}) \geq E(\mathcal{Q},q)$, follows directly from the exact primal feasibility of the optimal thermal operator, as established in~\eqref{eq:f_tilde_trace_simplification} in the proof of Proposition~\ref{prop:no-BE-ent-simple-approx-err-bnd}.

To establish the upper bound, we utilize the weak duality of the exact, unregularized SDP. Recall the definition of the regularized dual objective function in~\eqref{eq:dual_obj_function}, which contains the logarithmic trace term $T\Tr[\ln(I-e^{-K_\mu/T})]$. For this term to evaluate to a finite value at optimality, the argument of the matrix logarithm must be strictly positive definite, meaning $I-e^{-K_{\mu_T^\star}/T} > 0$. This inequality holds if and only if all eigenvalues of $e^{-K_{\mu_T^\star}/T}$ are strictly less than $1$, which in turn dictates that the dual slack operator must be strictly positive definite: $K_{\mu_T^\star} \coloneqq H - \mu_T^\star \cdot Q > 0$. 
Because the feasibility condition for the exact, unregularized dual problem is simply $K_\mu \geq 0$ (as established in Proposition~\ref{prop:dual_general_sdp}), the optimal regularized dual vector $\mu_T^\star$ is guaranteed to be a strictly feasible point for the unregularized dual problem. By standard weak duality, any valid dual feasible vector provides a lower bound on the exact primal optimal value:
\begin{equation}
    E(\mathcal{Q},q) \geq \mu_T^\star \cdot q.
\end{equation}

We can therefore bound the difference between the physical energy of the thermal operator and the unregularized optimal value as follows:
\begin{align}
    \tilde{f}_{T}(\mu_{T}^{\star}) - E(\mathcal{Q},q) &\leq \tilde{f}_{T}(\mu_{T}^{\star}) - \mu_T^\star \cdot q \\
    &= \Tr[H X_T(\mu_T^\star)] - \Tr[(\mu_T^\star \cdot Q) X_T(\mu_T^\star)] \\
    &= \Tr[K_{\mu_T^\star} X_T(\mu_T^\star)]. \label{eq:proof_duality_gap_substitution}
\end{align}
Equation~\eqref{eq:proof_duality_gap_substitution} reveals that the energy error is bounded exactly by the regularized duality gap. 
Because $K_{\mu_T^\star}$ and $X_T(\mu_T^\star)$ share a common eigenbasis, we can evaluate this trace sum over their joint eigenvalues. Let $\lambda_j > 0$ be the eigenvalues of $K_{\mu_T^\star}$. The expected occupation number for each mode is $x_j = (e^{\lambda_j/T} - 1)^{-1}$. Therefore, the energy penalty for a single mode is $\lambda_j x_j = \lambda_j (e^{\lambda_j/T} - 1)^{-1}$. 

We partition the trace sum into the $d_0$ low-energy modes ($\lambda_j = \lambda_{\min}$) and the $d-d_0$ excited modes ($\lambda_j \geq \lambda_{\min} + \Delta$):
\begin{equation}
    \Tr[K_{\mu_T^\star} X_T(\mu_T^\star)] = \sum_{j \in \text{ground}} \frac{\lambda_j}{e^{\lambda_j/T} - 1} + \sum_{j \in \text{excited}} \frac{\lambda_j}{e^{\lambda_j/T} - 1}.
\end{equation}

For the low-energy modes, we utilize the regularized complementary slackness identity established in~\eqref{eq:reg_comp_slackness_eigen}, which states $\lambda_j x_j = T x_j \ln(1+1/x_j)$. Because $x \ln(1+1/x) \leq 1$ for all $x \geq 0$, the energy penalty for any mode is strictly bounded by $T$, regardless of how large its occupation $x_j$ becomes. Thus:
\begin{equation}
    \sum_{j \in \text{ground}} \frac{\lambda_j}{e^{\lambda_j/T} - 1} \leq \sum_{j \in \text{ground}} T = T d_0.
\end{equation}

For the excited modes, we observe that the function $g(\lambda) = \lambda (e^{\lambda/T} - 1)^{-1}$ is strictly monotonically decreasing for all $\lambda > 0$. Because these modes are separated from the ground space by the spectral gap $\Delta$, their energy is strictly bounded from below by $\lambda_{\min} + \Delta$. Thus, their contribution is bounded by evaluating the function at this minimum excited energy:
\begin{equation}
    \sum_{j \in \text{excited}} \frac{\lambda_j}{e^{\lambda_j/T} - 1} \leq (d-d_0) \frac{\lambda_{\min} + \Delta}{e^{(\lambda_{\min} + \Delta)/T} - 1}.
\end{equation}

Summing the bounds for the two subspaces yields the strict upper bound on the unregularized energy error:
\begin{equation}
    \tilde{f}_{T}(\mu_{T}^{\star}) - E(\mathcal{Q},q) \leq T \cdot d_0 + \left(d-d_0\right) \frac{\lambda_{\min} + \Delta}{e^{(\lambda_{\min} + \Delta)/T} - 1}, \label{eq:proof_upper_bound_final}
\end{equation}
which establishes~\eqref{eq:duality_gap_spectral_bound}.

Finally, to ensure the total error is lower than or equal to $\varepsilon$, we restrict each of the two terms in~\eqref{eq:proof_upper_bound_final} to be $\leq \varepsilon/2$. 
For the ground-state term:
\begin{equation}
    T d_0 \leq \frac{\varepsilon}{2} \implies T \leq \frac{\varepsilon}{2 d_0}.
\end{equation}
For the excited-state term:
\begin{align}
    \left(d-d_0\right) \frac{\lambda_{\min} + \Delta}{e^{(\lambda_{\min} + \Delta)/T} - 1} \leq \frac{\varepsilon}{2} &\iff e^{(\lambda_{\min} + \Delta)/T} - 1 \geq \frac{2(d-d_0)(\lambda_{\min} + \Delta)}{\varepsilon} \nonumber \\
    &\iff \frac{\lambda_{\min} + \Delta}{T} \geq \ln\!\left(1 + \frac{2(d-d_0)(\lambda_{\min} + \Delta)}{\varepsilon}\right) \nonumber \\
    &\iff T \leq \frac{\lambda_{\min} + \Delta}{\ln\!\left(1 + \frac{2(d-d_0)(\lambda_{\min} + \Delta)}{\varepsilon}\right)}.
\end{align}
Taking the minimum of these two temperature thresholds yields~\eqref{eq:temp_threshold_spectral}, which guarantees the desired precision and completes the proof.

\section{Gradient and Hessian of the dual objective}
\label{app:gradient-hessian}

\subsection{Proof of Proposition~\ref{prop:gradient-BE-gen-obj-func} (gradient
of dual objective function)}

\label{sec:Proof-of-gradient-BE_dual}
Using the definition of the objective function $f_{T}(\mu)$ in~\eqref{eq:dual_obj_function}, consider that
\begin{align}
 \frac{\partial}{\partial\mu_{i}}f_{T}(\mu)
 & =\frac{\partial}{\partial\mu_{i}}\left(\mu\cdot q + T\Tr\!\left[\ln\!\left(I - e^{-\frac{1}{T}\left(H-\mu\cdot Q\right)}\right)\right]\right)\\
 & =q_{i}+T\frac{\partial}{\partial\mu_{i}}\Tr\!\left[\ln\!\left(I - e^{-\frac{1}{T}\left(H-\mu\cdot Q\right)}\right)\right].\label{eq:grad-proof}
\end{align}
Now consider that
\begin{align}
 \frac{\partial}{\partial\mu_{i}}\Tr\!\left[\ln\!\left(I- e^{-\frac{1}{T}\left(H-\mu\cdot Q\right)}\right)\right]
 & = \Tr\!\left[\left(I-e^{-\frac{1}{T}\left(H-\mu\cdot Q\right)}\right)^{-1}\frac{\partial}{\partial\mu_{i}}\left(I-e^{-\frac{1}{T}\left(H-\mu\cdot Q\right)}\right)\right] \label{step:grad_composite_fun}\\
 & =- \Tr\!\left[\left(I-e^{-\frac{1}{T}\left(H-\mu\cdot Q\right)}\right)^{-1}\frac{\partial}{\partial\mu_{i}}e^{-\frac{1}{T}\left(H-\mu\cdot Q\right)}\right]\\
 & = - \Tr\!\left[\begin{array}{c}
\left(I-e^{-\frac{1}{T}\left(H-\mu\cdot Q\right)}\right)^{-1}\int_{0}^{1}dt\,e^{-\frac{t}{T}\left(H-\mu\cdot Q\right)}\times\\
\frac{\partial}{\partial\mu_{i}}\left[-\frac{1}{T}\left(H-\mu\cdot Q\right)\right]e^{-\frac{1-t}{T}\left(H-\mu\cdot Q\right)}
\end{array}\right] \label{step:Duhamel's_step}\\
& = -\int_{0}^{1}dt\,\Tr\!\left[\begin{array}{c}
e^{-\frac{1-t}{T}\left(H-\mu\cdot Q\right)}\left(I-e^{-\frac{1}{T}\left(H-\mu\cdot Q\right)}\right)^{-1}\times\\
e^{-\frac{t}{T}\left(H-\mu\cdot Q\right)}\frac{Q_{i}}{T}
\end{array}\right] \label{step:cyclic_trace}\\
 & = -\frac{1}{T}\int_{0}^{1}dt\,\Tr\!\left[\begin{array}{c}
e^{-\frac{1-t}{T}\left(H-\mu\cdot Q\right)}e^{-\frac{t}{T}\left(H-\mu\cdot Q\right)}\times\\
\left(I-e^{-\frac{1}{T}\left(H-\mu\cdot Q\right)}\right)^{-1}Q_{i}
\end{array}\right]\\
 & =- \frac{1}{T}\Tr\!\left[e^{-\frac{1}{T}\left(H-\mu\cdot Q\right)}\left(I-e^{-\frac{1}{T}\left(H-\mu\cdot Q\right)}\right)^{-1}Q_{i}\right]\label{step:commutativity}\\
 & =- \frac{1}{T}\Tr\!\left[\left(e^{\frac{1}{T}\left(H-\mu\cdot Q\right)}-I\right)^{-1}Q_{i}\right]\\
 & = - \frac{1}{T}\Tr\!\left[X_{T}(\mu)Q_{i}\right].\label{eq:final-line-grad-proof}
\end{align}
The equality in~\eqref{step:grad_composite_fun} follows because
\begin{equation}
\frac{\partial}{\partial x}\Tr[f(A(x))]=\Tr\!\left[f'(A(x))\frac{\partial}{\partial x}A(x)\right]
\end{equation}
for a differentiable function $f$ of a matrix $A(x)$. The equality
in~\eqref{step:Duhamel's_step} follows from Duhamel's formula for the derivative of matrix exponentials
(see, e.g., \cite[Proposition~47]{wilde2025quantumfisherinformationmatrices}):
\begin{equation}
\frac{\partial}{\partial x}e^{A(x)}=\int_{0}^{1}dt\,e^{tA(x)}\left(\frac{\partial}{\partial x}A(x)\right)e^{\left(1-t\right)A(x)}.
\end{equation}
The equality in~\eqref{step:cyclic_trace} follows from the cyclicity of trace and the fact that
\begin{equation}
\frac{\partial}{\partial\mu_{i}}\left[-\frac{1}{T}\left(H-\mu\cdot Q\right)\right]=\frac{Q_{i}}{T}.
\end{equation}
The equality in~\eqref{step:commutativity} follows because $\left(I-e^{-\frac{1}{T}\left(H-\mu\cdot Q\right)}\right)^{-1}$
commutes with $e^{-\frac{t}{T}\left(H-\mu\cdot Q\right)}$. The final
equality in~\eqref{eq:final-line-grad-proof} follows by applying the definition of the thermal operator $X_{T}(\mu)$ in~\eqref{eq:BE_operator}. Combining~\eqref{eq:final-line-grad-proof}
with~\eqref{eq:grad-proof}, we conclude that
\begin{equation}
\frac{\partial}{\partial\mu_{i}}f_T(\mu)=q_{i}-\Tr\!\left[X_{T}(\mu)Q_{i}\right].
\end{equation}

\subsection{Proof of Proposition~\ref{prop:hessian-BE-gen-obj-func} (Hessian
of dual objective function)}

\label{sec:Proof-of-Hessian-BE-gen-obj-func}

We first prove~\eqref{eq:1st-hessian-exp}. Consider that
\begin{align}
 \frac{\partial^{2}}{\partial\mu_{i}\partial\mu_{j}}f_{T}(\mu)
 & =\frac{\partial}{\partial\mu_{i}}\left(q_{j}-\Tr\!\left[X_{T}(\mu)Q_{j}\right]\right)\\
 & =-\frac{\partial}{\partial\mu_{i}}\Tr\!\left[X_{T}(\mu)Q_{j}\right]\\
 & =-\Tr\!\left[\frac{\partial}{\partial\mu_{i}}\left[\left(e^{\frac{1}{T}\left(H-\mu\cdot Q\right)}-I\right)^{-1}\right]Q_{j}\right]\\ 
 & =\Tr\!\left[
\left(e^{\frac{1}{T}\left(H-\mu\cdot Q\right)}-I\right)^{-1}\left(\frac{\partial}{\partial\mu_{i}}\left[e^{\frac{1}{T}\left(H-\mu\cdot Q\right)}-I\right]\right)
\left(e^{\frac{1}{T}\left(H-\mu\cdot Q\right)}-I\right)^{-1}Q_{j}
\right] \label{step:der_inverse_matrix} \\
& =\Tr\!\left[X_{T}(\mu)\left[\frac{\partial}{\partial\mu_{i}}e^{\frac{1}{T}\left(H-\mu\cdot Q\right)}\right]X_{T}(\mu)Q_{j}\right]\label{eq:hessian-mid-proof}\\
& =-\frac{1}{T}\int_{0}^{1}dt\,\Tr\!\left[
 X_{T}(\mu)e^{\frac{t}{T}\left(H-\mu\cdot Q\right)}Q_{i}
 e^{\frac{1-t}{T}\left(H-\mu\cdot Q\right)}X_{T}(\mu)Q_{j}
\right]\label{step:Duhamel's_step_2}\\
& =-\frac{1}{T}\int_{0}^{1}dt\,\Tr\!\left[X_{T}(\mu,t) Q_{i}X_{T}(\mu,1-t)Q_{j}\right],\label{eq:hessian-elements}
\end{align}
where in~\eqref{step:der_inverse_matrix} we used the fact that
\begin{equation}
    \frac{\partial}{\partial x}\left[A(x)^{-1}\right] =-A(x)^{-1}\left(\frac{\partial}{\partial x}A(x)\right)A(x)^{-1},
\end{equation}
in~\eqref{step:Duhamel's_step_2} we used Duhamel's formula for the derivative of matrix exponentials
\begin{equation}
    \frac{\partial}{\partial x}e^{A(x)} =\int_{0}^{1}ds\,e^{sA(x)}\left(\frac{\partial}{\partial x}A(x)\right)e^{\left(1-s\right)A(x)},
\end{equation}
and in the last equality~\eqref{eq:hessian-elements} we used the definition of $X_{T}(\mu,s)$ in~\eqref{eq:BEop-t-depend}.

We now prove~\eqref{eq:2nd-hessian-exp}. Recall from \cite[Lemmas~10 and 12]{patel2025quantumboltzmannmachinelearning}
that
\begin{equation}
\frac{\partial}{\partial x}e^{A(x)}=\frac{1}{2}\left\{ \Phi_{A(x)}\!\left(\frac{\partial}{\partial x}A(x)\right),e^{A(x)}\right\} ,\label{eq:deriv-exp-fourier-trans}
\end{equation}
where the quantum channel $\Phi_{A(x)}$ is defined as
\begin{equation}
\Phi_{A(x)}(Y)\coloneqq\int_{-\infty}^{\infty}dt\,\gamma(t)e^{-iA(x)t}Ye^{iA(x)t},
\end{equation}
and the high-peak tent probability density $\gamma(t)$ is defined
in~\eqref{eq:high-peak-tent-prob-dens} (see also~\cite{Hastings2007,Kim2012,Ejima2019,Kato2019,Anshu2021}).
Starting from~\eqref{eq:hessian-mid-proof} and applying~\eqref{eq:deriv-exp-fourier-trans},
while noting that
\begin{equation}
\frac{\partial}{\partial\mu_{i}}\left(\frac{1}{T}\left(H-\mu\cdot Q\right)\right)=-\frac{Q_{i}}{T},
\end{equation}
we find that

\begin{align}
 \frac{\partial^{2}}{\partial\mu_{i}\partial\mu_{j}}f_{T}(\mu)
 & =\Tr\!\left[X_{T}(\mu)\left[\frac{\partial}{\partial\mu_{i}}e^{\frac{1}{T}\left(H-\mu\cdot Q\right)}\right]X_{T}(\mu)Q_{j}\right]\\
 & =-\frac{1}{2T}\Tr\!\left[X_{T}(\mu)\left\{ \Phi_{\mu}\!\left(Q_{i}\right),e^{\frac{1}{T}\left(H-\mu\cdot Q\right)}\right\} X_{T}(\mu)Q_{j}\right]\\
 & =-\frac{1}{2T}\Tr\!\left[X_{T}(\mu)\Phi_{\mu}\!\left(Q_{i}\right)e^{\frac{1}{T}\left(H-\mu\cdot Q\right)}X_{T}(\mu)Q_{j}\right] 
 -\frac{1}{2T}\Tr\!\left[X_{T}(\mu)e^{\frac{1}{T}\left(H-\mu\cdot Q\right)}\Phi_{\mu}\!\left(Q_{i}\right)X_{T}(\mu)Q_{j}\right]\\
 & = -\frac{1}{2T}\Tr\!\left[X_{T}(\mu)\Phi_{\mu}\!\left(Q_{i}\right)\left(X_{T}(\mu) +I\right)Q_{j}\right] -\frac{1}{2T}\Tr\!\left[\left(X_{T}(\mu)+I\right)\Phi_{\mu}\!\left(Q_{i}\right)X_{T}(\mu)Q_{j}\right] \label{step:def_X+I}\\
 & =-\frac{1}{T}\operatorname{Re}\!\left\{\Tr\!\left[X_{T}(\mu)\Phi_{\mu}(Q_{i})\left(X_{T}(\mu)+I\right)Q_{j}\right]\right\},
\end{align}
where the equality in~\eqref{step:def_X+I} follows from the definition of $X_T(\mu)+I$ in~\eqref{eq:other-outcome-BE}.

\subsection{Proof of Proposition~\ref{prop:hessian-NSD-spec-up-bnd} (bounds
on Hessian of dual objective function)}
\label{sec:hessian-NSD-spec-up-bnd}We first prove that the Hessian
is negative semidefinite. Let $v\in\mathbb{R}^{c}$ be arbitrary,
and let $\nabla^{2}f_{T}(\mu)$ denote the Hessian matrix for $f_{T}(\mu)$,
with elements as given in~\eqref{eq:1st-hessian-exp}. Consider that
\begin{align}
 -v^{T}\nabla^{2}f_{T}(\mu)v
 & =-\sum_{i,j\in\left[c\right]}v_{i}\left[\frac{\partial^{2}}{\partial\mu_{i}\partial\mu_{j}}f_T(\mu)\right]v_{j}\\
 & =\frac{1}{T}\sum_{i,j\in\left[c\right]}v_{i}\int_{0}^{1}ds\,\Tr\!\left[Q_{i}X_{T}(\mu,1-s)Q_{j}X_{T}(\mu,s)\right]v_{j}\\
 & =\frac{1}{T}\int_{0}^{1}ds\,\Tr\!\left[
\left(\sum_{i\in\left[c\right]}v_{i}Q_{i}\right)X_{T}(\mu,1-s)
\left(\sum_{j\in\left[c\right]}v_{j}Q_{j}\right)X_{T}(\mu,s)\right]\\
 & =\frac{1}{T}\int_{0}^{1}ds\,\Tr\!\left[WX_{T}(\mu,1-s)WX_{T}(\mu,s)\right]\label{eq:Hessian-concave-last-step}\\
 & \geq0,
\end{align}
where $W\coloneqq\sum_{i\in\left[c\right]}v_{i}Q_{i}$. The last inequality
follows because the matrix $X_{T}(\mu,1-s)$ is positive semidefinite,
which implies that $WX_{T}(\mu,1-s)W$ is positive semidefinite given
that $W$ is Hermitian. Also, the matrix $X_{T}(\mu,s)$ is positive
semidefinite. As such, the trace expression in~\eqref{eq:Hessian-concave-last-step}
is non-negative for all $s\in\left[0,1\right]$, so that the integral
is non-negative also.

Let us now find an upper bound on the spectral norm of the Hessian matrix. To do so, we first establish a uniform bound on the operator norm of $X_{T}(\mu,s)$, defined in~\eqref{eq:BEop-t-depend}. Recall that $K_{\mu}\coloneqq H-\mu\cdot Q$ is assumed to be strictly positive definite, and let $\lambda_{\min}(K_{\mu})>0$ denote its smallest eigenvalue. Since $X_{T}(\mu,s)$ has the spectral decomposition
\begin{equation}
X_{T}(\mu,s)=\sum_{j}\frac{e^{s\lambda_{j}/T}}{e^{\lambda_{j}/T}-1}\left|\varphi_{j}\right\rangle \!\left\langle \varphi_{j}\right|,
\end{equation}
where $\lambda_{j}$ are the eigenvalues of $K_{\mu}$, and because the scalar function $h_{s}(\lambda)\coloneqq e^{s\lambda/T}/(e^{\lambda/T}-1)$ is monotonically decreasing in $\lambda>0$ for every $s\in[0,1]$, we conclude that
\begin{equation}
\left\Vert X_{T}(\mu,s)\right\Vert =\frac{e^{s\lambda_{\min}(K_\mu)/T}}{e^{\lambda_{\min}(K_\mu)/T}-1}.\label{eq:XT-s-norm}
\end{equation}
To verify the monotonicity claim, note that for $\lambda>0$ and $s\in[0,1]$,
\begin{equation}
\frac{\partial}{\partial\lambda}h_{s}(\lambda)=\frac{e^{s\lambda/T}}{T(e^{\lambda/T}-1)^{2}}\left[s\left(e^{\lambda/T}-1\right)-e^{\lambda/T}\right]=\frac{e^{s\lambda/T}}{T(e^{\lambda/T}-1)^{2}}\left[(s-1)e^{\lambda/T}-s\right]\leq0,
\end{equation}
where the last inequality follows because $(s-1)e^{\lambda/T}\leq0$ and $-s\leq0$ for $s\in[0,1]$ and $\lambda>0$.

Let us observe that the product of the norms of $X_{T}(\mu,s)$ and $X_{T}(\mu,1-s)$ is independent of $s$:
\begin{equation}
\left\Vert X_{T}(\mu,s)\right\Vert \cdot\left\Vert X_{T}(\mu,1-s)\right\Vert =\frac{e^{s\lambda_{\min}(K_\mu)/T}}{e^{\lambda_{\min}(K_\mu)/T}-1}\cdot\frac{e^{(1-s)\lambda_{\min}(K_\mu)/T}}{e^{\lambda_{\min}(K_\mu)/T}-1}=\frac{e^{\lambda_{\min}(K_\mu)/T}}{(e^{\lambda_{\min}(K_\mu)/T}-1)^{2}}=\bar{n}_{\mu}(\bar{n}_{\mu}+1),\label{eq:norm-product}
\end{equation}
where we have defined the Bose--Einstein occupation number of the lowest-energy mode as
\begin{equation}
\bar{n}_{\mu}\coloneqq \frac{1}{e^{\lambda_{\min}(K_{\mu})/T}-1}=\left\Vert X_{T}(\mu)\right\Vert .\label{eq:nbar-def}
\end{equation}
Note also that $\left\Vert X_{T}(\mu)+I\right\Vert =\bar{n}_{\mu}+1$.

With this in hand, let us now first bound the individual Hessian matrix elements. Consider that
\begin{align}
 \left|\frac{\partial^{2}}{\partial\mu_{i}\partial\mu_{j}}f_{T}(\mu)\right|
 & =\frac{1}{T}\left|\int_{0}^{1}ds\,\Tr\!\left[Q_{i}X_{T}(\mu,1-s)Q_{j}X_{T}(\mu,s)\right]\right|\\
 & \leq\frac{1}{T}\int_{0}^{1}ds\,\left|\Tr\!\left[X_{T}(\mu,s)Q_{i}X_{T}(\mu,1-s)Q_{j}\right]\right|\\
 & \leq\frac{1}{T}\int_{0}^{1}ds\,\left\Vert X_{T}(\mu,s)\right\Vert \left\Vert Q_{i}\right\Vert _{1}\left\Vert X_{T}(\mu,1-s)\right\Vert \left\Vert Q_{j}\right\Vert \label{step:holder_ineq}\\
 & \leq\frac{\bar{n}_{\mu}(\bar{n}_{\mu}+1)}{T}\left\Vert Q_{i}\right\Vert _{1}\left\Vert Q_{j}\right\Vert .\label{eq:hessian-elements-bound}
\end{align}
The inequality in~\eqref{step:holder_ineq} follows from the H\"older inequality, and the last step follows from~\eqref{eq:norm-product} and the fact that the product of norms is independent of $s$, so that the integral over $s\in[0,1]$ evaluates to one.

An alternative proof of the upper bound in~\eqref{eq:hessian-elements-bound}
is as follows, making use of the alternative expression for the Hessian elements in~\eqref{eq:2nd-hessian-exp}:
\begin{align}
  \left|\frac{\partial^{2}}{\partial\mu_{i}\partial\mu_{j}}f_{T}(\mu)\right|
 & =\left|-\frac{1}{T}\operatorname{Re}\!\left\{\Tr\!\left[X_{T}(\mu)\Phi_{\mu}(Q_{i})\left(X_{T}(\mu)+I\right)Q_{j}\right]\right\}\right|\\
 & \leq\frac{1}{T}\left|\Tr\!\left[X_{T}(\mu)\Phi_{\mu}(Q_{i})\left(X_{T}(\mu)+I\right)Q_{j}\right]\right|\\
 & \leq\frac{1}{T}\left\Vert X_{T}(\mu)\right\Vert \left\Vert \Phi_{\mu}(Q_{i})\right\Vert _{1}\left\Vert X_{T}(\mu)+I\right\Vert \left\Vert Q_{j}\right\Vert \\
 & \leq\frac{\bar{n}_{\mu}(\bar{n}_{\mu}+1)}{T}\left\Vert Q_{i}\right\Vert _{1}\left\Vert Q_{j}\right\Vert .\label{eq:hessian-elements-bound-1}
\end{align}
The last inequality follows because $\left\Vert X_{T}(\mu)\right\Vert =\bar{n}_{\mu}$, $\left\Vert X_{T}(\mu)+I\right\Vert =\bar{n}_{\mu}+1$, and $\left\Vert \Phi_{\mu}(Q_{i})\right\Vert _{1}\leq\left\Vert Q_{i}\right\Vert _{1}$, given that the trace norm is convex and unitarily invariant and $\Phi_{\mu}$ is a convex mixture of unitary conjugations.

Let us now obtain an upper bound on the largest singular value of
the Hessian (i.e., its spectral norm). Consider that
\begin{align}
\left\Vert \nabla^{2}f_{T}(\mu)\right\Vert  & \leq\left\Vert \nabla^{2}f_{T}(\mu)\right\Vert _{1} \label{step:spectral_norm}\\
 & = -\Tr\!\left[\nabla^{2}f_{T}(\mu)\right] \label{step:Hessian_negative}\\
 & \leq\sum_{i\in\left[c\right]}\left|\frac{\partial^{2}}{\partial\mu_{i}^{2}}f_{T}(\mu)\right|\\
 & \leq\frac{\bar{n}_{\mu}(\bar{n}_{\mu}+1)}{T}\sum_{i\in\left[c\right]}\left\Vert Q_{i}\right\Vert _{1}\left\Vert Q_{i}\right\Vert \label{step:hessian_bound}.
\end{align}
The inequality in~\eqref{step:spectral_norm} follows because the spectral norm does not exceed
the trace norm. The equality in~\eqref{step:Hessian_negative} follows because the Hessian matrix
$\nabla^{2}f_{T}(\mu)$ is negative semidefinite. The inequality in~\eqref{step:hessian_bound}
follows from~\eqref{eq:hessian-elements-bound}.

If the dual vectors $\mu$ are restricted to a parameter set $\mathcal{M}$ such that $\inf_{\mu\in\mathcal{M}}\lambda_{\min}(H-\mu\cdot Q)\geq\lambda_{\min}>0$, then defining the worst-case uniform bound $\bar{n}_{\lambda_{\min}}\coloneqq (e^{\lambda_{\min}/T}-1)^{-1}$, we can bound the occupation number across the entire set: $\bar{n}_{\mu} \leq \bar{n}_{\lambda_{\min}}$. Thus, the following uniform bound holds for the spectral norm:
\begin{equation}
\sup_{\mu\in\mathcal{M}}\left\Vert \nabla^{2}f_{T}(\mu)\right\Vert \leq L_{T}\coloneqq\frac{\bar{n}_{\lambda_{\min}}(\bar{n}_{\lambda_{\min}}+1)}{T}\sum_{i\in\left[c\right]}\left\Vert Q_{i}\right\Vert _{1}\left\Vert Q_{i}\right\Vert .\label{eq:smoothness-parameter_app}
\end{equation}

\section{Quantum algorithms and complexity analysis}
\label{app:quantum-algorithms}

\subsection{Proof of Proposition~\ref{prop:BE-algorithm-correctness}}
\label{sec:proof-BE-algorithm}

In this section, we provide the full mathematical derivation proving that the classical aggregation of measurements from the quantum circuit in Figure~\ref{fig:BE_circuit} yields an unbiased estimate of the Bose--Einstein thermal trace $\Tr[X_T(\mu) Q_i]$. The proof proceeds in three steps: (1) state evolution of the Hadamard test, (2) the Joint Monte Carlo expectation over time and state index, and (3) the geometric series reconstruction.

\paragraph*{Step 1: State Evolution of the Hadamard Test.}
For a given circuit execution, assume a fixed state index $k$ and a fixed evolution time $t$. We initialize the data register in the corresponding density matrix $\rho_k$. 
Let $U = e^{-i t K_\mu}$ denote the unitary time-evolution operator applied to the data register. The initial state of the bipartite system (the control qubit and the data register) is the product state:
\begin{equation}
    \sigma_0 = |0\rangle\!\langle 0| \otimes \rho_k.
\end{equation}
Applying the first Hadamard gate to the control qubit, as shown in the quantum circuit in Figure~\ref{fig:BE_circuit}, creates a coherent superposition in the control space:
\begin{equation}
    \sigma_1 = (H \otimes I) \sigma_0 (H \otimes I) = |+\rangle\!\langle +| \otimes \rho_k = \frac{1}{2} \Big( |0\rangle\!\langle 0| + |0\rangle\!\langle 1| + |1\rangle\!\langle 0| + |1\rangle\!\langle 1| \Big) \otimes \rho_k.
\end{equation}
Next, applying the controlled-unitary gate $CU = |0\rangle\!\langle 0| \otimes I + |1\rangle\!\langle 1| \otimes U$, the state updates to $\sigma_2 = CU \sigma_1 CU^\dagger$, which evaluates to:
\begin{equation}
    \sigma_2 = \frac{1}{2} \Big( |0\rangle\!\langle 0| \otimes \rho_k + |0\rangle\!\langle 1| \otimes \rho_k U^\dagger + |1\rangle\!\langle 0| \otimes U\rho_k + |1\rangle\!\langle 1| \otimes U\rho_k U^\dagger \Big).
\end{equation}
Finally, we apply the second Hadamard gate to the control qubit to yield the state $\sigma_3 = (H \otimes I) \sigma_2 (H \otimes I)$. We then measure the control qubit in the computational $Z$-basis. 
The expectation value of the Pauli $Z$ measurement is given by $\langle Z \rangle = \Tr[(Z \otimes I)\sigma_3]$. Using the cyclic property of the trace and the identity $HZH = X$, we can compute this expectation directly from $\sigma_2$:
\begin{equation}
    \langle Z \rangle = \Tr[(HZH \otimes I)\sigma_2] = \Tr[(X \otimes I)\sigma_2].
\end{equation}
Substituting $X = |0\rangle\!\langle 1| + |1\rangle\!\langle 0|$ into the trace over the control qubit extracts the off-diagonal block elements of $\sigma_2$:
\begin{align}
    \langle Z \rangle &= \Tr_{\operatorname{data}} \left[ \langle 0| \sigma_2 |1\rangle + \langle 1| \sigma_2 |0\rangle \right] \\
    &= \Tr_{\operatorname{data}} \left[ \frac{1}{2} \rho_k U^\dagger + \frac{1}{2} U \rho_k \right] \\
    &= \frac{1}{2} \left( \overline{\Tr[\rho_k U]} + \Tr[\rho_k U] \right) = \operatorname{Re}\!\left\{\Tr[\rho_k U]\right\}.\label{eq:hadamard_z_expectation}
\end{align}

Substituting $U = e^{-i t K_\mu}$ into~\eqref{eq:hadamard_z_expectation}, the expected outcome of a single circuit run yields exactly $\operatorname{Re}\!\left\{\Tr[\rho_k e^{-i t K_\mu}]\right\}$. 

\paragraph*{Step 2: Joint Monte Carlo Expectation.}
Algorithm~\ref{alg:BE-thermal-alg} constructs a Joint Monte Carlo estimator for each term $m$ by drawing $N_{\operatorname{shots}}^{(m)}$ independent pairs $(t, k)$ from the joint probability distribution $P(t, k) = p(t; m/T) p_k$, where $p(t; m/T)$ is the Cauchy distribution defined in~\eqref{eq:Cauchy_prob} and $p_k = |\alpha_{i,k}| / \left\|\alpha_i\right\|_1$. For each shot, the scaled classical estimator is $Y = \left\|\alpha_i\right\|_1 \operatorname{sgn}(\alpha_{i,k}) Z$. 

By the Law of Large Numbers, in the limit of infinite shots ($N_{\operatorname{shots}}^{(m)} \to \infty$), the sample average converges exactly to the expected value over this joint distribution:
\begin{align}
    \mathbb{E}_{(t,k)}[Y] &= \sum_k p_k \left\|\alpha_i\right\|_1 \operatorname{sgn}(\alpha_{i,k}) \mathbb{E}_{t} \left[ \langle Z \rangle \right] \nonumber \\
    &= \sum_k \alpha_{i,k} \operatorname{Re}\!\left\{ \Tr\!\left[ \rho_k \int_{-\infty}^{\infty} dt \ p(t; m/T) e^{-i t K_\mu} \right] \right\}. \label{eq:joint_expectation_integral}
\end{align}
To evaluate the Cauchy temporal integral, we use the fact that $K_\mu$ is Hermitian and thus admits a diagonal eigenbasis $K_\mu = \sum_j \lambda_j |\phi_j\rangle\!\langle\phi_j|$. Expanding the matrix exponential yields:
\begin{equation}
    \int_{-\infty}^{\infty} dt \ p(t; m/T) e^{-i t K_\mu} = \sum_{j} |\phi_j\rangle\!\langle\phi_j| \int_{-\infty}^{\infty} dt \ p(t; m/T) e^{-i t \lambda_j}.
\end{equation}
The scalar integral on the right is exactly the characteristic function (Fourier transform) of the Cauchy distribution, which evaluates analytically to $e^{-\frac{m}{T}|\lambda_j|}$. Since $K_\mu$ is strictly positive definite ($K_\mu > 0$), all eigenvalues are strictly positive ($\lambda_j > 0$), allowing us to drop the absolute value:
\begin{equation}
    \sum_{j} |\phi_j\rangle\!\langle\phi_j| e^{-\frac{m}{T}\lambda_j} = e^{-m K_\mu/T}.
\end{equation}
Because $e^{-m K_\mu / T}$ is a purely real, Hermitian operator (and $\rho_k$ is Hermitian), their trace is strictly real, meaning the $\operatorname{Re}\!\left\{\cdot\right\}$ wrapper in~\eqref{eq:joint_expectation_integral} acts as the identity. Thus, the joint temporal and spatial sampling perfectly recovers the unnormalized imaginary-time trace:
\begin{equation}
    \mathbb{E}_{(t,k)}[Y] = \sum_k \alpha_{i,k} \Tr\!\left[ \rho_k e^{-m K_\mu/T} \right] = \Tr\!\left[ Q_i e^{-m K_\mu/T} \right]. \label{eq:double_mc_result}
\end{equation}

\paragraph*{Step 3: Geometric Series Reconstruction.}
Finally, the algorithm sums these thermal components over the geometric series index $m$. In the limit of infinite series depth ($M \to \infty$), this sum exactly reconstructs the target operator trace:
\begin{equation}
    \sum_{m=1}^\infty \Tr\!\left[ Q_i e^{-m K_\mu/T} \right]= \Tr\!\left[ Q_i \left(\sum_{m=1}^\infty e^{-m K_\mu/T}\right) \right] = \Tr[X_T(\mu) Q_i].
\end{equation}
This completes the proof that the output of Algorithm~\ref{alg:BE-thermal-alg} provides an exact, unbiased estimate of the dual gradient component in the limit of infinite Joint Monte Carlo shots ($N_{\operatorname{shots}}^{(m)} \to \infty$ for all $m$) and infinite series depth ($M \to \infty$).

\subsection{Computational Complexity of Algorithm~\ref{alg:BE-thermal-alg}}
\label{app:complexity-analysis}

In this section, we derive the computational complexity of estimating the thermal trace $\Tr[X_T(\mu) Q_i]$ using Algorithm~\ref{alg:BE-thermal-alg}. To guarantee a total additive error bounded by $\varepsilon$, we partition the error budget equally among three sources of algorithmic approximation: the truncation of the infinite geometric series ($\varepsilon/3$), the truncation of the Cauchy time-evolution distribution ($\varepsilon/3$), and the statistical variance of the Monte Carlo sampling ($\varepsilon/3$).

\paragraph*{Step 1: Series Truncation.}
The target quantity is the infinite sum $G_i = \sum_{m=1}^\infty w_m$, where we have defined 
\begin{equation}
    w_m = \Tr[Q_i e^{-m K_\mu/T}].\label{eq:w_m}
\end{equation} 
To bound the truncation error $\sum_{m=M+1}^\infty w_m \leq \varepsilon/3$, we analyze the spectral decay. 
Because $Q_i = \sum \alpha_{i,k} \rho_k$, we can bound the magnitude of the $m$-th term:
\begin{equation}
    |w_m| \leq \sum_k |\alpha_{i,k}| \left| \Tr[\rho_k e^{-m K_\mu / T}] \right|.
\end{equation}
Since the minimum eigenvalue of the strictly positive dual slack operator $K_\mu$ is $\lambda_{\min} > 0$, the maximum eigenvalue of the positive operator $e^{-m K_\mu / T}$ is $e^{-m \lambda_{\min} / T}$. Because $\rho_k$ is a normalized density matrix ($\Tr[\rho_k]=1$), we have $\Tr[\rho_k e^{-m K_\mu / T}] \leq e^{-m \lambda_{\min} / T}$. Thus, the individual term is bounded by:
\begin{equation}\label{step:spectral_decay}
    |w_m| \leq \left\|\alpha_i\right\|_1 e^{-m \lambda_{\min} / T}.
\end{equation}
Summing the tail of this geometric bound from $m = M+1$ to infinity yields:
\begin{align}
    \text{Tail Error} &= \sum_{m=M+1}^\infty |w_m| \\
    & \leq \sum_{m=M+1}^\infty \left\|\alpha_i\right\|_1 e^{-m \lambda_{\min} / T}\\
    & = \left\|\alpha_i\right\|_1 \frac{e^{-(M+1)\lambda_{\min}/T}}{1 - e^{-\lambda_{\min}/T}} ,
\end{align}
where we used the formula for the sum of an infinite geometric series $\sum_{m=k}^{\infty} r^m = \frac{r^k}{1 - r}$. Setting this error bound to $\varepsilon/3$ and solving for the required truncation depth $M$ gives:
\begin{equation}
    M \geq \frac{T}{\lambda_{\min}} \ln\left( \frac{3 \left\|\alpha_i\right\|_1}{\varepsilon \left( 1 - e^{-\lambda_{\min}/T} \right)} \right) - 1. \label{eq:comp_M_exact}
\end{equation}
In the high-temperature regime ($T \gg \lambda_{\min}$), applying the first-order approximation $1 - e^{-\lambda_{\min}/T} \approx \lambda_{\min}/T$ yields the worst-case asymptotic scaling $M = \mathcal{O}\left( \frac{T}{\lambda_{\min}} \ln(1/\varepsilon) \right)$. In the low-temperature regime ($T \lesssim \lambda_{\min}$), the series converges significantly faster. We define $M$ as the ceiling of the bound in~\eqref{eq:comp_M_exact} to safely maintain the error budget. The algorithm thus evaluates the finite sum $\tilde{G}_i = \sum_{m=1}^M w_m$, requiring each of the $M$ components to be estimated to an individual precision of $\varepsilon / (3M)$.

\paragraph*{Step 2: Cauchy Time Truncation.}
For a specific integer $m$, the component $w_m$ is evaluated by integrating the unitary evolution over a Cauchy distribution $p(t; \tau_m)$ with scale parameter $\tau_m = m/T$. Because the scale parameter grows linearly with both the series index $m$ and the inverse temperature $1/T$, the Cauchy distribution becomes progressively flatter and more heavily tail-dominated in these regimes (see Figure~\ref{fig:cauchy_tails}). This increases the probability of sampling macroscopically large evolution times $t$, which inflates the depth of the Hamiltonian simulation circuit. Therefore, to ensure the quantum runtime remains strictly bounded, we truncate this distribution at a maximum simulation time $t_{\max}^{(m)}$.

\begin{figure}
    \centering
    \includegraphics[width=0.5\columnwidth]{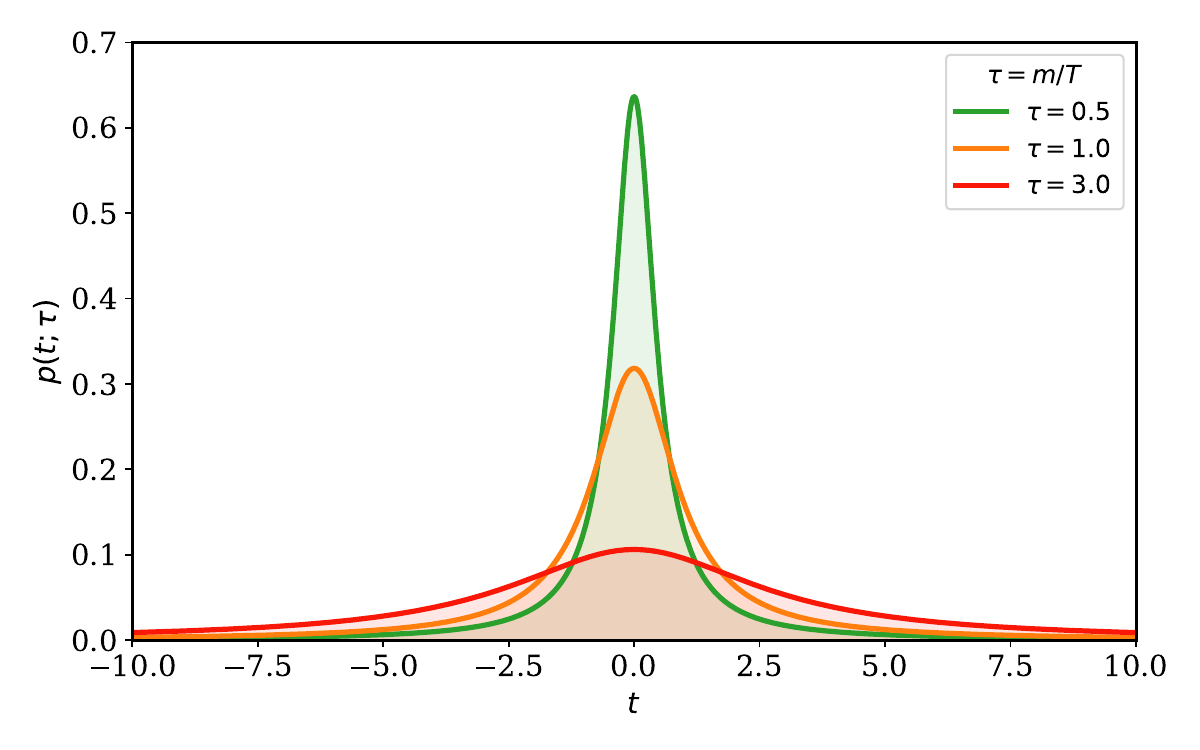}
    \caption{The probability density function of the Cauchy distribution $p(t; \tau)$ for varying scale parameters $\tau = m/T$. As the series index $m$ increases or temperature $T$ decreases, the distribution flattens, increasing the probability of sampling large evolution times $t$ and necessitating the maximum simulation cutoff $t_{\max}^{(m)}$.}
    \label{fig:cauchy_tails}
\end{figure}

Because the trace $\Tr[Q_i e^{-i t K_\mu}]$ is strictly bounded by $\left\|\alpha_i\right\|_1$, the absolute error introduced by discarding the tails of the Cauchy distribution ($|t| > t_{\max}^{(m)}$) scales as:
\begin{equation}\label{step:err_t_trunc}
    \text{Bias}_{\text{trunc}} \approx 2 \left\|\alpha_i\right\|_1 \int_{t_{\max}^{(m)}}^\infty p(t;\tau_m) dt \approx 2 \left\|\alpha_i\right\|_1 \frac{\tau_m}{\pi t_{\max}^{(m)}} = \frac{2 \left\|\alpha_i\right\|_1 m}{\pi T t_{\max}^{(m)}}.
\end{equation}
To restrict this bias to $\mathcal{O}(\varepsilon / M)$ for each term, the maximum simulation time must scale as:
\begin{equation}
    t_{\max}^{(m)} = \mathcal{O}\left( \frac{\left\|\alpha_i\right\|_1 m M}{T \varepsilon} \right). \label{eq:comp_xmax}
\end{equation}

\paragraph*{Step 3: Quantum Gate Complexity per Shot.}
To execute the Hadamard test circuit in Figure~\ref{fig:BE_circuit} for a sampled time $t$, the quantum computer must simulate the unitary $U = e^{-i t K_\mu}$. Using optimal Hamiltonian simulation algorithms, such as Qubitization or advanced Trotter formulas~\cite{Low2019hamiltonian}, the number of quantum gates required scales linearly with the simulation time and the 1-norm of the Hamiltonian coefficients $\left\|h\right\|_1$.
Thus, for a given $m$, the worst-case gate depth for a single circuit execution is:
\begin{equation}
    N_{\text{gates}}^{(m)} = \mathcal{O}\left( \left\|h\right\|_1 t_{\max}^{(m)} \right) = \mathcal{O}\left( \left\|h\right\|_1 \left\|\alpha_i\right\|_1 \frac{m M}{T \varepsilon} \right). \label{eq:comp_gates}
\end{equation}

\paragraph*{Step 4: Monte Carlo Sample Complexity.}
To estimate $w_m$, we construct a Joint Monte Carlo estimator. We simultaneously sample a state index $k$ and a time $t$ from the truncated Cauchy distribution. 
The single-shot observable is the scaled Pauli $Z$ measurement: $Y = \left\|\alpha_i\right\|_1 \operatorname{sgn}(\alpha_{i,k}) Z$. Because $Z \in [-1, 1]$, the variance of this single-shot estimator is strictly bounded by $\text{Var}(Y) \leq \left\|\alpha_i\right\|_1^2$.
By Hoeffding's inequality, the number of independent shots $N_{\operatorname{shots}}^{(m)}$ required to suppress the statistical variance of the mean to the target precision $\mathcal{O}(\varepsilon/M)$ scales as:
\begin{equation}
    N_{\operatorname{shots}}^{(m)} = \mathcal{O}\left( \frac{\text{Var}(Y)}{(\varepsilon/M)^2} \right) = \mathcal{O}\left( \frac{M^2 \left\|\alpha_i\right\|_1^2}{\varepsilon^2} \right). \label{eq:comp_shots}
\end{equation}

\paragraph*{Step 5: Total Asymptotic Complexity.}
The total algorithmic runtime is the sum of the gate operations across all shots for all $M$ terms in the geometric series:
\begin{equation}
    \text{Total Runtime} = \sum_{m=1}^M N_{\operatorname{shots}}^{(m)} \times N_{\text{gates}}^{(m)}.
\end{equation}
Substituting equations~\eqref{eq:comp_gates} and~\eqref{eq:comp_shots} yields:
\begin{align}
    \text{Total Runtime} &= \sum_{m=1}^M \mathcal{O}\left( \frac{M^2 \left\|\alpha_i\right\|_1^2}{\varepsilon^2} \right) \mathcal{O}\left( \left\|h\right\|_1 \left\|\alpha_i\right\|_1 \frac{m M}{T \varepsilon} \right) \nonumber \\
    &= \mathcal{O}\left( \frac{M^3 \left\|\alpha_i\right\|_1^3 \left\|h\right\|_1}{T \varepsilon^3} \right) \sum_{m=1}^M m \nonumber \\
    &= \mathcal{O}\left( \frac{M^5 \left\|\alpha_i\right\|_1^3 \left\|h\right\|_1}{T \varepsilon^3} \right).
\end{align}
Finally, substituting the asymptotic worst-case series truncation depth $M = \widetilde{\mathcal{O}}(T / \lambda_{\min})$ derived from the bound in~\eqref{eq:comp_M_exact}, we obtain the final asymptotic scaling:
\begin{equation}
    \text{Total Runtime} = \widetilde{\mathcal{O}}\left( \left\|\alpha_i\right\|_1^3 \left\|h\right\|_1 \frac{T^4}{\lambda_{\min}^5 \varepsilon^3} \right),
\end{equation}
where the $\widetilde{\mathcal{O}}$ notation suppresses the logarithmic dependence on $\varepsilon$, $\left\|\alpha_i\right\|_1$, and $\lambda_{\min}/T$. Thus, the algorithm evaluates the Bose--Einstein trace with a polynomial scaling in all physical parameters.

\subsection{Proof of Correctness for the Hessian Estimation Algorithm}
\label{sec:proof-hessian-algorithm}

In this section, we prove that the quantum circuit in Figure~\ref{fig:hessian_circuit} yields an unbiased estimate of the Hessian matrix element $H_{ij}$. The proof proceeds in three steps: (1) state evolution of the controlled-SWAP circuit, (2) the Joint Monte Carlo expectation over time, continuous parameter $s$, and state indices, and (3) the geometric series reconstruction.

\paragraph*{Step 1: State Evolution of the Controlled-SWAP Circuit.}
For a given circuit execution, assume a fixed tuple of variables $(s, t_1, t_2, k, l)$. We initialize the two data registers in the corresponding density matrices $\rho_k$ and $\sigma_l$. 
Let $U_1 = e^{-i t_1 K_\mu}$ and $U_2 = e^{-i t_2 K_\mu}$ denote the unitary time-evolution operators applied to the second and first data registers, respectively. The initial state of the tripartite system (the control qubit and the two data registers) is the product state:
\begin{equation}
    \sigma_0 = |0\rangle\!\langle 0| \otimes \rho_k \otimes \sigma_l.
\end{equation}
Applying the first Hadamard gate to the control qubit creates a coherent superposition:
\begin{equation}
    \sigma_1 = (H \otimes I \otimes I) \sigma_0 (H \otimes I \otimes I) = |+\rangle\!\langle +| \otimes \rho_k \otimes \sigma_l.
\end{equation}
Next, we apply the controlled-SWAP gate $CS = |0\rangle\!\langle 0| \otimes I \otimes I + |1\rangle\!\langle 1| \otimes \text{SWAP}$. Using the property that $\text{SWAP}(\rho_k \otimes \sigma_l)\text{SWAP} = \sigma_l \otimes \rho_k$, the state updates to $\sigma_2 = CS \sigma_1 CS^\dagger$:
\begin{align}
    \sigma_2 = \frac{1}{2} \Big( |0\rangle\!\langle 0| \otimes \rho_k \otimes \sigma_l + |0\rangle\!\langle 1| \otimes (\rho_k \otimes \sigma_l)\text{SWAP} + |1\rangle\!\langle 0| \otimes \text{SWAP}(\rho_k \otimes \sigma_l) + |1\rangle\!\langle 1| \otimes \sigma_l \otimes \rho_k \Big).
\end{align}
We then apply the controlled-Hamiltonian evolutions. Let $V = U_2 \otimes U_1$ represent the joint unitary operation on the data registers. The controlled-unitary is $CV = |0\rangle\!\langle 0| \otimes I \otimes I + |1\rangle\!\langle 1| \otimes V$. The state becomes $\sigma_3 = CV \sigma_2 CV^\dagger$:
\begin{align}
    \sigma_3 = \frac{1}{2} \Big( |0\rangle\!\langle 0| \otimes \rho_k \otimes \sigma_l + |0\rangle\!\langle 1| \otimes (\rho_k \otimes \sigma_l)\text{SWAP} V^\dagger + |1\rangle\!\langle 0| \otimes V \text{SWAP} (\rho_k \otimes \sigma_l) + |1\rangle\!\langle 1| \otimes V(\sigma_l \otimes \rho_k)V^\dagger \Big).
\end{align}
Finally, we apply the second Hadamard gate to the control qubit and measure it in the computational $Z$-basis. 
The expectation value of the Pauli $Z$ measurement is given by $\langle Z \rangle = \Tr[(Z \otimes I \otimes I)\sigma_3]$. Utilizing the identity $HZH = X$, we can evaluate this directly from $\sigma_3$:
\begin{equation}
    \langle Z \rangle = \Tr[(X \otimes I \otimes I)\sigma_3] = \Tr_{\operatorname{data}} \left[ \langle 0| \sigma_3 |1\rangle + \langle 1| \sigma_3 |0\rangle \right].
\end{equation}
Extracting the off-diagonal blocks from $\sigma_3$ yields:
\begin{align}
    \langle Z \rangle &= \Tr_{\operatorname{data}} \left[ \frac{1}{2} (\rho_k \otimes \sigma_l)\text{SWAP} V^\dagger + \frac{1}{2} V \text{SWAP} (\rho_k \otimes \sigma_l) \right] \\
    &= \operatorname{Re}\!\left\{\Tr_{\operatorname{data}}[ V \text{SWAP} (\rho_k \otimes \sigma_l) ]\right\}. \label{eq:hessian_pre_swap}
\end{align}
To evaluate this trace, we substitute $V = U_2 \otimes U_1$ and employ the trace identity for bipartite swap operations: $\Tr[(A \otimes B)\text{SWAP}(C \otimes D)] = \Tr[A D B C]$. Applying this identity to~\eqref{eq:hessian_pre_swap} yields:
\begin{equation}
    \langle Z \rangle = \operatorname{Re}\!\left\{\Tr[U_2 \sigma_l U_1 \rho_k]\right\}.
\end{equation}
By the cyclic property of the trace, this is identically $\operatorname{Re}\!\left\{\Tr[U_1 \rho_k U_2 \sigma_l]\right\}$. Substituting the explicit time-evolution operators, the expected outcome of a single circuit run is exactly:
\begin{equation}
    \langle Z \rangle = \operatorname{Re}\!\left\{\Tr[e^{-i t_1 K_\mu} \rho_k e^{-i t_2 K_\mu} \sigma_l]\right\}.
\end{equation}

\paragraph*{Step 2: Joint Monte Carlo Expectation.}
Algorithm~\ref{alg:hessian-est-alg} constructs a Joint Monte Carlo estimator for each index pair $(m_1, m_2)$ by drawing $N_{\operatorname{shots}}^{(m_1, m_2)}$ independent tuples from the joint probability distribution of time, state indices, and the interpolation parameter $s$. For each shot, the scaled classical estimator is $Y = \left\|\alpha_i\right\|_1 \left\|\alpha_j\right\|_1 \operatorname{sgn}(\alpha_{i,k}\alpha_{j,l}) Z$.

By the Law of Large Numbers, in the limit of infinite shots ($N_{\operatorname{shots}}^{(m_1, m_2)} \to \infty$), the sample average converges exactly to the expected value over this joint distribution. Summing over the index probabilities $p_i(k)$ and $p_j(l)$ cancels the normalization norms and restores the operator coefficients:
\begin{equation}
    \mathbb{E}[Y] = \int_{0}^{1} ds \operatorname{Re}\!\left\{ \Tr\!\left[ \left(\int_{-\infty}^{\infty} dt_1 \ p(t_1) e^{-i t_1 K_\mu}\right) Q_i \left(\int_{-\infty}^{\infty} dt_2 \ p(t_2) e^{-i t_2 K_\mu}\right) Q_j \right] \right\}.
\end{equation}
Following the spectral decomposition logic established in Appendix~\ref{sec:proof-BE-algorithm}, the Fourier transform of the Cauchy distributions $p(t_1; \tau_1)$ and $p(t_2; \tau_2)$ evaluates analytically to the non-unitary thermal weights. Because $K_\mu > 0$, the temporal integrals perfectly map the unitary operators to $e^{-\tau_1 K_\mu}$ and $e^{-\tau_2 K_\mu}$ respectively:
\begin{equation}
    \mathbb{E}[Y] = \int_{0}^{1} ds \operatorname{Re}\!\left\{\Tr\!\left[ e^{-\tau_1(m_1,s) K_\mu} Q_i e^{-\tau_2(m_2,s) K_\mu} Q_j \right]\right\}. \label{eq:hessian_double_mc_result}
\end{equation}

\paragraph*{Step 3: Geometric Series Reconstruction.}
Because the true Hessian matrix $H_{ij}$ is the second derivative of a real-valued scalar objective function $f_T(\mu)$, its elements are strictly real. Therefore, the $\operatorname{Re}\!\left\{\cdot\right\}$ wrapper naturally aligns with the underlying physical quantity.
Finally, the algorithm sums the components across the truncated double geometric series $m_1, m_2 \in \{1, \dots, M\}$ and multiplies by the global factor $-1/T$. In the limit of infinite series depth ($M \to \infty$), this exactly reconstructs the desired quantity for the Hessian matrix element $H_{ij}$:
\begin{equation}
    -\frac{1}{T} \int_{0}^{1}ds \sum_{m_1=1}^{\infty} \sum_{m_2=1}^{\infty} \Tr\!\left[ e^{-\frac{m_1-s}{T} K_\mu} Q_i e^{-\frac{m_2-1+s}{T} K_\mu} Q_j \right] = -\frac{1}{T}\int_{0}^{1}ds\,\Tr\!\left[X_{T}(\mu,s)Q_{i}X_{T}(\mu,1-s)Q_{j}\right].
\end{equation}

\subsection{Computational Complexity of Algorithm~\ref{alg:hessian-est-alg}}
\label{app:complexity_hessian}

In this section, we derive the computational complexity of estimating the Hessian matrix element $H_{ij}$ using Algorithm~\ref{alg:hessian-est-alg}. Because the algorithmic structure parallels the gradient estimator, we extend the error partitioning logic from Appendix~\ref{app:complexity-analysis} to the double geometric series. We allocate the total additive error budget $\varepsilon$ equally among the series truncation ($\varepsilon/3$), the Cauchy time truncation ($\varepsilon/3$), and the Monte Carlo statistical variance ($\varepsilon/3$).

\paragraph*{Step 1: Series Truncation.} The target quantity is the double infinite sum $H_{ij} = -\frac{1}{T} \int_0^1 ds \sum_{m_1=1}^{\infty} \sum_{m_2=1}^{\infty} w_{m_1, m_2}(s)$, where we define the integrand weights as:
\begin{equation}
    w_{m_1, m_2}(s) \coloneqq \Tr\!\left[e^{-\frac{m_1-s}{T}K_\mu} Q_i e^{-\frac{m_2-1+s}{T}K_\mu} Q_j\right].
\end{equation}
Under the Quantum State Model introduced for this algorithm, the Hermitian operators are decomposed into linear combinations of density matrices: $Q_i = \sum_k \alpha_{i,k} \rho_k$ and $Q_j = \sum_l \alpha_{j,l} \sigma_l$. By expanding $w_{m_1, m_2}(s)$ using these decompositions and following the spectral decay logic established in~\eqref{step:spectral_decay}, the magnitude of the joint thermal weights decays exponentially with the minimum eigenvalue $\lambda_{\min}$ of the dual slack operator $K_\mu$. Because the sum of the scaled exponent parameters is $\frac{m_1-s}{T} + \frac{m_2-1+s}{T} = \frac{m_1+m_2-1}{T}$, the terms are bounded uniformly for all $s \in [0,1]$ by:
\begin{equation}
    |w_{m_1, m_2}(s)| \leq \left\|\alpha_i\right\|_1 \left\|\alpha_j\right\|_1 e^{-(m_1+m_2-1)\lambda_{\min}/T},
\end{equation}
where $\left\|\alpha_i\right\|_1 = \sum_k |\alpha_{i,k}|$ and $\left\|\alpha_j\right\|_1 = \sum_l |\alpha_{j,l}|$ are the 1-norms of the state coefficients. The tail of the double geometric series consists of two two-dimensional strips: $\{m_{1}>M,m_{2}\geq 1\}$ and $\{m_{1}\leq M,m_{2}>M\}$. Each of these strips contributes
\begin{equation}
    \sum_{m_{1}>M}\sum_{m_{2}\geq 1}\left\|\alpha_i\right\|_1\left\|\alpha_j\right\|_1\,e^{-(m_{1}+m_{2}-1)\lambda_{\min}/T} = \frac{\left\|\alpha_i\right\|_1\left\|\alpha_j\right\|_1\,e^{-M\lambda_{\min}/T}}{(1-e^{-\lambda_{\min}/T})^{2}}\cdot e^{\lambda_{\min}/T},
\end{equation}
and similarly for the second strip. Setting the sum of these two contributions to $\varepsilon/3$ and solving for $M$ yields, up to a logarithmic factor that we suppress in the $\widetilde{\mathcal{O}}$ notation,
\begin{equation}
    M = \widetilde{\mathcal{O}}\left( \frac{T}{\lambda_{\min}} \ln(1/\varepsilon) \right).
\end{equation}
The algorithm thus evaluates a finite grid of $M^2$ terms. To maintain the overall error budget, each of the $M^2$ individual components must be estimated to a target precision of $\mathcal{O}(\varepsilon/M^2)$.

\paragraph*{Step 2: Cauchy Time Truncation.} Each component involves integrating two unitary evolutions over Cauchy distributions $p(t_1;\tau_1)$ and $p(t_2;\tau_2)$. The scale parameters $\tau_1, \tau_2$ are bounded by $\mathcal{O}(M/T)$. To restrict the bias introduced by discarding the heavy tails to $\mathcal{O}(\varepsilon/M^2)$ for each term, the maximum simulation times $t_{\max}^{(m_1, m_2)}$ must be extended. Paralleling the derivation in~\eqref{step:err_t_trunc}, the cutoff scales as:
\begin{equation}
    t_{\max} = \mathcal{O}\left( \frac{\left\|\alpha_i\right\|_1 \left\|\alpha_j\right\|_1 M^3}{T \varepsilon} \right).
\end{equation}

\paragraph*{Step 3: Quantum Gate Complexity per Shot.} For a sampled time pair $(t_1, t_2)$, the quantum computer simulates the unitaries $U_1$ and $U_2$. Using optimal Hamiltonian simulation, the gate depth scales linearly with the simulation time and the 1-norm of the Hamiltonian coefficients $\left\|h\right\|_1$. The worst-case gate depth for a single execution of the circuit in Figure~\ref{fig:hessian_circuit} is:
\begin{equation}
    N_{\text{gates}}^{(m_1, m_2)} = \mathcal{O}\left( \left\|h\right\|_1 t_{\max} \right) = \mathcal{O}\left( \left\|h\right\|_1 \left\|\alpha_i\right\|_1 \left\|\alpha_j\right\|_1 \frac{M^3}{T\varepsilon} \right).
\end{equation}

\paragraph*{Step 4: Monte Carlo Sample Complexity.} The single-shot observable is $Y = \left\|\alpha_i\right\|_1 \left\|\alpha_j\right\|_1 \mathrm{sgn}(\alpha_{i,k}\alpha_{j,l}) Z$. Because $Z \in \{-1, 1\}$, the variance is strictly bounded by $\mathrm{Var}(Y) \leq \left\|\alpha_i\right\|_1^2 \left\|\alpha_j\right\|_1^2$. By Hoeffding's inequality, the number of independent shots required to suppress the statistical variance of the mean to the target precision $\mathcal{O}(\varepsilon/M^2)$ scales as:
\begin{equation}
    N_{\operatorname{shots}}^{(m_1, m_2)} = \mathcal{O}\left( \frac{\mathrm{Var}(Y)}{(\varepsilon/M^2)^2} \right) = \mathcal{O}\left( \frac{M^4 \left\|\alpha_i\right\|_1^2 \left\|\alpha_j\right\|_1^2}{\varepsilon^2} \right).
\end{equation}

\paragraph*{Step 5: Total Asymptotic Complexity.} The total algorithmic runtime is the sum of the gate operations across all shots for all $M^2$ terms in the double series grid:
\begin{align}
    \text{Total Runtime} &= \sum_{m_1=1}^{M} \sum_{m_2=1}^{M} N_{\operatorname{shots}}^{(m_1, m_2)} \times N_{\text{gates}}^{(m_1, m_2)} \\
    &= M^2 \times \mathcal{O}\left( \frac{M^4 \left\|\alpha_i\right\|_1^2 \left\|\alpha_j\right\|_1^2}{\varepsilon^2} \right) \mathcal{O}\left( \left\|h\right\|_1 \left\|\alpha_i\right\|_1 \left\|\alpha_j\right\|_1 \frac{M^3}{T\varepsilon} \right) \\
    &= \mathcal{O}\left( \frac{M^9 \left\|\alpha_i\right\|_1^3 \left\|\alpha_j\right\|_1^3 \left\|h\right\|_1}{T \varepsilon^3} \right).
\end{align}
Substituting the asymptotic series truncation depth $M = \widetilde{\mathcal{O}}(T/\lambda_{\min})$, we obtain the final asymptotic scaling:
\begin{equation}
    \text{Total Runtime} = \widetilde{\mathcal{O}}\left( \left\|\alpha_i\right\|_1^3 \left\|\alpha_j\right\|_1^3 \left\|h\right\|_1 \frac{T^8}{\lambda_{\min}^9 \varepsilon^3} \right),
\end{equation}
where the $\widetilde{\mathcal{O}}$ notation suppresses logarithmic dependence on the physical parameters.

\section{Bose--Einstein relative entropy: proofs of properties}
\label{app:be_relative_entropy}

\subsection{Proof of Proposition~\ref{prop:BE-faithfulness} (Faithfulness and Support)}
\label{app:proof_BE_faithfulness}

\begin{proof}
We first prove the non-negativity $D_{\mathrm{BE}}(X\|Y) \geq 0$, with equality if and only if $X=Y$. As established in Remark~\ref{rem:bregman-formulation}, the Bose--Einstein relative entropy coincides with the matrix Bregman divergence generated by the trace function $F(X) = \Tr[f(X)] = -S_{\mathrm{BE}}(X)$, where the underlying scalar function is $f(x) \coloneqq x\ln x - (x+1)\ln(x+1)$.

As shown in Appendix~\ref{app:concavity_BEentropy} for the scalar bosonic entropy $g(x) = - f(x)$, the second derivative is $f''(x) = \frac{1}{x(x+1)}$. Because $f''(x) > 0$ for all $x > 0$, the function $f(x)$ is convex. The non-negativity of a matrix Bregman divergence generated by a convex trace function follows directly from Klein's inequality~\cite[Eq. 4]{Petz1994}. Specifically, Klein's inequality dictates that for any strictly convex, differentiable scalar function $f$ and positive semidefinite operators $A$ and $B$:
\begin{equation} \label{eq:klein-inequality}
    \Tr[f(A) - f(B) - (A-B)f'(B)] \geq 0,
\end{equation}
with equality if and only if $A=B$.

We now apply Klein's inequality by substituting $A=X$ and $B=Y$. By extending the scalar function $f$ and its derivative $f'$ to matrix operators, we recognize that the trace of $f(X)$ corresponds exactly to the negative Bose--Einstein entropy:
\begin{equation}
    \Tr[f(X)] = \Tr[X\ln X - (X+I)\ln(X+I)] = -S_{\mathrm{BE}}(X).
\end{equation}
Similarly, the matrix function for the derivative evaluated at $Y$ is $f'(Y) = \ln Y - \ln(Y+I)$. Substituting these operator identities into the left-hand side of~\eqref{eq:klein-inequality} yields:
\begin{align}
    & \Tr[f(X)] - \Tr[f(Y)] - \Tr[(X-Y)f'(Y)] \nonumber \\
    & = -S_{\mathrm{BE}}(X) + S_{\mathrm{BE}}(Y) - \Tr[(X-Y)(\ln Y - \ln(Y+I))] \\
    \begin{split} 
    &= -S_{\mathrm{BE}}(X) + \Tr[(Y+I)\ln(Y+I) - Y\ln Y]  - \Tr[X\ln Y - X\ln(Y+I) - Y\ln Y + Y\ln(Y+I)]
    \end{split} \\
    & = -S_{\mathrm{BE}}(X) + \Tr[X\ln(Y+I) - X\ln Y] + \Tr[\ln(Y+I)]\\
    & = -S_{\mathrm{BE}}(X) + \Tr[(X+I)\ln(Y+I) - X\ln Y].
\end{align}
This final expression coincides with the Bose--Einstein relative entropy $D_{\mathrm{BE}}(X\|Y)$, as defined in~\eqref{eq:BE_rel_entropy_main}. Therefore, $D_{\mathrm{BE}}(X\|Y) \geq 0$, with equality achieved if and only if $X=Y$.

To prove the support condition, suppose there exists an eigenvector $|\psi\rangle$ in the kernel of $Y$ (meaning $Y|\psi\rangle = 0$) such that $X|\psi\rangle = x|\psi\rangle$ with $x > 0$. Evaluating the term $\Tr[-X\ln Y]$ in the eigenbasis of $Y$ yields a component proportional to $-x \ln(0) \to \infty$. Therefore, $D_{\mathrm{BE}}(X\|Y)$ is finite if and only if the kernel of $Y$ is contained within the kernel of $X$, which is equivalent to stating that $\operatorname{supp}(X) \subseteq \operatorname{supp}(Y)$.
\end{proof}

\subsection{Proof of Proposition~\ref{prop:BE-isometric} (Unitary Invariance)}
\label{app:proof_BE_isometric}

\begin{proof}
Let $U$ be a unitary, i.e., $U^\dagger U = UU^\dagger = I$. Substituting $UXU^\dagger$ and $UYU^\dagger$ into the definition of the Bose--Einstein relative entropy in~\eqref{eq:BE_rel_entropy_main}, yields:
\begin{equation}\label{step:first_iso}
    D_{\mathrm{BE}}(UXU^\dagger\|UYU^\dagger) = -S_{\mathrm{BE}}(UXU^\dagger) + \Tr[(UXU^\dagger+I)\ln(UYU^\dagger+I) - UXU^\dagger\ln(UYU^\dagger)].
\end{equation}
Because the Bose--Einstein entropy, defined in~\eqref{eq:BE-entropy-def}, is evaluated as the trace of a matrix function, we can apply the identity $f(UXU^\dagger) = U f(X) U^\dagger$ and the cyclic property of the trace to obtain $S_{\mathrm{BE}}(UXU^\dagger) = \Tr[U f(X) U^\dagger] = \Tr[f(X) U^\dagger U] = S_{\mathrm{BE}}(X)$. Utilizing the same property for the matrix logarithm, we have $\ln(UYU^\dagger+I) = \ln(U(Y+I)U^\dagger) = U\ln(Y+I)U^\dagger$. Substituting this back into the trace term in the right-hand side of~\eqref{step:first_iso} gives:
\begin{equation}
    \Tr[(UXU^\dagger+I)\ln(UYU^\dagger+I) - UXU^\dagger\ln(UYU^\dagger)] = \Tr[UXU^\dagger U\ln(Y+I)U^\dagger + U\ln(Y+I)U^\dagger - UXU^\dagger U\ln Y U^\dagger].
\end{equation}
Using $U^\dagger U = I$ and the cyclic property of the trace ($\Tr[U A U^\dagger] = \Tr[A U^\dagger U] = \Tr[A]$), the unitaries cancel out, recovering the original formula for $D_{\mathrm{BE}}(X\|Y)$.
\end{proof}

\subsection{Proof of Proposition~\ref{prop:BE-spectral} (Spectral Expansion)}
\label{app:proof_BE_spectral}

\begin{proof}
Given the spectral decompositions $X = \sum_i x_i |\psi_i\rangle\!\langle\psi_i|$ and $Y = \sum_j y_j |\phi_j\rangle\!\langle\phi_j|$, we first expand the single-argument entropy $-S_{\mathrm{BE}}(X)$ in the definition of the Bose--Einstein relative entropy $D_{\mathrm{BE}}(X\|Y)$ in~\eqref{eq:BE_rel_entropy_main}. Because $\{|\phi_j\rangle\}$ forms a complete basis, $\sum_j |\phi_j\rangle\!\langle\phi_j| = I$. We can therefore insert this resolution of the identity:
\begin{align}
    -S_{\mathrm{BE}}(X) &= \Tr[X\ln X - (X+I)\ln(X+I)] \\
    &= \sum_{i} \left[ x_i\ln x_i - (x_i+1)\ln(x_i+1) \right] \langle\psi_i|\psi_i\rangle \\
    &= \sum_{i,j} \left[ x_i\ln x_i - (x_i+1)\ln(x_i+1) \right] \langle\psi_i|\phi_j\rangle\!\langle\phi_j|\psi_i\rangle \\
    &= \sum_{i,j} \left[ x_i\ln x_i - (x_i+1)\ln(x_i+1) \right] |\langle\psi_i|\phi_j\rangle|^2.
\end{align}
We now expand the cross terms $\Tr[(X+I)\ln(Y+I)]$ and $-\Tr[X\ln Y]$ in the joint eigenbasis:
\begin{align}
    \Tr[(X+I)\ln(Y+I)] &= \sum_{i,j} (x_i+1) \ln(y_j+1) |\langle\psi_i|\phi_j\rangle|^2, \\
    -\Tr[X\ln Y] &= \sum_{i,j} -x_i \ln(y_j) |\langle\psi_i|\phi_j\rangle|^2.
\end{align}
Summing these three components and grouping the scalar functions yields:
\begin{equation}
    D_{\mathrm{BE}}(X\|Y) = \sum_{i,j} \left( x_i\ln x_i - (x_i+1)\ln(x_i+1) + (x_i+1)\ln(y_j+1) - x_i\ln y_j \right) |\langle\psi_i|\phi_j\rangle|^2.
\end{equation}
By factoring the logarithms, the bracketed scalar term simplifies to $d_{\mathrm{BE}}(x_i \| y_j) = x_i \ln \frac{x_i}{y_j} + (x_i+1) \ln \frac{y_j+1}{x_i+1}$, proving the proposition.
\end{proof}

\subsection{Proof of Proposition~\ref{prop:BE-convexity-main} (Convexity and Data Processing)}
\label{app:proof_BE_convexity}

\begin{proof}
We first prove that $D_{\mathrm{BE}}(X\|Y)$ is strictly convex in $X$ for a fixed $Y$. The second term in the definition of the Bose--Einstein relative entropy in~\eqref{eq:BE_rel_entropy_main}, $X \mapsto \Tr[(X+I)\ln(Y+I) - X\ln Y]$, is linear in $X$. Because the Bose--Einstein entropy $S_{\mathrm{BE}}(X)$ is strictly concave in $X$ (see Appendix~\ref{app:concavity_BEentropy}), its negation $-S_{\mathrm{BE}}(X)$ is strictly convex. Since the sum of a strictly convex function and a linear function is strictly convex, $D_{\mathrm{BE}}(X\|Y)$ is strictly convex in $X$.

To prove the failure of joint convexity, it is sufficient to show that the underlying scalar Bose--Einstein relative entropy $d_{\mathrm{BE}}(x\|y)$, defined in~\eqref{eq:scalar_BE_rel_entr}, is not jointly convex. A multi-variable function is jointly convex if and only if its Hessian matrix $\mathcal{H}(x,y)$ is positive semidefinite for all $x, y > 0$, requiring the determinant $\det(\mathcal{H})$ to be non-negative everywhere. The second partial derivatives of $d_{\mathrm{BE}}(x\|y)$ with respect to $x$ and $y$ are:
\begin{align}
    \mathcal{H}_{xx} &= \frac{\partial^2 d_{\mathrm{BE}}}{\partial x^2} = \frac{1}{x} - \frac{1}{x+1} = \frac{1}{x(x+1)}, \\
    \mathcal{H}_{xy} &= \frac{\partial^2 d_{\mathrm{BE}}}{\partial x \partial y} = \frac{1}{y+1} - \frac{1}{y} = -\frac{1}{y(y+1)}, \\
    \mathcal{H}_{yy} &= \frac{\partial^2 d_{\mathrm{BE}}}{\partial y^2} = -\frac{x+1}{(y+1)^2} + \frac{x}{y^2} = \frac{2xy + x - y^2}{y^2(y+1)^2}.
\end{align}
Calculating the determinant $\det(\mathcal{H}) = \mathcal{H}_{xx}\mathcal{H}_{yy} - (\mathcal{H}_{xy})^2$:
\begin{align}
    \det(\mathcal{H}) &= \left[ \frac{1}{x(x+1)} \right] \left[ \frac{2xy + x - y^2}{y^2(y+1)^2} \right] - \left[ \frac{-1}{y(y+1)} \right]^2 \nonumber \\
    &= \frac{2xy + x - y^2 - x(x+1)}{x(x+1)y^2(y+1)^2} = -\frac{(x-y)^2}{x(x+1)y^2(y+1)^2}.
\end{align}
Because the denominator is strictly positive, and the numerator $-(x-y)^2$ is strictly negative for any $x \neq y$, the determinant is strictly negative. An indefinite Hessian implies the presence of saddle points, confirming that $d_{\mathrm{BE}}(x\|y)$, and consequently $D_{\mathrm{BE}}(X\|Y)$, is not jointly convex.

Finally, we demonstrate how the failure of joint convexity guarantees the failure of the data-processing inequality (DPI). We do this by showing that if a divergence were to satisfy the DPI, it would be forced to be jointly convex. Consider the CPTP map $\Phi$ corresponding to the partial trace over an orthogonal ancillary classical register $A$: $\Phi(X_{AB}) = \Tr_A(X_{AB})$. Let us construct the classical-quantum block-diagonal operators $X_{AB} = p |0\rangle\!\langle 0| \otimes X_1 + (1-p) |1\rangle\!\langle 1| \otimes X_2$ and $Y_{AB} = p |0\rangle\!\langle 0| \otimes Y_1 + (1-p) |1\rangle\!\langle 1| \otimes Y_2$ for some scalar $p \in [0,1]$. If the generalized DPI holds, then $D_{\mathrm{BE}}(\Phi(X_{AB})\|\Phi(Y_{AB})) \leq D_{\mathrm{BE}}(X_{AB}\|Y_{AB})$. The partial trace yields the convex combinations $\Phi(X_{AB}) = p X_1 + (1-p) X_2$ and $\Phi(Y_{AB}) = p Y_1 + (1-p) Y_2$. Concurrently, due to the direct-sum additivity of the Bose--Einstein relative entropy (Proposition~\ref{prop:BE-additivity}), the divergence of the block-diagonal states evaluates to $p D_{\mathrm{BE}}(X_1\|Y_1) + (1-p) D_{\mathrm{BE}}(X_2\|Y_2)$. Thus, substituting these into the DPI inequality yields:
\begin{equation}
    D_{\mathrm{BE}}(p X_1 + (1-p) X_2 \| p Y_1 + (1-p) Y_2) \leq p D_{\mathrm{BE}}(X_1\|Y_1) + (1-p) D_{\mathrm{BE}}(X_2\|Y_2),
\end{equation}
which coincides with the mathematical definition of joint convexity. Therefore, because we have proven that $D_{\mathrm{BE}}(X\|Y)$ is not jointly convex, it cannot satisfy the data-processing inequality under arbitrary CPTP maps.
\end{proof}   

\begin{remark}[Bregman Divergences and Joint Convexity]
The failure of joint convexity highlights the distinction between the Bose--Einstein relative entropy and the Fermi--Dirac relative entropy $D_{\mathrm{FD}}(M_1\|M_2)$ defined in~\cite[Definition 27]{liu2026sdp_fermi}. Both divergences belong to the class of matrix Bregman divergences~\cite{Dhillon2007Matrix}. A Bregman divergence generated by a strictly convex scalar function $f(x)$ is jointly convex if and only if the reciprocal of its second derivative, $1/f''(x)$, is a concave function~\cite[Theorem 3.3]{Bauschke2001Joint}. For fermions, $f_{\mathrm{FD}}(x) = x\ln x + (1-x)\ln(1-x)$, which gives $1/f_{\mathrm{FD}}''(x) = x - x^2$ (a strictly concave parabola), geometrically guaranteeing joint convexity~\cite[Example 3.5]{Bauschke2001Joint}. However, for bosons, the entropy $f_{\mathrm{BE}}(x) = -g(x) = x\ln x - (x+1)\ln(x+1)$ yields $1/f_{\mathrm{BE}}''(x) = x^2 + x$ (see the derivation in~\eqref{step:second_der_g(x)} in Appendix~\ref{app:concavity_BEentropy}), which is strictly convex, thus precluding joint convexity.
\end{remark}

\subsection{Proof of Proposition~\ref{prop:BE-additivity} (Additivity under Direct Sums)}
\label{app:proof_BE_additivity}

\begin{proof}
Let $X \oplus Z$ and $Y \oplus W$ be block-diagonal matrices.
Applying the property that matrix functions acting on block-diagonal operators evaluate block-wise~\cite{Higham2008functions} ($f(X \oplus Z) = f(X) \oplus f(Z)$) to the matrix logarithm and taking the trace yields the additivity of the scalar entropy:
\begin{equation}
    -S_{\mathrm{BE}}(X \oplus Z) = -S_{\mathrm{BE}}(X) - S_{\mathrm{BE}}(Z).
\end{equation}
For the cross terms in the relative entropy in~\eqref{eq:BE_rel_entropy_main}, the product of two block-diagonal matrices is the direct sum of their products, and the trace of a direct sum is the sum of the traces, thus giving:
\begin{align}
    \Tr[((X \oplus Z)+I)\ln((Y \oplus W)+I)] &= \Tr[(X+I)\ln(Y+I) \oplus (Z+I)\ln(W+I)] \\
    &= \Tr[(X+I)\ln(Y+I)] + \Tr[(Z+I)\ln(W+I)].
\end{align}
Following the same logic for the $-X\ln Y$ term and grouping the components corresponding to the respective subspaces recovers $D_{\mathrm{BE}}(X\|Y) + D_{\mathrm{BE}}(Z\|W)$.
\end{proof}

\subsection{Proof of Theorem~\ref{thm:BE-affine-monotonicity} (Affine Monotonicity) and Lemma~\ref{lem:BE-integral-representation}}
\label{app:proof_BE_affine_monotonicity}

We prove the integral representation first, then the scalar contraction inequality, and finally lift to the operator level.

\paragraph*{Integral representation (Lemma~\ref{lem:BE-integral-representation}).}
For any twice-differentiable convex function $f$, the Bregman divergence $D_f(x,y) = f(x) - f(y) - f'(y)(x-y)$ admits the standard integral form
\begin{equation}
D_f(x,y) = (x-y)^2 \int_0^1 (1-t)\,f''(y + t(x-y))\,dt.
\end{equation}
Applied to $f(x) = x \ln x - (x+1)\ln(x+1)$, whose derivatives are
\begin{equation}
f'(x) = \ln \frac{x}{x+1}, \qquad f''(x) = \frac{1}{x(x+1)},
\end{equation}
a direct calculation verifies that $D_f(x,y)$ coincides with the scalar divergence $d_{\mathrm{BE}}(x\|y)$ defined in~\eqref{eq:scalar_BE_rel_entr}. Substituting $f''$ into the integral form yields~\eqref{eq:BE-integral-rep}.

\paragraph*{Scalar contraction under affine maps (Theorem~\ref{thm:BE-affine-monotonicity}).}
Let $\Lambda(z) = az + b$ with $a, b \geq 0$ denote a generic affine map (we use $\Lambda$ rather than $T$ here to avoid clashing with the temperature). Applying~\eqref{eq:BE-integral-rep} to $\Lambda(x), \Lambda(y)$ gives
\begin{equation}
d_{\mathrm{BE}}(\Lambda(x) \| \Lambda(y)) = a^2 (x-y)^2 \int_0^1 \frac{1-t}{\Lambda(z_t)\,(\Lambda(z_t)+1)}\,dt,
\end{equation}
so comparing $d_{\mathrm{BE}}(x\|y)$ and $d_{\mathrm{BE}}(\Lambda(x)\|\Lambda(y))$ reduces to comparing the integrands pointwise in $t$. It is therefore sufficient to prove
\begin{equation}
\frac{a^2}{\Lambda(z)\,(\Lambda(z)+1)} \leq \frac{1}{z(z+1)} \qquad \forall z > 0,
\end{equation}
or equivalently $\Lambda(z)(\Lambda(z)+1) \geq a^2 z(z+1)$. Expanding,
\begin{equation}
(az+b)(az+b+1) - a^2 z(z+1) = b(b+1) + az(2b+1-a),
\end{equation}
which is non-negative for all $z \geq 0$ whenever $2b+1 \geq a$. This establishes the scalar form
\begin{equation}\label{eq:BE-scalar-affine}
d_{\mathrm{BE}}(x\|y) \geq d_{\mathrm{BE}}(ax+b\|ay+b).
\end{equation}

\paragraph*{Operator-level lift.}
By the spectral expansion of Proposition~\ref{prop:BE-spectral},
\begin{equation}
D_{\mathrm{BE}}(X\|Y) = \sum_{i,j} d_{\mathrm{BE}}(x_i \| y_j)\,|\langle\psi_i|\phi_j\rangle|^2.
\end{equation}
The affine map $Z \mapsto aZ + bI$ leaves the eigenvectors of $X$ and $Y$ invariant and shifts their eigenvalues by $z \mapsto az + b$. The spectral expansion of $D_{\mathrm{BE}}(aX + bI\|aY + bI)$ is therefore the same sum with each $d_{\mathrm{BE}}(x_i\|y_j)$ replaced by $d_{\mathrm{BE}}(ax_i+b\|ay_j+b)$. Applying~\eqref{eq:BE-scalar-affine} term by term yields~\eqref{eq:BE-affine-monotonicity-bnd}.

\paragraph*{Verification of the three named cases.}
The three specializations listed after Theorem~\ref{thm:BE-affine-monotonicity} all satisfy the condition $2b+1 \geq a$:
\begin{itemize}
    \item Attenuator: $2b+1-a = 2(1-\eta)N + 1 - \eta = (1-\eta)(2N+1) \geq 0$.
    \item Amplifier: $2b+1-a = 2(G-1)(N+1) + 1 - G = (G-1)(2N+1) \geq 0$.
    \item Additive noise: $2b+1-a = 2N \geq 0$.
\end{itemize}

\subsection{Proof of Proposition~\ref{prop:BE-fisher-info} (Bose--Einstein Fisher Information)}
\label{sec:proof-be-fisher}

\begin{proof}
To derive the expressions for the Fisher information matrix, we first evaluate the relative entropy $D_{\mathrm{BE}}(X_{T}(\mu)\|X_{T}(\mu+\varepsilon))$ using Definition~\ref{def:BE-rel-entropy}.
Let $X = X_T(\mu)$ and $Y = X_T(\mu+\varepsilon) = (e^{\frac{1}{T}(H-(\mu+\varepsilon)\cdot Q)}-I)^{-1}$.
Let us express the matrix logarithms of $Y$ in terms of $H$ and $Q$:
\begin{equation}
    \ln(Y+I) - \ln Y = \ln\left( (Y+I)Y^{-1} \right) = \ln\left( I + Y^{-1} \right) = \frac{1}{T}\big(H - (\mu+\varepsilon)\cdot Q\big).
\end{equation}
Using this identity, the trace term in the relative entropy becomes:
\begin{align}
    \Tr[(X+I)\ln(Y+I) - X\ln Y] &= \Tr[X(\ln(Y+I) - \ln Y)] + \Tr[\ln(Y+I)] \\
    &= \frac{1}{T}\Tr\!\left[X_T(\mu) \big(H - (\mu+\varepsilon)\cdot Q\big)\right] - \Tr\!\left[\ln\!\left(I - e^{-\frac{1}{T}(H-(\mu+\varepsilon)\cdot Q)}\right)\right],
\end{align}
where we used the identity $Y+I = (I - e^{-\frac{1}{T}(H-(\mu+\varepsilon)\cdot Q)})^{-1}$ to simplify the second trace.
Substituting this back into the relative entropy yields:
\begin{align}
    D_{\mathrm{BE}}(X_{T}(\mu)\|X_{T}(\mu+\varepsilon)) &= -S_{\mathrm{BE}}(X_T(\mu)) + \frac{1}{T}\Tr\!\left[X_T(\mu) \big(H - (\mu+\varepsilon)\cdot Q\big)\right] \nonumber \\
    &\qquad - \Tr\!\left[\ln\!\left(I - e^{-\frac{1}{T}(H-(\mu+\varepsilon)\cdot Q)}\right)\right].
\end{align}

We now take the second partial derivatives with respect to the perturbation variables $\varepsilon_i, \varepsilon_j$ evaluated at $\varepsilon=0$.
The first term $-S_{\mathrm{BE}}(X_T(\mu))$ is independent of $\varepsilon$.
The second term is linear in $\varepsilon$, meaning its second derivative is strictly zero.
The only surviving term is the third one:
\begin{align}
    \left[I(\mu)\right]_{i,j} &= \left.\frac{\partial^{2}}{\partial\varepsilon_{i}\partial\varepsilon_{j}}D_{\mathrm{BE}}(X_{T}(\mu)\|X_{T}(\mu+\varepsilon))\right|_{\varepsilon=0} \nonumber \\
    &= -\left.\frac{\partial^{2}}{\partial\varepsilon_{i}\partial\varepsilon_{j}}\Tr\!\left[\ln\!\left(I - e^{-\frac{1}{T}(H-(\mu+\varepsilon)\cdot Q)}\right)\right]\right|_{\varepsilon=0}.
\end{align}
Comparing this to the dual objective function $f_T(\mu)$ defined in~\eqref{eq:dual_obj_function}, we recognize that the trace term is exactly $\frac{1}{T}(f_T(\mu+\varepsilon) - (\mu+\varepsilon)\cdot q)$. Because the linear term $(\mu+\varepsilon)\cdot q$ vanishes under the second derivative, we have:
\begin{equation}
    \left[I(\mu)\right]_{i,j} = -\frac{1}{T} \frac{\partial^2}{\partial\mu_i\partial\mu_j} f_T(\mu) = -\frac{1}{T} [\nabla^2 f_T(\mu)]_{i,j}.
\end{equation}
Finally, substituting the established analytical forms for the Hessian elements $-\frac{\partial^2}{\partial\mu_i\partial\mu_j} f_T(\mu)$ derived previously in~\eqref{eq:1st-hessian-exp} and~\eqref{eq:2nd-hessian-exp} directly yields the integral and quantum channel expressions:
\begin{align}
    \left[I(\mu)\right]_{i,j} &= \frac{1}{T^{2}}\int_{0}^{1}ds\,\Tr\!\left[X_{T}(\mu,s)Q_{i}X_{T}(\mu,1-s)Q_{j}\right] \\
    &= \frac{1}{T^{2}}\operatorname{Re}\!\left\{\Tr\!\left[X_{T}(\mu)\Phi_{\mu}(Q_{i})\left(X_{T}(\mu)+I\right)Q_{j}\right]\right\},
\end{align}
completing the proof.
\end{proof}

\end{document}